\newcommand{\CLASSINPUTinnersidemargin}{0.85in}
\newcommand{\CLASSINPUTbottomtextmargin}{2.45cm}

\documentclass[journal,12pt,draftcls,onecolumn,a4paper,ARXIV]{IEEEtran} 
\usepackage{cite} 
\usepackage{ctable}
\usepackage{graphicx}
\usepackage{graphics,color}
\usepackage{units}
\usepackage{array}
\usepackage{amsmath,multirow} 
\usepackage{amssymb,amsthm} 
\usepackage{caption}
\usepackage[justification=centering,font=footnotesize]{subfig}
\graphicspath{{figures/}} 
\usepackage[subnum]{cases}
\usepackage{dsfont}

\newcommand{\definedas}{\overset{\underset{\Delta}{}}{=}}

\newcommand{\Pul}{P_{\text{UL}}}
\newcommand{\Pc}{P_{\text{C}}}
\newcommand{\hdl}{h_{\text{DL}}}
\newcommand{\hul}{h_{\text{UL}}}
\newcommand{\mul}{m_{\text{UL}}}
\newcommand{\Omegaul}{\Omega_{\text{UL}}}
\newcommand{\Omegadl}{\Omega_{\text{DL}}}
\newcommand{\Pdl}{P_{\text{DL}}}
\newcommand{\pr}{\mathbb{P}}
\newcommand{\dd}{\text{d}}
\newcommand{\E}{\mathbb{E}}
\newcommand{\e}{{\rm e}}
\newcommand{\J}{{\rm j}}

\newcommand{\V}[1]{\boldsymbol{#1}} 
\newcommand{\quotes}[1]{``#1"}

\def\Biggg#1{{\hbox{$\left#1\vbox to23.5pt{}\right.$}}} %

\theoremstyle{plain}
\newtheorem{corollary}{Corollary}
\newtheorem{theorem}{Theorem}

\newtheorem{proposition}{Proposition}

\theoremstyle{definition}

\theoremstyle{remark}
\newtheorem{remark}{Remark}

\begin{document}

\title{Performance Analysis of Near-Optimal Energy Buffer Aided Wireless Powered Communication}
\author{\IEEEauthorblockN{Rania Morsi, Diomidis S. Michalopoulos, and Robert Schober}
\thanks{This paper was presented in part at the IEEE Asilomar Conference on Signals, Systems and Computers, Pacific Grove, CA 2014 \cite{Morsi_Asilomar2014} and at the IEEE International Conference on Communications, London, UK 2015 \cite{Morsi_ICC2015}. }
\thanks{Rania Morsi and Robert Schober are with the Institute of Digital Communications, Friedrich-Alexander-University Erlangen-N\"urnberg (FAU), Germany (email: rania.morsi@fau.de, robert.schober@fau.de). Diomidis S. Michalopoulos is with Nokia Bell Labs, Munich, Germany (email: diomidis.michalopoulos@nokia.com).}}

\maketitle
\vspace{-1.7cm}
\begin{abstract}\vspace{-0.2cm}
In this paper, we consider a wireless powered communication system, where an energy harvesting (EH) node harvests energy from a radio frequency (RF) signal broadcasted by an access point (AP) in the downlink (DL). The node stores the harvested energy in an energy buffer and uses the stored energy to transmit data to the AP in the uplink (UL). We investigate two simple online transmission policies for the EH node, namely a best-effort policy and an on-off policy, which do not require  knowledge of the EH profile nor of the UL channel state information. In particular, for both policies, the EH node transmits in each time slot with a constant desired power if sufficient energy is available in its energy buffer. Otherwise, the node transmits with the maximum possible power in the best-effort policy and remains silent in the on-off policy in order to preserve its energy for future use. For both policies, we use the theory of discrete-time \emph{continuous-state} Markov chains to analyze the limiting distribution of the stored energy for finite- and infinite-size energy buffers. We provide this limiting distribution in closed form for a Nakagami-$m$ fading DL channel, i.e., for a Gamma distributed EH process and analyze the outage probability for a Nakagami-$m$ fading UL channel. The analytical results derived in this paper are not limited to EH via RF WPT but are applicable for any independent and identically distributed EH process, originating from e.g. solar and wind energy. Our results reveal that, for low-to-medium outage probabilities, the best-effort policy is superior to the on-off policy and the optimal constant UL transmit power of the EH node that minimizes the outage probability is always less than the average harvested power but increases with the capacity of the energy buffer. The opposite behaviour is observed for high outage probabilities, where turning off the transmission in case of insufficient stored energy results in an improved outage performance compared to always transmitting with best effort. Furthermore, we show that despite the low-complexity of the proposed online policies, their minimum outage probability is near-optimal and closely approaches the outage probability of the optimal offline power allocation policy. \end{abstract} %
\vspace{-0.7cm}
\section{Introduction}
The performance of battery-powered wireless communication networks, such as sensor networks, is limited by the lifetime of the network nodes. Periodic replacement of the nodes' batteries is costly and sometimes impossible when the sensor nodes are placed in a hazardous environment or embedded inside the human body. The lifetime bottleneck problem of energy-constrained wireless networks can be alleviated by allowing the nodes to harvest energy from renewable energy sources to facilitate the wireless information transfer (WIT) to their  designated receivers. However, opportunistic energy harvesting (EH) from conventional renewable energy sources such as solar and wind energy is in general intermittent,  uncontrollable, weather dependent, and not available indoors. In contrast, radio frequency (RF)-based wireless power transfer (WPT) is partially controllable and can be provided on demand to charge low-power devices, such as sensors \cite{Kansal_2007,Krikidis2014,Schober_SWPT_review_2015,WPC_survey_Zhang_2016}. In the following, we provide a brief literature review on  wireless powered communication (WPC) systems, where EH nodes are solely powered by RF energy to facilitate their WIT \cite{WPC_survey_Zhang_2016}. 
 \subsection{Literature Review}
A common feature of EH communication networks is the randomness of both the amount of harvested energy and the fading of the information channel. Therefore, one main objective of energy management polices for EH networks is to match the energy consumption profile of the EH node to the random energy generation profile of the energy source and to the random information channel \cite{Kansal_2007,Krikidis2014,Schober_SWPT_review_2015,WPC_survey_Zhang_2016,WPC_TDMA,Throughput_max_energy_saving_2014,Optimal_online_policies_Mitran_Journal_2014,Resource_allocation_WPC_storage_2016,outage_minimization_Rui_Zhang_2014}. 
For example, in \cite{WPC_TDMA}, a multiuser WPC system is considered, where the users' sum rate is maximized. In this system, the EH nodes perform short-term energy storage, where the harvested energy is \emph{fully} consumed for WIT on a slot by slot basis without buffering energy for use in future time slots\footnote{This transmission model is known in the literature as \quotes{harvest-then-transmit} \cite{WPC_TDMA} or \quotes{harvest-use} policy \cite{Optimal_HUS_2015}. We refer to this model as \quotes{buffer-less} to distinguish it from transmission policies that buffer energy for use in future time slots.}. However,  the performance can be improved by storing the harvested energy in an energy buffer and  optimizing the transmit power of the EH node  based on the stored energy and the quality of the information link  \cite{Throughput_max_energy_saving_2014,Optimal_online_policies_Mitran_Journal_2014,Resource_allocation_WPC_storage_2016,outage_minimization_Rui_Zhang_2014,
DTS_WPC_Energy_accumulation_2015,Optimal_HUS_2015}. For example, in \cite{Throughput_max_energy_saving_2014}, the authors extended the work in \cite{WPC_TDMA} by showing that buffering energy can improve the system throughput. In \cite{Optimal_online_policies_Mitran_Journal_2014}, a continuous-time multiple-access EH communication system is studied for finite and infinite energy storage and a Poisson distributed energy arrival process. The transmit power is optimized online as a function of the remaining battery energy to maximize the sum  throughput. 
In \cite{Resource_allocation_WPC_storage_2016}, time and energy resources are jointly optimized on a slot by slot basis for a multi-user WPC system with finite and infinite energy storage. In \cite{outage_minimization_Rui_Zhang_2014}, the outage probability is minimized for an EH communication  system with infinite-size energy buffers assuming causal/non-causal knowledge of the EH profile. In \cite{DTS_WPC_Energy_accumulation_2015}, the average throughput of an RF EH node with a finite-size energy buffer is analyzed using a discrete-state Markov chain. The EH node decides based on the channel state information (CSI) whether to transmit information  or to harvest energy.\vspace{-0.2cm}
\subsection{Motivation and Contribution}
Optimal offline transmission policies typically require non-causal knowledge of the energy and channel state at the EH node, whereas optimal online solutions are typically based on dynamic programming which entails a high computational complexity for a continuous energy state space and a long transmission horizon, see  \cite{Throughput_max_energy_saving_2014,outage_minimization_Rui_Zhang_2014}, and the references therein. Therefore, these optimal policies may not be feasible in practice. For example, typical EH wireless sensor networks are expected to comprise many small, inexpensive sensors with limited computational power and energy storage \cite{Krikidis2014,Schober_SWPT_review_2015,WPC_survey_Zhang_2016}. In such networks, even causal CSI may not be available at the EH nodes nor at the energy source. To the best of our knowledge, the performance of practical transmission policies that employ energy storage and do not require knowledge of the CSI nor the EH profile at the EH node and the energy source has not been studied in the literature, so far.

Motivated by these practical considerations, in this paper, we consider a WPC system, in which the EH node and the energy source do not have knowledge of the CSI nor of the EH profile. In particular, we assume that an access point (AP) transmits an RF energy signal with a constant power in the downlink (DL) and an EH node harvests the received RF energy, stores it in its energy buffer, and uses the stored energy to transmit data to the AP in the uplink (UL). We study two simple online transmission policies for the EH node, namely a best-effort and an on-off policy. Due to the lack of knowledge of the UL CSI and the EH profile at the EH node, we assume that for both transmission policies the EH node  transmits with constant power $M$ if sufficient energy is available in its energy buffer. If the stored energy is not sufficient to transmit with power $M$, then the node is said to be in the low-energy mode of operation. In this mode, the EH node transmits with the maximum possible power in the best-effort policy and remains silent in the on-off policy.  Hence, regarding the manner in which they deal with the case when the desired power is not available in the energy buffer, the best-effort and the on-off policies can be considered as the extremes among all policies that aim to transmit with a constant power.

The main contributions of this paper can be summarized as follows:
\begin{itemize}
\item We study the performance of the best-effort and the on-off transmission policies for a WPC system assuming no knowledge of the CSI and the EH profile at the EH node and the AP. Furthermore, our analytical results are not only valid for energy harvested via RF WPT, but also for any independent and identically distributed (i.i.d.) EH process, originating from e.g. solar, wind, thermal, and vibrational energy.
\item For both transmission policies, we model the stored energy by a discrete-time \emph{continuous-state} Markov chain and provide an integral equation for its limiting distribution for both infinite- and finite-size energy buffers and for \emph{any} i.i.d.  EH process. Our analysis takes into account the inefficiency of the power amplifier at the EH node, the constant power consumed by the EH node circuit during transmission, and the storage inefficiency of the energy buffer.
\item Unlike our preliminary work in \cite{Morsi_Asilomar2014} and  \cite{Morsi_ICC2015}, which considered Rayleigh fading DL and UL channels and obtained only approximate expressions for the outage probability of the UL channel, in this paper, we provide for both transmission policies the limiting energy distribution in closed form  for an i.i.d. Gamma distributed EH process, i.e., for a Nakagami-$m$ DL block fading channel. Furthermore, we provide \emph{exact} analytical expressions for the outage probability for a Nakagami-$m$ faded UL channel.  
Based on the derived analytical expressions for the outage probability, we can find the optimal constant transmit power $M$ that minimizes the outage probability with a simple one-dimensional search.
\item We show that, interestingly, the minimum outage probability of the two proposed online policies is near-optimal and  closely approaches the outage probability of the optimal offline power allocation policy in \cite{outage_minimization_Rui_Zhang_2014}. The optimal offline policy in  \cite{outage_minimization_Rui_Zhang_2014} requires non-causal knowledge of the EH profile. The optimal online algorithm in \cite{outage_minimization_Rui_Zhang_2014} requires causal EH profile knowledge and is based on dynamic programming which entails a high computational complexity for a continuous energy state space and a long transmission horizon. On the other hand, in each time slot, the proposed simple online transmission policies only require knowledge of whether the desired constant transmit power is available in the energy buffer. For given statistical properties of the channel, the proposed policies need to perform only a \emph{single} one-dimensional search to determine the optimal constant transmit power $M$.
\item  Our results reveal that, for low-to-medium outage probabilities,  the best-effort policy outperforms the on-off policy and the optimal UL transmit power $M$ that minimizes the outage probability is always less than the average harvested power but increases with the capacity of the energy buffer. On the contrary, for high outage probabilities, the on-off policy has a superior outage performance compared to the best-effort policy and the optimal UL transmit power is larger than the average harvested power but decreases with the capacity of the energy buffer.
\end{itemize}
The remainder of the paper is organized as follows. Section \ref{s:System_model} presents the overall system model. In Sections \ref{s:Infinite_buffer} and \ref{s:Finite_buffer}, we study the limiting distributions of the stored energy for infinite- and finite-capacity energy buffers, respectively. In Section \ref{s:Outage_analysis}, we analyze the outage probability of the communication link, when both UL and DL channels are Nakagami-$m$ faded. Numerical and simulation results are provided in Section \ref{s:Simulations}. Finally, Section \ref{s:conclusion} concludes the paper. 
\vspace{0.2cm}\emph{Notation:} We use the following notations and functions throughout this paper. $\E[\cdot]$ denotes expectation. $\{X\}$ refers to a random sequence $X$. $\mathbb{P}(\cdot)$ denotes the probability of an event.  $\mathds{1}_A=1$ if event $A$ is true and $\mathds{1}_A=0$ otherwise. $[x]^+=\max(x,0)$.  $W_0(z)$ is the principle branch of the Lambert W function which is the inverse function of $z=w\e^{w}$. $\J=\sqrt{-1}$ is the imaginary unit. $x^*$, $\Re\{x\}$, and $\Im\{x\}$ denote respectively the complex conjugate, the real part, and the imaginary part of $x$. $\V{1}_m$ is the all-ones column vector of length $m$. $\V{x}^T$ denotes the transpose of vector $\V{x}$. For a positive integer $m$, the Gamma function is defined as $\Gamma(m)=(m-1)!$, the upper incomplete Gamma function is given by $\Gamma(m,x)=(m-1)!\e^{-x}\sum_{n=0}^{m-1}\frac{x^n}{n!}$, and the lower incomplete Gamma function is given by $\gamma(m,x)=\Gamma(m)-\Gamma(m,x)=(m-1)!\left(1-\e^{-x}\sum_{n=0}^{m-1}\frac{x^n}{n!}\right)$. $\definedas$ stands for ``is defined as". Finally, 
$\mathbb{Z}^+$ denotes the set of positive integers.
\vspace{-0.1cm}
\section{System Model}
\label{s:System_model}
We consider a time-slotted point-to-point WPC system with an AP and an EH node, cf. Fig. \ref{fig:system_model}. The AP is connected to a fixed power supply, whereas the EH node does not have a fixed power source and needs to harvest energy to be able to communicate with the AP. In particular, the EH node captures the RF energy transferred by the AP in the DL, stores it in its energy buffer, and uses the stored energy to transmit its backlogged data to the AP in the UL. The considered DL WPT and UL WIT system employs frequency-division duplex (FDD), where WPT and WIT take place concurrently in two different dedicated frequency bands, i.e., the EH node operates in a half-duplex out-of-band transmission mode. The AP is assumed to have no knowledge of the DL CSI and the EH node is assumed to have no knowledge of the UL CSI nor of the amount of harvested energy. It is only assumed that the EH node knows whether the desired constant transmit power is available in the energy buffer. In the following, we describe the communication, channel, EH, and energy storage models as well as the considered system imperfections.
\begin{figure}[!t] \vspace{-0.3cm}
\centering
\begin{minipage}[b]{0.47\textwidth}
\centering
\includegraphics[width=1.\textwidth, trim= {0cm 0 0cm 0cm},clip]{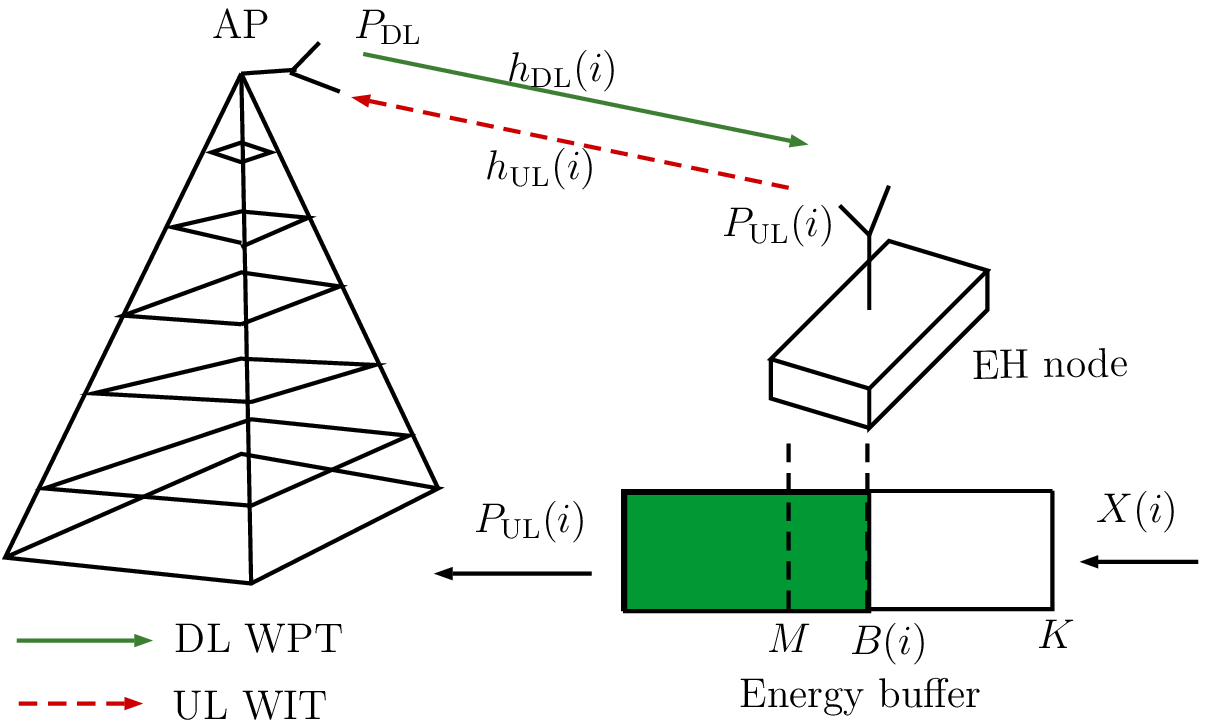}
\caption{A WPC system with DL WPT and UL WIT with energy storage at the EH node.}
\label{fig:system_model}
\end{minipage}
\hspace{0.4cm}
\begin{minipage}[b]{.47\textwidth}
\centering\vspace{-0.7cm}
\includegraphics[width=1\textwidth, trim= {0cm 0 0cm 0.7cm},clip]{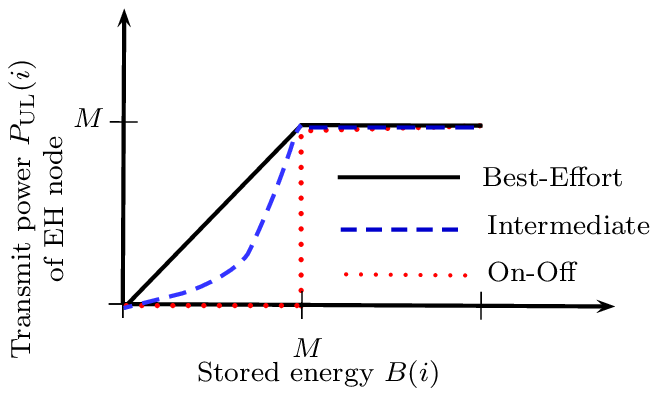}
 \caption{The transmit power of the EH node versus the energy stored in its buffer for the best-effort policy, the on-off policy, and an intermediate transmission policy.}
\label{fig:OO_BE_motivation}
\end{minipage}\vspace{-0.5cm}
\end{figure}
\subsection{Communication Model}
\clearpage
We consider a time-slotted system, with a unit-length time slot. Hence, we use the terms energy and power interchangeably. In time slot $i$, defined as the time interval $[i,i+1)$, the EH node transmits data to the AP with power $ \Pul(i)$. Since the EH node does not know the UL CSI, we are interested in studying transmission policies for which the EH node transmits with a constant power, denoted by      $M$, whenever sufficient energy is available in its energy buffer, cf. Fig. \ref{fig:OO_BE_motivation}. If the amount of energy stored is less than $M$, the transmit power of the EH node is a function of the stored energy. In this paper, we study two extreme cases among all transmission policies that aim to transmit with a constant power, namely the best-effort and the on-off policies illustrated in Fig. \ref{fig:OO_BE_motivation}. In particular, for the best-effort policy, the EH node transmits with a constant  power, $M$, or with the maximum possible power if there is not enough energy  available in its energy buffer. Hence, for the best-effort policy, the EH node transmits with UL power\vspace{-0.2cm}
\begin{equation}
\Pul(i)\Big|_{\text{best-effort}}=\min(B(i),M),
\label{eq:Pul_policy_BE}
\end{equation}\vspace{-0.1cm}where $B(i)$ is the energy stored in the energy buffer at the beginning of time slot $i$. For the on-off policy, the EH node also transmits with constant power, $M$, if it is available in its energy buffer, but remains silent otherwise. Hence, the UL transmit power for the on-off policy is\vspace{-0.1cm}
\begin{equation}
\Pul(i)\Big|_{\text{on-off}}=M\mathds{1}_{B(i)>M}=\begin{cases}\\[-6ex]0 & B(i)\leq M \\[-1.5ex] M & B(i) > M\end{cases}.
\label{eq:Pul_policy_OO}
\end{equation}\vspace{-0.1cm}
We note that the on-off policy is motivated by the fact that a constant transmit power allows the use of a power-efficient power amplifier at the EH node. Furthermore, in the low-energy mode of operation, i.e., when $B(i)< M$, cf. Fig. \ref{fig:OO_BE_motivation}, allowing the EH node to transmit with a low power, as in the best-effort policy, may result in a high probability of transmission outage in the UL WIT. Hence, saving this energy for transmission with a higher power in a future time slot, as in the on-off policy, may lead to an improved outage performance. We note that, any \quotes{intermediate} policy, whose transmit power vs. stored energy curve lies between those of the best-effort and on-off  policies, cf. Fig. \ref{fig:OO_BE_motivation}, is expected to result in a performance that lies between the performances of these two policies.
\vspace{-0.4cm}
\subsection{Channel Model} 
\label{ss:Channel_model}\vspace{-0.2cm}
Both the UL and DL channels are assumed to be flat block fading, i.e., the channels remain constant during one time slot and change \emph{independently} from one slot to the next. Hence, the UL and DL channel power gain sequences $\{\hul(i)\}$ and $\{\hdl(i)\}$ are mutually independent i.i.d. random processes. We assume that $\{\hul(i)\}$ and $\{\hdl(i)\}$ capture the joint effect of the large scale path loss
and the small-scale multipath fading. Furthermore, both UL and DL channels are assumed to be stationary and ergodic processes with means $\Omegaul=\E[\hul(i)]$ and $\Omegadl=\E[\hdl(i)]$, respectively. In this paper, we analyze the WPC system for Nakagami-$m$ fading UL and DL channels. The Nakagami-$m$ fading model is adopted since it is known to provide a good fit to outdoor and indoor multipath propagation and it is a general fading model that reduces to Rayleigh fading for $m=1$ and can approximate Ricean fading \cite{Statistical_channel_model}, \cite[Eq. (2.26)]{Digital_comm_fading_Alouini2005}. We assume that the DL CSI and the UL CSI are unknown to the AP and the EH node, respectively, and that additive white Gaussian noise of power $\sigma^2$ impairs the received signal at the AP.\vspace{-0.2cm}
\subsection{EH Model}
\label{ss:EH_model}
In time slot $i$, the EH node collects $X(i)$ units of RF energy  and stores it in its energy buffer. We assume that the RF energy broadcasted by the AP is the only source of energy for the EH node. We further assume that the energy replenished in a time slot may only be used in future time slots. We adopt the EH receiver model in \cite{WIPT_Architecture_Rui_Zhang_2012}, where the harvested energy in time slot $i$ is given by $X(i)=\eta\Pdl\hdl(i)$, where $0<\eta< 1$ is the RF-to-direct current (DC) conversion efficiency of the EH module and $\Pdl$ is the constant DL transmit power of the AP. Since the DL channel $\{\hdl(i)\}$ is i.i.d., the energy replenishment sequence $\{X(i)\}$ is also an i.i.d. stationary and ergodic process with mean $\bar{X}=\eta\Pdl\Omegadl$, probability density function (pdf) $f(x)$, cumulative distribution function (cdf) $F(x)=\mathbb{P}(X(i)<x)$ and complementary cumulative distribution function (ccdf) $\bar{F}(x)=1-F(x)$.

\begin{remark}We note that although we focus on energy harvested via RF WPT, all analytical results in this paper are applicable to any i.i.d. EH process $\{X(i)\}$ with pdf $f(x)$ and cdf $F(x)$.
\label{rem:validity_any_iid_EH_model}
\end{remark}
\vspace{-0.5cm}
\subsection{Energy Storage Model}
\label{ss:storage_model}
The harvested energy $X(i)$ is stored in an energy buffer, such as a rechargeable battery and/or a supercapacitor \cite{Culler_2005},\cite{supercapacitor_on_chip}, with storage capacity $K$. The dynamics of the storage process $\{B(i)\}$ are given by the storage equation\vspace{-0.3cm}
\begin{equation}
B(i+1)=\min\left(B(i)-\Pul(i)+X(i),K\right)
\label{eq:general_storage_equation}
\end{equation}
which reduces to\vspace{-0.3cm}
\begin{equation}
B(i+1)\Big|_{\text{best-effort}}=\min\left([B(i)-M]^++X(i),K\right)
\label{eq:general_storage_equation_BE}
\end{equation}
for the best-effort policy and to\vspace{-0.3cm}
\begin{equation}
B(i+1)\Big|_{\text{on-off}}=\min\left(B(i)-M\mathds{1}_{B(i)> M}+X(i),K\right) 
\label{eq:general_storage_equation_OO}
\end{equation}
for the on-off policy. The storage processes  in (\ref{eq:general_storage_equation_BE}) and (\ref{eq:general_storage_equation_OO}) are discrete-time Markov chains on a \emph{continuous} state space $S$, where  $S=[0,K]$  and $S=[0,\infty)$ for a finite- and an infinite-size energy buffer, respectively.
\begin{remark}
Interestingly, our storage model for the best-effort policy in (\ref{eq:general_storage_equation_BE}) is similar (but not identical) to the water dam model proposed by Moran in 1956 in \cite{Moran_1956}. In Moran's model, every year $X(i)$ units of water flow into a dam of capacity $K$ and a constant amount of water $M$ is released just before the following year. Moran studies the amount of water $\{Z(i)\}$ stored in the dam just after release, which is modeled by the storage equation\footnote{Our storage equation for $B(i)$ given in (4) is different from Moran's storage equation for $Z(i)$ because we study the stored quantity after the harvested energy $X(i)$ has been added to the energy buffer, whereas Moran studies the stored quantity after the water has been released from the water buffer.} $Z(i+1)= [\min(Z(i)+X(i),K)-M]^+$, which for an infinite-capacity dam reduces to \small $Z(i+1)=\begin{cases} 0 & Z(i)+X(i) \leq M \\ Z(i)+X(i)-M & Z(i)+X(i)> M \end{cases}$. \normalsize
From \cite{Infinite_dam_Gani_Prabhu_1957}, the stationary distribution\footnote{If a Markov chain is characterized by a distribution $\pi$ in a certain time slot and  this distribution is unchanged for all future time slots, then $\pi$ is said to be a stationary distribution of the Markov chain \cite{Meyn_Tweedie}.} of $\{Z(i)\}$ (if it exists) is obtained from that of another process $\{U(i)\}$ defined as $U(i)=Z(i)+X(i)$, which is in fact identical in distribution to our process $\{B(i)\}$. This is because, if we add $X(i+1)$ to $Z(i+1)$, we get \small $U(i+1)=\begin{cases} X(i+1) & U(i) \leq M \\ U(i)-M+X(i+1) & U(i)> M \end{cases},$ \normalsize
which is identical in distribution to $\{B(i)\}$ in (\ref{eq:general_storage_equation_BE}) for $K\to\infty$. Hence, the distribution of $\{U(i)\}$ in Moran's dam model is identical to the distribution of $\{B(i)\}$ in our energy buffer model. Similarly, for a finite storage capacity, $\{B(i)\}$ is equivalent to $\{\min(U(i),K)\}$.
\label{remark:equivalence_to_Morans_Model}
\end{remark}
\begin{remark}
The storage model for the on-off policy in (\ref{eq:general_storage_equation_OO}) is similar to the double service rate process  defined by Gaver and Miller in 1962 in \cite[Section 3]{Gaver_Miller_1962}. In their process, one of two service rates is used to serve customers in a queue depending on the queue length (or equivalently the total waiting time). If the waiting time exceeds a prescribed value, say $M$, customers are served with a fast service rate $r_2>r_1$, otherwise, customers are served with a low service rate $r_1$. As explained in \cite{Gaver_Miller_1962}, this model may also describe a dam storage system, where water is released at an increased rate if the dam content exceeds a prescribed value $M$. Here, we use the model to describe an energy storage system with an energy release rate of $r_2=M$ if the energy stored is greater than $M$, and an energy release rate of $r_1=0$, otherwise.
However, we note a minor difference between this model and our on-off model. In Gaver's and Miller's model, the release rate may change within one time slot, namely if the waiting time drops below $M$ in the middle of the time slot, the rate switches from $r_2$ to $r_1$ within the same slot. In our model, however, if the energy stored at the beginning of the time slot $B(i)$ is larger than $M$, the transmit power remains $M$ throughout the whole slot. 
\label{remark:equivalence_onoff_to_Millers_model}
\end{remark}
\begin{remark}
Having related our energy storage models to dam and queuing models, we note that physically distinct Markov models may lead to the same \emph{update} equations, such as those in (\ref{eq:general_storage_equation_BE}) and (\ref{eq:general_storage_equation_OO}) and therefore to the same limiting distributions\footnote{A limiting distribution of a Markov chain is a stationary distribution that the chain asymptotically converges to starting from any initial distribution \cite{Meyn_Tweedie}.}. More specifically, models with equal-length time slots (such as our energy storage model and Moran's dam model) and models with random-length time slots (such as queuing models) may be mathematically equivalent. Moreover, a model where the input/output flow occurs instantaneously may be mathematically identical to a model whose input/output flow has a steady rate during the whole time slot. In other words, two models having the same update equations are mathematically equivalent regardless of how or when the inputs and outputs occur. A detailed discussion on these equivalences is given by Gani in \cite[pp.199, 200]{review_paper_Gani_1957}.
\label{remark:equivalent_models}
\end{remark}
\begin{figure}[!t] 
\vspace{-0.85cm}
 \subfloat[Best-effort policy]{\label{subfig:BE_samplepath}\includegraphics[width=0.495\textwidth]{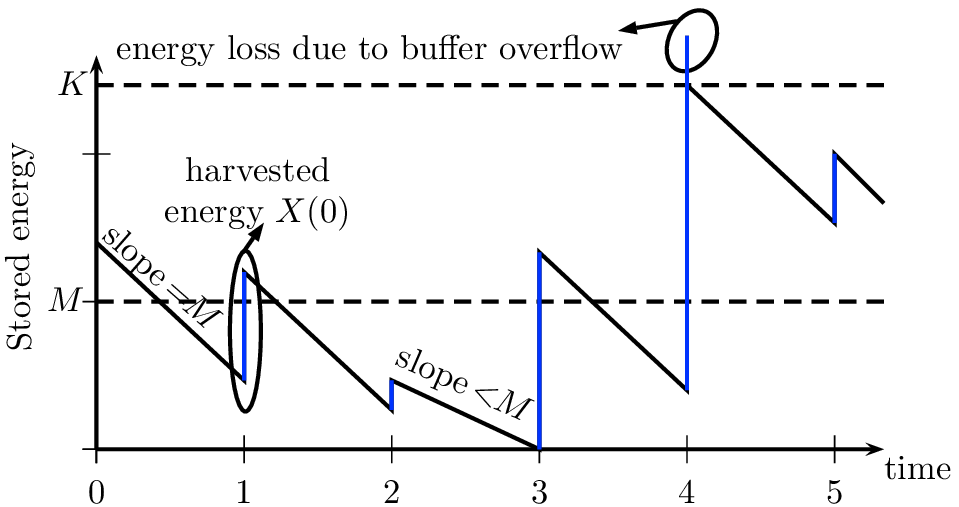}}
 \subfloat[On-off policy]{\label{subfig:OO_samplepath}\includegraphics[width=0.495\textwidth]{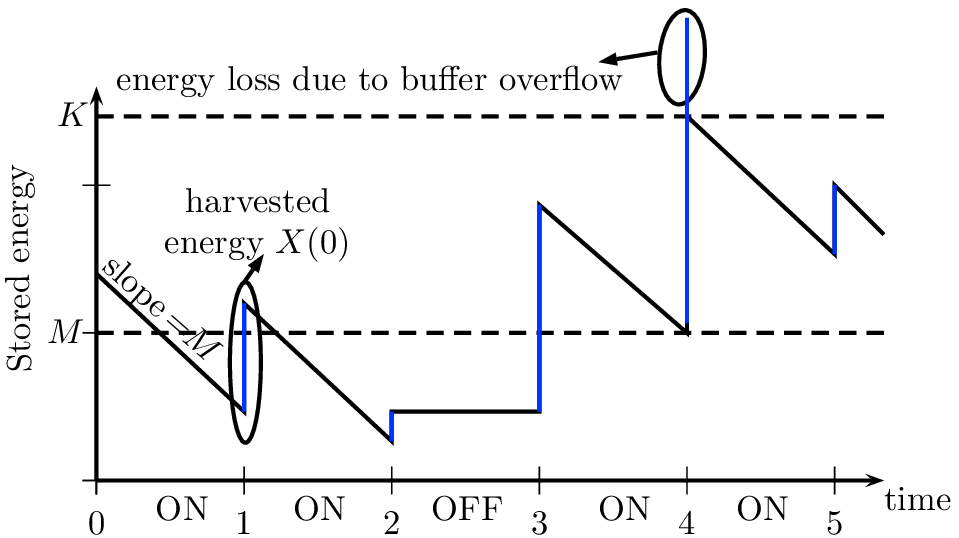}}
\caption{Sample paths for the best-effort and on-off policies for a finite-size energy buffer for the same given harvested energy sample sequence.}
\label{fig:BE_OO_samplepath}\vspace{-0.8cm}
\end{figure}

In Fig. \ref{fig:BE_OO_samplepath}, we show possible sample paths for the best-effort and the on-off policies for a finite-size buffer. Since the harvested energy $X(i)$ may only be used in future time slots, cf. Section \ref{ss:EH_model}, we add it --for illustration purposes-- instantaneously at the end of the time slot (shown as a jump in the sample path) instead of adding it with a steady rate throughout the time slot. This, however, has no effect on the storage model, cf. Remark \ref{remark:equivalent_models}. Fig. \ref{fig:BE_OO_samplepath} shows that if the stored energy at the beginning of a time slot $B(i)$ is larger than $M$, then  for both policies, the EH node transmits with power $M$, i.e., the energy is drained from the buffer at a rate of $M$. Otherwise, the energy is drained at a rate of  $B(i)<M$ for the best-effort policy or zero for the on-off policy. Furthermore, whenever the harvested energy causes the buffer to overflow, the excess energy is lost. Since the on-off policy saves energy in the low-energy mode of operation, an energy buffer overflow is more likely to occur compared to the best-effort policy for  finite-size energy buffers.
\vspace{-0.3cm}
\subsection{Consideration of Imperfections}
\label{s:imperfections}
We consider non-idealities of the power amplifier and the energy buffer. In particular, we take into account the following imperfections: (a) To produce an RF power of $\Pul$, the power amplifier of the EH node consumes a total power of $\rho\Pul$, where $\rho> 1$ is the power amplifier inefficiency. (b) The EH node circuit consumes a constant power of $\Pc$ during transmission. (c) The energy buffer is characterized by a storage efficiency $0<\beta<1$, i.e., if $X$ amount of energy is applied at the input of the buffer, only an amount of $\beta X$ is stored. 
Given that the desired UL transmit power is $\Pul(i)=M$, then if $B(i)<(\Pc+\rho M)$, the UL transmit power for the best-effort policy is reduced to satisfy $B(i)=\Pc+\rho\Pul(i)$, i.e., $\Pul(i)=[B(i)-\Pc]^+/\rho$, and for the on-off policy, the EH node remains silent. Hence, the UL transmit powers of the EH node for the best-effort and the on-off policies reduce respectively to 
\begin{equation}
\Pul(i)\Big|_{\text{best-effort}}\!=\begin{cases} \frac{[B(i)-\Pc]^+}{\rho} & B(i)\leq \tilde{M} \\ M & B(i) > \tilde{M}\end{cases},\quad{\rm  and } \quad \Pul(i)\Big|_{\text{on-off}}\!=\begin{cases} 0 & B(i)\leq \tilde{M} \\ M & B(i) > \tilde{M}\end{cases},
\label{eq:Puwith_imperfections}
\end{equation} where $\tilde{M}\!=\!\Pc+\rho M$.  Hence, the corresponding storage equations for the best-effort and the on-off policies reduce respectively to $
B(i+1)\Big|_{\text{best-effort}}=\min\left([B(i)\,-\,(\Pc+\rho M)]^++\beta X(i),K\right)$ and $B(i+1)\Big|_{\text{on-off}}=\min\left(B(i)\!-\!(\Pc+\rho M)\mathds{1}_{B(i)>(\Pc+\rho M)}+\beta X(i),K\right)$, which are identical to (\ref{eq:general_storage_equation_BE}) and  (\ref{eq:general_storage_equation_OO}), respectively,  after replacing $M$ by $\tilde{M}\!=\!\Pc+\rho M$ and $f(x)$ by $\tilde{f}(x)\!=\!\frac{1}{\beta}f\left(\frac{x}{\beta}\right)$. Thus, the average harvested energy considering storage inefficiency reduces to  $\widetilde{\bar{X}}=\beta\eta\Pdl \Omegadl$, i.e., $\beta$ has the effect of scaling the RF-to-DC conversion efficiency. In the following, we perform the analysis for an ideal system (i.e., $\Pc=0$, $\rho\!=\!1$ and $\beta\!=\!1$). For a non-ideal system, the results derived in Sections \ref{s:Infinite_buffer} and \ref{s:Finite_buffer} hold with the aforementioned substitutions.
\begin{remark}
We note that if a transmission policy is based on the amount of the stored energy, such as in (\ref{eq:Pul_policy_BE}) and (\ref{eq:Pul_policy_OO}), then the limiting distribution of the stored energy is essential for deriving almost any performance metric of the system. To this end, for a given transmission policy, the limiting energy distribution can be obtained in two steps. First, an integral equation for the limiting distribution can be formulated for a general i.i.d. EH process. Second, given the statistics of the EH process, the limiting energy distribution can be obtained by solving these integral equations. In this manner, we obtain in Sections \ref{s:Infinite_buffer} and \ref{s:Finite_buffer} the limiting energy distributions for the best-effort and the on-off transmission policies for infinite and finite-size energy buffers, respectively, which are then used in Section \ref{s:Outage_analysis} to analyze the outage performance of the communication link.
\end{remark}
\section{Infinite-Capacity Energy Buffer}
\label{s:Infinite_buffer}
Studying an infinite-size energy buffer is of interest since the average  harvested energy can be several orders of magnitude lower than the capacity of commonly-used energy buffers, e.g. rechargeable batteries.
\vspace{-0.2cm}\subsection{Existence of a Limiting Energy Distribution}
\begin{theorem}\normalfont
For the storage processes $\{B(i)\}$ in (\ref{eq:general_storage_equation_BE}) and (\ref{eq:general_storage_equation_OO}) with an infinite-size energy buffer, if the maximum output energy is less than the average harvested input energy, i.e., if $M<\bar{X}$, then $\{B(i)\}$ does not possess a stationary distribution. Furthermore, after a finite number of time slots, $\Pul(i)=M$ holds almost surely (a.s.).
\label{theo:no_stationary_dist}
\end{theorem}\vspace{-0.4cm}
\ifARXIV
\begin{proof} The proof is provided  in \cite[Appendix A]{Morsi_ICC2015} and in Appendix \ref{app:no_stationary_dist}. \end{proof}
\else
\begin{proof} The proof is provided in \cite[Appendix A]{Morsi_ICC2015} and \cite[Appendix A]{Morsi_storage_arxiv}\footnote{Due to space limitations, proofs already given in \cite{Morsi_Asilomar2014} and  \cite{Morsi_ICC2015} are not repeated in this manuscript. In addition, we refer the reader to \cite{Morsi_storage_arxiv}, which is identical to the current manuscript but includes all proofs from \cite{Morsi_Asilomar2014} and  \cite{Morsi_ICC2015} as well as detailed versions of new proofs.}. 
\end{proof}
\fi\vspace{-0.3cm}
\begin{theorem}\normalfont
For the storage processes $\{B(i)\}$ in (\ref{eq:general_storage_equation_BE}) and (\ref{eq:general_storage_equation_OO}) with an infinite-size energy buffer, if the maximum output energy is larger than the average harvested input energy, i.e., if $M>\bar{X}$, then $\{B(i)\}$ is a stationary and ergodic random process which possesses a unique stationary distribution $\pi$ that is absolutely continuous on $(0,\infty)$. Furthermore, the process converges in total variation to the limiting distribution $\pi$ from any initial distribution. In addition, $\E[\Pul(i)]=\bar{X}$ holds.
\label{theo:stationary_dist_infinite}
\end{theorem}\vspace{-0.2cm}
\ifARXIV
\begin{proof} The proof is provided partially in \cite[Appendix B]{Morsi_ICC2015} and fully in Appendix \ref{app:stationary_dist_infinite_BE}.
\end{proof}
\else
\begin{proof} The proof is provided  partially in \cite[Appendix B]{Morsi_ICC2015} and fully in \cite[Appendix B]{Morsi_storage_arxiv}. \end{proof}
\fi \vspace{-0.4cm}
\vspace{-0.1cm}\subsection{Best-Effort Policy with an Infinite-Size Energy Buffer}
\label{ss:BE_scheme_INF}
\begin{theorem}\normalfont
Consider the storage process $\{B(i)\}$ of the best-effort policy in (\ref{eq:general_storage_equation_BE}) with an infinite-size energy buffer and $M>\bar{X}$. Let $g(x)$ on $(0,\infty)$ be the limiting pdf of the stored energy. If $f(x)$ and $F(x)$ are respectively the pdf and the cdf of the EH process $\{X(i)\}$, then $g(x)$ must satisfy the following integral equation\vspace{-0.3cm}
\begin{equation}
g(x)=f(x)\int\limits_0^M g(u) \dd u + \int\limits_M^{M+x} f(x-u+M) g(u) \dd u,
\label{eq:Integral_eqn_infinite_BE}
\end{equation}
and the limiting cdf $G(x)$ of the stored energy satisfies the following Lindley integral equation
\begin{equation}
G(x)=\int\limits_{u=0}^x G(x-u+M)\dd F(u).
\label{eq:CDF_Integral_eqn_infinite_BE}
\end{equation}
\label{theo:integral_eq_infinite_BE}
\end{theorem}
\vspace{-1cm}
\begin{proof}
\ifARXIV
The proof is provided partially in \cite[Proof of Theorem 3]{Morsi_ICC2015} and fully in Appendix \ref{app:proof_BE_integral_eqn_infinite}.
\else
The proof is provided partially in \cite[Proof of Theorem 3]{Morsi_ICC2015} and fully in \cite[Appendix C]{Morsi_storage_arxiv}.
\fi
\end{proof} \vspace{-0.2cm}
Because of the relation between our storage model and Moran's model, cf. Remark \ref{remark:equivalence_to_Morans_Model}, (\ref{eq:Integral_eqn_infinite_BE}) is identical to \cite[eq. (5)]{Infinite_dam_Gani_Prabhu_1957} and (\ref{eq:CDF_Integral_eqn_infinite_BE}) is identical to \cite[eq. (7)]{Infinite_dam_Gani_Prabhu_1957}. Next, we consider the case when the DL channel is Nakagami-$m$ block fading, i.e., the EH process is Gamma distributed.
 \vspace{-0.3cm}\begin{corollary}\normalfont 
Consider the storage process $\{B(i)\}$ of the best-effort policy in (\ref{eq:general_storage_equation_BE}) with an infinite-size energy buffer and $M\!>\!\bar{X}$. If the EH process $\{X(i)\}$ is Gamma distributed with an integer shape parameter $m\in\{1,2,\ldots\}$ and pdf $f(x)\!=\!\frac{\lambda^m}{\Gamma(m)}x^{m-1}\e^{-\lambda x}$, where $\lambda\!\!=\!\!m/\bar{X}$, then the stored energy has a limiting cdf given by $G(x)\!\!=\!\!1-\sum_{n=0}^{m-1}c_n\e^{-\lambda_n x}$ and a limiting pdf given by $g(x)\!\!=\!\!\sum_{n=0}^{m-1}\lambda_n c_n\e^{-\lambda_n x}$, where $\lambda_n\!=\lambda+\frac{m}{M} W_0\left(-\delta \e^{-\delta}\e^{-\J\frac{2 \pi n}{m}}\right)$ and $\delta\!=\!M/\bar{X}=\lambda M/m$. The coefficients $c_n$, $n=0,\ldots,m-1$, are obtained by solving the non-homogeneous system of linear equations $\V{A} \V{c}=\V{1}_m$, where $\V{c}=[c_0,\ldots,c_{m-1}]^T$ and $\V{A}$ is an $m\times m$ matrix whose entry in the $s^{\text{th}}$ row and the $n^{\text{th}}$ column is given by $A_{s,n}=\left(\frac{\lambda-\lambda_n}{\lambda}\right)^s$, where $s,n\in\{0,\ldots,m-1\}$.
\label{theo:stationary_dist_infinite_Nakagami}
\end{corollary} \vspace{-0.4cm}
\begin{proof} The proof is provided in Appendix \ref{app:stationary_dist_infinite_Nakagami_BE}. \end{proof} \vspace{-0.3cm}
\vspace{-0.3cm}\subsection{On-Off Policy with an Infinite-Size Energy Buffer}
\label{ss:OO_scheme_INF}
\begin{theorem}\normalfont
Consider the storage process $\{B(i)\}$ of the on-off policy in (\ref{eq:general_storage_equation_OO}) with an infinite-size energy buffer and $M>\bar{X}$. Let $g(x)$ on $(0,\infty)$ be the limiting pdf of the stored energy. If $f(x)$ is the pdf of the EH process $\{X(i)\}$, then $g(x)$ must satisfy the following integral equation\vspace{-0.1cm}
\begin{numcases}{g(x)\!=\!\! \label{eq:g_integral_eqn_infinite_OO}}
\int\limits_{u=0}^x f(x-u) g(u) \dd u + \int\limits_{u=M}^{M+x} f(x-u+M) g(u) \dd u, &  $0\leq x< M$ \label{eq:Integral_eqn_infinite_parta_OO}\\
\int\limits_{u=0}^M f(x-u)g(u) \dd u + \int\limits_{u=M}^{M+x} f(x-u+M) g(u) \dd u,  &  $x\geq M$  \label{eq:Integral_eqn_infinite_partb_OO}\end{numcases}
\label{theo:OO_integral_eqn_infinite}
\end{theorem}\vspace{-0.8cm}
\ifARXIV
\begin{proof} The proof is provided in \cite[Proof of Theorem 3]{Morsi_Asilomar2014} and in Appendix \ref{app:proof_OO_integral_eqn_infinite}. \end{proof}
\else
\begin{proof} The proof is provided in \cite[Proof of Theorem 3]{Morsi_Asilomar2014} and \cite[Appendix E]{Morsi_storage_arxiv}. \end{proof}
\fi
Next, we consider the case when the DL channel is Nakagami-$m$ block fading, i.e., the EH process is Gamma distributed.
\begin{corollary} \normalfont
Consider the storage process $\{B(i)\}$ of the on-off policy in (\ref{eq:general_storage_equation_OO}) with an infinite-size energy buffer and $M\!>\!\bar{X}$. If the EH process $\{X(i)\}$ is Gamma distributed with an integer shape parameter $m\in\{1,2,\ldots\}$ and pdf $f(x)\!=\!\frac{\lambda^m}{\Gamma(m)}x^{m-1}\e^{-\lambda x}$, where $\lambda\!\!=\!\!m/\bar{X}$, then the stored energy has a limiting pdf given by
\begin{numcases}{\hspace{-0.7cm}g(x)\!=\!\! \label{eq:PDF_INF_BE_Nakagami}}
\begin{aligned}
\!\!\!&\sum\limits_{n=0}^{m-1} \lambda_n c_n\Biggg[ \e^{-\lambda_n x}+\frac{\lambda}{m}\sum\limits_{k=0}^{m-1}\Big(a_{nk}\e^{-\theta_k x}+b_{nk}\e^{-\theta_k^* x}-(a_{nk}+b_{nk})\e^{-\lambda_n x}\Big)\\[-2ex]
&\!\!\!-\!\sum\limits_{t=0}^{m-1}\frac{(\lambda\!-\!\lambda_n)^t}{t!}\e^{-\lambda x}\Bigg\{x^t\!+\! \frac{1}{m}\sum\limits_{k=0}^{m-1}\Re\{(\lambda\e^{\J\eta_k})^{-t}\e^{\lambda x\e^{\J\eta_k}}\gamma(t\!+\!1,\lambda x\e^{\J\eta_k})\}\Bigg\}\Biggg],
\end{aligned} & \hspace{-0.5cm}$0\!\leq \!x\!<\!\!M$ \label{eq:g1_INF_OO}\\
\sum\limits_{n=0}^{m-1} \lambda_n c_n\e^{-\lambda_n x},  &  \hspace{0cm}$x\!\geq\! M$  \label{eq:g2_INF_OO}\end{numcases}
where $\eta_k=\frac{2\pi k}{m}$, $\theta_k=\lambda(1-\e^{\J\eta_k})$, $a_{nk}=[2\big(\lambda_n\e^{-\J\eta_k}+\theta_k^*\big)]^{-1}$, $b_{nk}=[2\big(\lambda_n\e^{\J\eta_k}+\theta_k\big)]^{-1}$, $\lambda_n\!=\lambda+\frac{m}{M} W_0\left(-\delta \e^{-\delta}\e^{-\J\frac{2 \pi n}{m}}\right)$, and  $\delta\!=\!M/\bar{X}=\lambda M/m$.  The coefficients $c_n$, $n=0,\ldots,m-1$, are obtained by solving the system of linear equations $\V{A} \V{c}=\V{1}_m$, where $\V{c}=[c_0,\ldots,c_{m-1}]^T$ and $\V{A}$ is an $m\times m$ matrix whose entry in the $s^{\text{th}}$ row and the $n^{\text{th}}$ column, for $s,n\in\{0,\ldots,m-1\}$, is given by $A_{s,n}=D_n+B_{sn}$, where
\ifARXIV \vspace{-1cm} \fi
\begin{equation}
\begin{aligned}
D_n&\!=\!\frac{\lambda_n\lambda}{m}\sum\limits_{k=0}^{m-1}\Big(a_{nk}\zeta_k+b_{nk}\zeta_k^*-\frac{a_{nk}+b_{nk}}{\lambda_n}(1-\e^{-\lambda_n M})\Big)
-\lambda_n\sum\limits_{t=0}^{m-1}(\lambda\!-\!\lambda_n)^t\Biggg[\frac{\lambda^{-t-1}}{t!}\gamma(t+1,\lambda M)\\[-1ex]
&+\frac{1}{m}\sum\limits_{k=0}^{m-1}\Re\Bigg\{(\lambda\e^{\J\eta_k})^{-t}\Big(\zeta_k
-\sum\limits_{q=0}^{t}\frac{\e^{\J\eta_k q}\lambda^{-1}}{q!}\gamma(q+1,\lambda M)\Big)\Bigg\}\Biggg]+1,
\end{aligned}
\label{eq:D_n_OO}
\end{equation}
\begin{equation}
\begin{aligned}
&B_{sn}\!\!=\!\!-\frac{\lambda_n(\lambda\!\!-\!\!\lambda_n)^s}{\lambda^{s+1}}\!\!+\!\!\frac{\lambda^m(-\!1)^{m\!-\!1\!-\!s}\lambda_n}{\lambda^{s+1}(m\!-\!\!1\!\!-\!\!s)!}\!\!\Bigg[\!\!I(\lambda_n\!\!-\!\!\lambda)\!+\!\!\frac{\lambda}{m}\!\!\sum\limits_{k=0}^{m-1}\!\!\Big(\!a_{nk}I(\theta_k\!\!-\!\!\lambda)\!+\!b_{nk}I(\theta_k^*\!\!-\!\!\lambda)\!\!-\!\!(a_{nk}\!+\!b_{nk})I(\lambda_n\!\!\!-\!\!\lambda)\Big)\\[-1ex]
&-\!\sum\limits_{t=0}^{m-1}\frac{(\lambda\!-\!\lambda_n)^{t}}{t!}\Bigg\{\frac{M^{m+t-s}}{(m+t-s)}\!+\!\frac{1}{m}\sum\limits_{k=0}^{m-1}\!\Re\Big\{(\lambda\e^{\J\eta_k})^{-t}t!\Big[I(-\lambda\e^{\J\eta_k})
\!-\!\sum\limits_{q=0}^{t}\frac{(\lambda\e^{\J\eta_k})^q}{q!}\frac{M^{m+q-s}}{(m\!+\!q\!-\!s)}\Big]\Big\}\Bigg\}\Bigg],
\end{aligned}
\label{eq:B_sn_OO}
\end{equation}
with $I(\beta)=\beta^{-m+s}\gamma(m-s,\beta M)$ and $\zeta_{k}=M$ for $k=0$ and $\zeta_k=(1-\e^{-\theta_k M})/\theta_k$ for $ k\neq 0$.
\label{theo:stationary_dist_infinite_Nakagami_OO}
\end{corollary} \vspace{-0.3cm}
\ifARXIV\vspace{-0.2cm}
\begin{proof} The proof is provided in Appendix \ref{app:stationary_dist_infinite_Nakagami_OO}. \end{proof}
\else
\begin{proof} The proof is provided in Appendix \ref{app:stationary_dist_infinite_Nakagami_OO_shortened} and is given in more detail in \cite[Appendix F]{Morsi_storage_arxiv}. \end{proof}\vspace{-0.3cm}
\fi
\begin{remark} \normalfont
The limiting distributions of the energy stored in an infinite-size energy buffer for the best-effort and on-off policies obtained in Corollaries \ref{theo:stationary_dist_infinite_Nakagami} and \ref{theo:stationary_dist_infinite_Nakagami_OO} are real-valued despite the complex-valued parameters $\lambda_n$ and $c_n$. Furthermore, $\lambda_{n}\!=\!\lambda_{m-n}^*$ and $c_{n}\!=\!c_{m-n}^*$, $\forall$ $n\!=\!1,\ldots,\frac{m}{2}\!-\!1$  if $m$ is even and $\forall$ $n\!=\!1,\ldots,\frac{m-1}{2}$, if $m$ is odd, see 
\ifARXIV
Appendix \ref{app:real_valued_pdfs}.
\else
\cite[Appendix G]{Morsi_storage_arxiv}.
\fi
\label{coro:real_valued_pdfs}
\end{remark}

\section{Finite-Capacity Energy Buffer}
\label{s:Finite_buffer}
Studying a finite-size energy buffer is of interest if the average harvested energy is in the same order of magnitude as the capacity of the energy buffer. This can be the case when a supercapacitor is used as an energy buffer since supercapacitors are characterized by a very small energy density compared to rechargeable batteries \cite{Culler_2005}. On the other hand, supercapacitors are advantageous over rechargeable batteries since they provide a faster charging rate, a longer cycle lifetime, and a higher storage efficiency \cite{Culler_2005},\cite{Kansal_2007} and can be integrated on chip, see \cite{supercapacitor_on_chip}. 
 \vspace{-0.4cm}
\subsection{Existence of a Limiting Energy Distribution}
\begin{theorem}  \normalfont
For the storage processes in (\ref{eq:general_storage_equation_BE}) and (\ref{eq:general_storage_equation_OO}), if the energy buffer has a  finite size $K$, and the EH process $\{X(i)\}$ is characterized by a distribution with an infinite positive tail, then the process $\{B(i)\}$ is a stationary and ergodic process which possesses a unique stationary distribution $\pi$ that has a density on $(0,K)$ and an atom at $K$. Furthermore, the process converges in total variation to the limiting distribution $\pi$ from any initial distribution. 
\label{theo:limiting_dist_Finite}
\end{theorem}\vspace{-0.3cm}
\ifARXIV
\begin{proof} The proof is provided in \cite[Appendix D]{Morsi_ICC2015} and in Appendix \ref{app:limiting_dist_Finite_existence_uniqueness_BE_OO}. \end{proof}
\else
\begin{proof} The proof is provided in \cite[Appendix D]{Morsi_ICC2015} and \cite[Appendix H]{Morsi_storage_arxiv}. 
\end{proof}
\fi
\vspace{-0.4cm}
\subsection{Best-Effort Policy with a Finite-Size Energy Buffer}
\label{ss:BE_scheme_Finite}
\begin{theorem}  
\label{theo:BE_Finite_intgeral_eqn}\normalfont
Consider the storage process $\{B(i)\}$ for the best-effort policy in (\ref{eq:general_storage_equation_BE}) and a finite buffer size $K$. Let $g(x)$ be the limiting pdf of the stored energy on $(0,K)$ and $\pi(K)$ be the limiting probability of a full buffer (i.e., the atom at $K$). If $f(x)$ and $\bar{F}(x)$ are respectively the pdf and the ccdf of the EH process $\{X(i)\}$, then, $g(x)$ and $\pi(K)$ must jointly satisfy\vspace{-0.1cm}
\begin{numcases}{\hspace{-0.3cm}g(x)\!=\!\! \label{eq:g_integral_eqn_finite_BE}}
\hspace{-0.1cm}f(x)\!\! \int\limits_{u=0}^{M}\!\! g(u) \dd u \!+\!\! \int\limits_{u=M}^{M+x}\!\! f(x-u+M) g(u)\dd u, & \hspace{-0.3cm}$0\!\leq \!x<\!K\!-\! M $ \label{eq:parta_BE}\\
\hspace{-0.1cm}f(x)\!\!\int\limits_{u=0}^{M} \!\! \!\!g(u) \dd u\! +\!\!\!\! \int\limits_{u=M}^{K}\!\!\!\! f(x\!-\! u\!+\! M) g(u)\dd u \!+\! \pi(K) f(x\!-\! K\!+\! M), & \hspace{-0.3cm}$K\!-\!M\!\leq \! x\!<\!K$, \label{eq:partb_BE}
\end{numcases}
\begin{equation}
\pi(K)\!=\!\frac{\left[\bar{F}(K)\int\limits_{u=0}^{M} g(u) \dd u + \!\!\int\limits_{u=M}^{K} \bar{F}(K-u+M) g(u)\dd u \right]}{1-\bar{F}(M)},
\label{eq:partc_BE}
\end{equation}
and the unit area condition $\int_{0}^{K} g(u) \dd u + \pi(K)=1.
$
\end{theorem}\vspace{-0.3cm}
\ifARXIV
\begin{proof} The proof is provided in \cite[Proof of Theorem 5]{Morsi_ICC2015} and in Appendix \ref{app:BE_Finite_intgeral_eqn}. \end{proof}
\else
\begin{proof} The proof is provided in \cite[Proof of Theorem 5]{Morsi_ICC2015} and \cite[Appendix I]{Morsi_storage_arxiv}. \end{proof}
\fi 
Next, we consider the case when the DL channel is Nakagami-$m$ block fading., i.e., the EH process is Gamma distributed.
\begin{corollary}  \normalfont
Consider the storage process $\{B(i)\}$ for the best-effort policy in (\ref{eq:general_storage_equation_BE}) and a finite buffer of size $K\!=\!lM\!+\!\Delta$, where $l\!\in\!\mathbb{Z}^+$ and $0\!\leq\!\Delta\! <\! M$. If the EH process $\{X(i)\}$ is i.i.d. Gamma distributed with an integer shape parameter $m\in\{1,2,\ldots\}$ and pdf $f(x)\!=\!\frac{\lambda^m}{\Gamma(m)}x^{m-1}\e^{-\lambda x}$, where $\lambda\!\!=\!\!m/\bar{X}$, then the stored energy has a limiting pdf which can be obtained in stripes of width $M$, cf. Fig. \ref{fig:g_x_stripes}. In particular, $g(x)\!=\!g_n(x),\, K\!-\!(n\!+\!1)M\!\leq\! x\!<\!K\!-\!nM,\, n\!=\!0,\ldots,l'$, with, $l'\!=\!l-1$ if
$\Delta\!=\!0$ and $l'\!=\!l$ if $\Delta\!\neq\!0$, where $g_n(x)$ is given by
\begin{equation}
g_n(x)\!=\!\e^{\!\!-\lambda x}\sum\limits_{r=0}^{m-1}\!\frac{\alpha_r}{r!}\Biggg[\!\!
\sum\limits_{q=0}^{n}\!\!\lambda^{(q+1)m-r}\e^{\!-\lambda M q}\frac{ (qM\!+\!x\!-\!K)^{(q+1)m-r-1}}{((q+1)m-r-1)!}
 - \!\sum\limits_{q=1}^{n}\!\!\lambda^{qm}\e^{\!-\lambda M q}\frac{ (qM\!+\!x\!-\!K)^{qm-1}}{(qm-1)!}\Biggg], 
  \label{eq:pdf_Nakagami_m}\vspace{-0.5cm}
\end{equation}
\ifARXIV\\[-2ex]
\noindent and the probability of a full buffer $\pi(K)$ is given by \vspace{-0.2cm}
\else
and the probability of a full buffer $\pi(K)$ is given by
\fi
\begin{equation}
\pi(K)=\e^{-\lambda K}\sum\limits_{r=0}^{m-1}\frac{\alpha_r}{r!},
\label{eq:Pi_K_Nakagami}
\end{equation}
\begin{figure}[!thp] 
\centering
\includegraphics[width=1\textwidth]{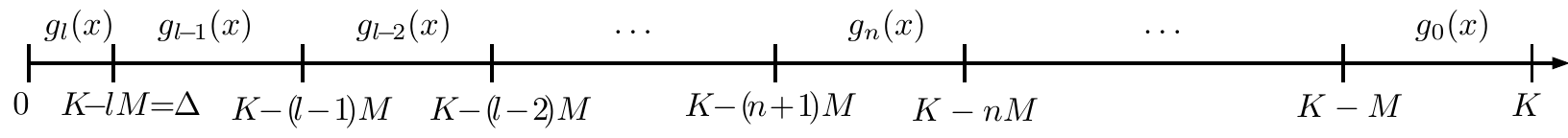}
 \caption{The energy distribution is obtained in stripes of width $M$ for an energy buffer of size $K\!=\!lM\!+\!\Delta$, $l\in\mathbb{Z}^+,\,0\!\leq\!\Delta\! <\! M$, see (\ref{eq:pdf_Nakagami_m}) and (\ref{eq:pdf_Nakagami_OO}) for the best-effort and the on-off policies, respectively.\vspace{-0.7cm}}
\label{fig:g_x_stripes}
\end{figure}
where the coefficients $\alpha_r$, $r=0,\ldots,m-1$, are obtained by solving the  system of linear equations $\left(\V{I}+\V{A}\right)\V{\alpha}=\V{1}_m$, where $\V{\alpha}=[\alpha_0,\ldots,\alpha_{m-1}]^T$, and $\V{A}$ is an $m\times m$ matrix whose entry in the $s^{\text{th}}$  row and the $r^{\text{th}}$ column, for $s,r=0,\ldots,m-1$, is given by
\begin{equation}
\begin{aligned}
\!&A_{sr}\!=\!\frac{(1\!-\!K^s\lambda^s)}{r!}\!\!\sum\limits_{q=0}^{l}\!\!\e^{-\lambda M q}\!\!\!\!\!\sum\limits_{t=qm}^{(q+1)m-r-1}\!\!\!\frac{\left(\lambda(qM\!-\!K)\right)^t}{t!}\!+\!K^s\frac{\lambda^s}{r!}\sum\limits_{q=0}^{l-1}\!\e^{-\lambda M (q+1)}\!\!\!\!\!\!\sum\limits_{t=qm}^{(q+1)m-r-1}\!\!\frac{\left(\lambda\left((q\!+\!1)M\!-\!K\right)\right)^t}{t!}\\
&\!+\!\frac{K^s\lambda^{s}}{r!}\!\sum\limits_{t=0}^s\binom{s}{t}K^{-t}t!\sum\limits_{q=0}^{l-1}\!\!\lambda^{qm}\e^{-\lambda M (q+1)}\!\left[\!\frac{\lambda^{m-r}\left((q\!+\!1)M-K\right)^{(q+1)m-r+t}}{((q+1)m-r+t)!}\!-\!\frac{\left((q\!+\!1)M\!-\!K\right)^{qm+t}}{(qm+t)!}\right].
\end{aligned}
\label{eq:asr}
\end{equation}
\label{coro:limiting_dist_Finite_Nakagami_exact_BE}  
\end{corollary}\vspace{-1.2cm}
\ifARXIV
\begin{proof} The proof is provided in Appendix \ref{app:limiting_dist_Finite_Nakagami_exact_BE}. \end{proof}
\else
\begin{proof} The proof is provided in Appendix \ref{app:limiting_dist_Finite_Nakagami_exact_BE}  and is given in more detail in \cite[Appendix J]{Morsi_storage_arxiv}. \end{proof}
\fi\vspace{-0.4cm}
\subsection{On-Off Policy with a Finite-Size Energy Buffer}
\label{ss:OO_scheme_Finite}
\begin{theorem}  \normalfont
\label{theo:OO_Finite_intgeral_eqn}
Consider the storage process $\{B(i)\}$ for the on-off policy in (\ref{eq:general_storage_equation_OO}) and a finite buffer size\footnote{Note that if $K<2M$, then (\ref{eq:parta}) is valid for $0<x<K-M$, (\ref{eq:partb}) does not exist, and (\ref{eq:partc}) remains  unchanged.} 
 $K >2M$. Let $g(x)$ be the limiting pdf of the stored energy on $(0,K)$ and $\pi(K)$ be the limiting probability of a full buffer (i.e., the atom at $K$). If $f(x)$ and $\bar{F}(x)$ are respectively the pdf and the ccdf of the EH process $\{X(i)\}$, then, $g(x)$ and $\pi(K)$ must jointly satisfy\vspace{-0.1cm}
\begin{numcases}{\hspace{-1.5cm}g(x)\!=\!\! \label{eq:g_integral_eqn_finite}}
\hspace{-0.1cm}\int\limits_{u=0}^{x}\!\!f(\!x\!-\!u\!) g(\!u\!) \dd u +\!\!\!\! \int\limits_{u=M}^{M+x} f(\!x\!-\!u\!+\!M\!) g(\!u\!)\dd u, & $0\!\leq\! x\!<\!M $ \label{eq:parta}\\
\hspace{-0.1cm}\int\limits_{u=0}^{M}\!\!f(\!x\!-\!u\!) g(\!u\!) \dd u +\!\!\!\! \int\limits_{u=M}^{M+x} f(\!x\!-\!u\!+\!M\!) g(\!u\!)\dd u, & $M\!\leq \!x\!<\!K\!-\!M $ \label{eq:partb}\\
\hspace{-0.1cm}\int\limits_{u=0}^{M}\!\!f(\!x\!-\!u\!) g(\!u\!) \dd u +\!\!\!\! \int\limits_{u=M}^{K}\!\! f(\!x\!-\! u\!+\! M\!) g(\!u\!)\dd u \!+\! \pi(\!K\!) f(\!x\!-\! K\!+\! M\!), & $K\!-\!M\!\leq \! x\!<\!K$
\label{eq:partc}
\end{numcases}
\begin{equation}
\pi(K)\!=\!\frac{\left[\int\limits_{u=0}^{M}\bar{F}(K-u) g(u) \dd u + \!\!\int\limits_{u=M}^{K} \bar{F}(K-u+M) g(u)\dd u \right]}{1-\bar{F}(M)},
\label{eq:partd}
\end{equation}
and the unit area condition $\int_{0}^{K} g(u) \dd u + \pi(K)=1.$
\end{theorem}\vspace{-0.4cm}
\ifARXIV
\begin{proof} The proof is provided in \cite[Proof of Theorem 5]{Morsi_Asilomar2014} and in Appendix \ref{app:OO_Finite_intgeral_eqn}. \end{proof}
\else
\begin{proof} The proof is provided in \cite[Proof of Theorem 5]{Morsi_Asilomar2014} and \cite[Appendix K]{Morsi_storage_arxiv}. \end{proof}
\fi
\begin{corollary}  \normalfont
Consider the storage process $\{B(i)\}$ for the on-off policy in (\ref{eq:general_storage_equation_OO}) and a finite buffer of size $K=lM$ with $l\in\{3,4,\ldots\}$ and a Nakagami-$m$ fading DL channel, i.e., an i.i.d. Gamma distributed EH process $\{X(i)\}$ with an integer shape parameter $m\in\{1,2,\ldots\}$ and pdf $f(x)\!=\!\frac{\lambda^m}{\Gamma(m)}x^{m-1}\e^{-\lambda x}$, where $\lambda\!\!=\!\!m/\bar{X}$. Then, the stored energy has a limiting pdf which can be obtained in stripes of width $M$, cf. Fig. \ref{fig:g_x_stripes} (at $\Delta=0$). In particular, 
 \small
\begin{numcases}{\hspace{-0.2cm}g(x)\!=\!\!\label{eq:pdf_Nakagami_OO}}
\begin{aligned}\!\!\!g_{l\!-\!1}(x)\!=\!\e^{-\lambda x}\!\!\sum\limits_{r=0}^{m-1}\!\frac{\alpha_r}{r!}\Biggg[\!\!&\sum\limits_{q=0}^{l-2}\!\lambda^{(q+2)m-r}\e^{\!-\lambda M(q+1)}\!\!\sum\limits_{t=0}^{(q\!+\!1)m\!-\!r\!-\!1}\!\!\frac{((q\!+\!1)M\!-\!K)^{(q\!+\!1)m\!-\!r\!-\!1\!-\!t}}{((q\!+\!1)m\!-\!r\!-\!1\!-\!t)!(m\!+\!t)!}C(x,t)\\[-1ex]
&\!\!\!-\sum\limits_{q=1}^{l-2}\lambda^{(q+1)m}\e^{-\lambda M(q+1)}\sum\limits_{t=0}^{qm\!-\!1}\frac{((q\!+\!1)M\!-\!K)^{qm\!-\!1\!-\!t}}{(qm\!-\!1\!-\!t)!(m\!+\!t)!}C(x,t)\Biggg],\end{aligned} & 
\hspace{-2.4cm} $0\!\leq \!x<\! M$ \label{eq:g1_OO}\\
\!\!\!g_n(x)\!=\!\e^{\!-\lambda x}\!\!\sum\limits_{r=0}^{m-1}\!\frac{\alpha_r}{r!}\Biggg[\!\!\!
\sum\limits_{q=0}^{n}\!\!\lambda^{(q\!+\!1)m-r}\e^{\!-\lambda M q}\frac{ (qM\!+\!x\!-\!K)^{(q\!+\!1)m\!-\!r\!-\!1}}{((q\!+\!1)m\!-\!r\!-\!1)!}
 -\!\sum\limits_{q=1}^{n}\!\!\lambda^{qm}\e^{\!-\lambda M q}\frac{ (qM\!+\!x\!-\!K)^{qm\!-\!1}}{(qm\!-\!1)!}\Biggg],   &\nonumber \\[-1ex]
\hspace{9.4cm} \begin{aligned}
  &&[K\!-\!(n\!+\!1)M]^+\leq x<K\!-\!nM  \\[-2ex] 
  &&n\!=\!0,\ldots,l-2,
\end{aligned} & \label{eq:g2_OO}\end{numcases} 

 \normalsize 
 where
 \small
 \begin{equation}
 C(x,t)=\left(x^{m+t}+\frac{1}{m}\sum\limits_{k=0}^{m-1}\Re\Big\{(\lambda\e^{\J\eta_k})^{-(m+t)}\e^{\lambda x \e^{\J\eta_k}}\gamma\Big(m+t+1,\lambda x \e^{\J\eta_k}\Big)\Big\}\right),
\label{eq:C_x_t}
\end{equation}
\normalsize
  and $\eta_k=\frac{2\pi k}{m}$. The probability of a full buffer  $\pi(K)$ is given by\vspace{-0.2cm}
 \small
\begin{equation}
\pi(K)=\e^{-\lambda K}\sum\limits_{r=0}^{m-1}\frac{\alpha_r}{r!},
\label{eq:Pi_K_Nakagami_OO}
\end{equation}
\normalsize
where $\alpha_r$, $r=0,\ldots,m-1$, is obtained by solving the non-homogeneous system of linear equations $\left(\V{I}+\V{A}\right)\V{\alpha}=\V{1}_m$, where $\V{\alpha}=[\alpha_0,\ldots,\alpha_{m-1}]^T$, and $\V{A}$ is an $m\times m$ matrix whose entry in the $s^{\text{th}}$  row and the $r^{\text{th}}$ column, for $s,r=0,\ldots,m-1$, is given by
\small 
\begin{equation}
\begin{aligned}
A_{sr}&=\frac{1}{r!}\Biggg\{\sum\limits_{q=0}^{l-2}\lambda^{(q+2)m-r}\e^{-\lambda M(q+1)}\sum\limits_{t=0}^{(q+1)m-r-1}\frac{((q+1)M-K)^{(q+1)m-r-1-t}}{((q+1)m-r-1-t)!}\Big(H(t)-\lambda^sF(t,s)\Big)\\[-1ex]
&-\sum\limits_{q=1}^{l-2}\lambda^{(q+1)m}\e^{-\lambda M(q+1)}\sum\limits_{t=0}^{qm-1}\frac{((q+1)M-K)^{qm-1-t}}{(qm-1-t)!}\Big(H(t)-\lambda^sF(t,s)\Big)\\[-1ex] 
&+\sum\limits_{q=0}^{l-2}\e^{-\lambda M(q+1)}\sum\limits_{t=qm}^{(q+1)m-r-1}\frac{\lambda^t((q+1)M-K)^{t}}{t!}\\[-1ex]
&+\lambda^sK^s\sum\limits_{q=0}^{l-2}\lambda^{qm}\e^{-\lambda M (q+1)}\sum\limits_{t=0}^s\binom{s}{t}K^{-t}t!\left[\frac{\lambda^{m-r}\left((q+1)M-K\right)^{(q+1)m-r+t}}{((q+1)m-r+t)!}-\frac{\left((q+1)M-K\right)^{qm+t}}{(qm+t)!}\right]
&\Biggg\},
\end{aligned}
\label{eq:asr_OO}
\end{equation}
\normalsize
where
\small
\begin{equation}
\begin{aligned}
&F(t,s)=\sum\limits_{b=0}^{s}\binom{s}{b}(K-M)^{s-b}\Biggg[\frac{b! M^{m+t+b+1}}{(m+t+b+1)!}+\frac{1}{m}\sum\limits_{k=0}^{m-1}\Re\Bigg\{\e^{\lambda M\e^{\J\eta_k}} \big(\lambda \e^{\J\eta_k}\big)^{-(m+t+b+1)}\gamma\big(b+1,\lambda M\e^{\J\eta_k}\big)\\[-1ex]
&-b!\sum\limits_{w=0}^{m+t} \big(\lambda \e^{\J\eta_k}\big)^{w-(m+t)}\frac{M^{w+b+1}}{(w+b+1)!}\Bigg\}\Biggg],
\end{aligned}
\label{eq:F_t_s}
\end{equation}
\small
\begin{flalign}
\hspace{-2cm}H(t)=\lambda^{-(m+t+1)}\Biggg[\frac{\gamma\Big(m+t+1,\lambda M\Big)}{(m+t)!}+\frac{1}{m}\sum\limits_{k=0}^{m-1}\Re\Big\{\rho_k-\sum\limits_{w=0}^{m+t}\frac{\e^{\J\eta_k(w-(m+t))}}{w!}\gamma\Big(w+1,\lambda M\Big)\Big\}\Biggg],
\label{eq:H_T}
\end{flalign}
and $\rho_k=
\begin{cases}
\lambda M & k=0\\
\e^{-\J\eta_k(m+t)}(1-\e^{\J\eta_k})^{-1}(1-\e^{-\lambda M(1-\e^{\J\eta_k})}))&k\neq 0
\end{cases}
$.\normalsize
\label{coro:limiting_dist_Finite_Nakagami_exact_OO}  
\end{corollary}
\ifARXIV
\begin{proof} The proof is provided in Appendix \ref{app:limiting_dist_Finite_Nakagami_exact_OO_arxiv}. \end{proof}
\else
\begin{proof}
The proof is provided in Appendix \ref{app:limiting_dist_Finite_Nakagami_exact_OO_short} and is given in more detail in \cite[Appendix L]{Morsi_storage_arxiv}.
\end{proof}
\fi
\begin{remark}
Although the energy distributions obtained in Corollaries  \ref{theo:stationary_dist_infinite_Nakagami}--\ref{coro:limiting_dist_Finite_Nakagami_exact_OO} seem quite involved, they are easy to implement and have a low computational complexity. The reason for this is three-fold. First, a matrix inversion is required only  for matrices of size $m\!\times \!m$, where $m$ is usually a small number, e.g., up to $m=4$ for realistic Nakagami-$m$ fading channels. This is unlike the case of discrete-state Markov chains, see \cite{DTS_WPC_Energy_accumulation_2015}, where even for Rayleigh fading, i.e., $m=1$,  matrix inversions of matrices of size $L\times L$ are required, where $L$ is the number of discrete energy states which needs to be large (e.g., $L\!=\!300$) for an accurate approximation of the energy distribution. Second, the limits of the summations in the expressions for the energy distribution  are either $m$ for an infinite-size buffer or $m$ and $l$, for a finite-size buffer. As discussed previously, $m$ is a small number and if $l\!=\!\lfloor K/M \rfloor$ is large, the buffer can be approximated as having infinite size. Furthermore, with an infinite-size buffer, only half of the unknown parameters $\lambda_n$ and $c_n$ have to be obtained,  cf. Remark \ref{coro:real_valued_pdfs}. Finally, since the energy distributions depend on the statistics of the system, they have to be obtained only once for a given setup. Then, using the derived energy distributions, the  performance of the considered transmission policies can be analyzed and the system parameters can be optimized, as will be shown in Sections \ref{s:Outage_analysis} and \ref{s:Simulations}.
\label{rem:low_computation_energy_distributions}
\end{remark}
\section{Outage Probability and Average Throughput Analysis}
\label{s:Outage_analysis}
In this section, we analyze the outage probability  and the corresponding average throughput of the UL channel, when the UL and DL channels are both Nakagami-$m$ faded. Thus, the UL channel power $\hul$ is Gamma distributed with shape parameter $\mul$, mean $\Omegaul$, and cdf $\mathbb{P}(\hul<x)=\frac{\gamma\left(\mul,\frac{\mul}{\Omegaul}x\right)}{\Gamma(\mul)}$. Since the CSI is unknown at the EH node, the node transmits data at a constant rate of $R$ bits/(channel use). Therefore, assuming the use of a capacity-achieving code, an outage occurs when $R>\log_2(1+\gamma)\Rightarrow \gamma<\gamma_{\rm thr}$, where $\gamma$ is the UL instantaneous signal-to-noise ratio (SNR) and $\gamma_{\rm thr}=2^{R}-1$ is the threshold SNR. Before analyzing the outage probability of the considered transmission policies, we first provide the optimal operating range for the maximum UL transmit power $M$ for both transmission policies for an infinite-size energy buffer.
\begin{corollary}\normalfont
For the best-effort and the on-off transmission policies with an infinite-size energy buffer, the optimal maximum UL transmit power $M$, which minimizes the outage probability of information transmission in the UL, is always larger than or equal to the average harvested power $\bar{X}$, i.e., $M\geq\bar{X}$ or $\delta=M/\bar{X}\geq 1$ is optimal.
\label{coro:delta1_opt_BE_INF}
\end{corollary}
\begin{proof}
In order to find the optimal maximum transmit power $M$ of the EH node, which minimizes the outage probability, we study how the best-effort and the on-off transmission policies behave as $M$  increases relative to the constant average input harvested power $\bar{X}$. First, if $M<\bar{X}$, then both transmission policies are identical and from Theorem \ref{theo:no_stationary_dist}, the energy will accumulate in the buffer and the desired transmit power $M$ will always be available in the energy buffer, i.e., $B(i)\geq M$ and $\Pul(i)=\min(B(i),M)=M,\,\forall\,\,i$  for the best-effort policy and $\Pul(i)=M\mathds{1}_{B(i)\geq M}=M,\,\, \forall \,\, i$ for the on-off policy. This means, for $M<\bar{X}$, increasing $M$ always improves the outage performance since the UL transmit power increases, and the system will suffer from fewer outage events for a given transmission rate. Therefore, the optimal maximum UL power $M$ must be larger than or equal to $\bar{X}$. This completes the proof. 
\end{proof}
\begin{remark}
We note that, so far, in Sections \ref{s:Infinite_buffer} and \ref{s:Finite_buffer}, we have studied the distribution of the stored energy for an ideal system, cf. Section \ref{s:imperfections}. For a non-ideal system, we can use the derived energy distributions after replacing $M$ by $\tilde{M}\!=\!\Pc+\rho M$, $\bar{X}$ by $\widetilde{\bar{X}}=\beta\bar{X}$, $\delta$ by $\tilde{\delta}=\tilde{M}/\widetilde{\bar{X}}$, and $\lambda$ by $\tilde{\lambda}=\frac{m}{\widetilde{\bar{X}}}=\frac{\lambda}{\beta}$. In this case, $\tilde{M}$ is the minimum amount of stored energy needed to transmit with UL power $M$ and the UL transmit power is given by (\ref{eq:Puwith_imperfections}).
\end{remark}
\begin{table*}[!tbp]
  \caption{Outage Probability for the best-effort and the on-off policies. $\Gamma_{\rm thr}=\frac{\mul\gamma_{\rm thr}\sigma^2}{\Omegaul}$, $\tilde{\delta}=\frac{\tilde{M}}{\widetilde{\bar{X}}}$, and $P_M\!=\!\mathbb{P}(\Pul\!=\!M)$. $\lambda_n$ and $c_n$ for the best-effort and the on-off policies are specified in Corollaries \ref{theo:stationary_dist_infinite_Nakagami} and \ref{theo:stationary_dist_infinite_Nakagami_OO}, respectively. $\alpha_r$ for the best-effort and the on-off policies are specified in Corollaries \ref{coro:limiting_dist_Finite_Nakagami_exact_BE} and \ref{coro:limiting_dist_Finite_Nakagami_exact_OO}, respectively. $N(t,a,b,c,d)=\int\limits_0^1 \frac{\e^{-\frac{a}{x}-b x}}{x^t}\frac{(x-c)^d}{d!}\dd x$.}\vspace{-0.2cm}
\label{tab:outage_probabilities}
		\begin{tabular}{@{}ll@{}}\toprule   
		Case & $P_M$ and $P_{\text{out}}$\\ \midrule \addlinespace[1em]
		$\begin{aligned} &\text{Best-effort and on-off} \\[-2ex]
&K\to\infty\text{, }\tilde{\delta}\leq 1 \\[-2ex]
&{\rm cf.\,\, Theorem\,\,\ref{theo:no_stationary_dist} }\end{aligned}$ & $\begin{aligned}&P_M=1\\[-2ex]
		&P_{\text{out}}=P_{\text{out}}\big|_{\Pul=M}=\frac{\gamma\left(\mul,\Gamma_{\rm thr}\rho/(\tilde{M}-\Pc)\right)}{\Gamma(\mul)}\end{aligned}$
\\ \midrule 
				$\begin{aligned} &\text{Best-effort, }K\to\infty \\[-2ex]
&\tilde{\delta}> 1 \\[-2ex]
&{\rm cf.\,\, Corollary\,\,\ref{theo:stationary_dist_infinite_Nakagami} } \end{aligned}$& $\begin{aligned}&P_M=\sum\limits_{n=0}^{m-1}c_n\e^{-\lambda_n \tilde{M}}\\[-5ex]
				&P_{\text{out}}\!=\!P_MP_{\text{out}}\big|_{\Pul=M}\!+\!(1\!-\!P_M)\!-\!\overbrace{(\tilde{M}\!-\!\Pc)\!\sum\limits_{n=0}^{m-1}\!\lambda_nc_n \e^{-\lambda_n\Pc}\!\!\sum\limits_{t=0}^{\mul-1}\!\frac{\left(\!\frac{\Gamma_{\rm thr}\rho}{\tilde{M}-\Pc}\!\right)^t}{t!}N\!\left(\!t,\frac{\Gamma_{\rm thr}\rho}{\tilde{M}\!-\!\Pc},\lambda_n(\tilde{M}\!-\!\Pc),0,0\!\right)}^{\Sigma_I}\end{aligned}$
\\ \midrule
				$\begin{aligned} &\text{On-Off, }K\to\infty \\[-2ex]
&\tilde{\delta}> 1 \\[-2ex]
&{\rm cf.\,\, Corollary\,\,\ref{theo:stationary_dist_infinite_Nakagami_OO} }  \end{aligned}$& $\begin{aligned}&P_M=\sum\limits_{n=0}^{m-1}c_n\e^{-\lambda_n \tilde{M}}\\[-1ex]
				&P_{\text{out}}=P_MP_{\text{out}}\big|_{\Pul=M}+(1-P_M)\end{aligned}$		
\\ \midrule
	$\begin{aligned} &\text{Best-effort, }K<\infty \\[-2ex]
&{\rm cf.\,\, Corollary\,\,\ref{coro:limiting_dist_Finite_Nakagami_exact_BE} } \end{aligned}$ & $\begin{aligned}&P_M=\sum\limits_{r=0}^{m-1}\frac{\alpha_r}{r!}\sum\limits_{q=0}^{l-1}\e^{-\tilde{\lambda} \tilde{M}(q+1)}\sum\limits_{t=qm}^{(q+1)m-r-1}\frac{\big(\tilde{\lambda}((q+1)\tilde{M}-K)\big)^t}{t!}\\[-1ex]
	&P_{\text{out}}=P_MP_{\text{out}}\big|_{\Pul=M}+(1-P_M)-\overbrace{\sum\limits_{t=0}^{\mul-1}\frac{(\Gamma_{\rm thr}\rho)^t}{t!} I_t}^{\Sigma_F},\\[-1ex]
		&I_t\!=\!\!\sum\limits_{r=0}^{m-1}\!\frac{\alpha_r}{r!}\Bigg[\!\sum\limits_{q=0}^{l'}\!\e^{-\tilde{\lambda} (q\tilde{M}+\Pc)}\frac{(\tilde{\lambda} (D\!-\!\Pc))^{(q\!+\!1)m\!-\!r}}{(D\!-\!\Pc)^t}N\!\left(\!t,\frac{\Gamma_{\rm thr}\rho}{D\!-\!\Pc},\tilde{\lambda}(D\!-\!\Pc),\frac{(K\!-\!(q\tilde{M}\!+\!\Pc))}{D\!-\!\Pc},(q\!+\!1)m\!-\!r\!-\!1\!\right)\\[-1ex]
	&\hspace{2cm}-\sum\limits_{q=1}^{l'}\e^{-\tilde{\lambda} (q\tilde{M}+\Pc)}\frac{(\tilde{\lambda}(D\!-\!\Pc))^{qm}}{(D\!-\!\Pc)^t}N\!\left(\!t,\frac{\Gamma_{\rm thr}\rho}{D\!-\!\Pc},\tilde{\lambda} (D\!-\!\Pc),\frac{(K\!-\!(q\tilde{M}\!+\!\Pc))}{(D\!-\!\Pc)},qm\!-\!1\!\right)\Bigg]\\
	& l'=\begin{cases} l-1 & \Delta<\Pc\\ l & \Delta>\Pc\end{cases},\quad D=\begin{cases} \tilde{M} & q<l \\	\Delta & q=l\end{cases},\quad K=l\tilde{M}+\Delta,\quad l\in\mathbb{Z^+},\quad \Delta<\tilde{M}.
	\end{aligned}$
\\ \midrule
$\begin{aligned} &\text{On-Off, }K<\infty \\[-2ex]
&{\rm cf.\,\, Corollary\,\,\ref{coro:limiting_dist_Finite_Nakagami_exact_OO} } \end{aligned}$
& $\begin{aligned}&P_M=\sum\limits_{r=0}^{m-1}\frac{\alpha_r}{r!}\sum\limits_{q=0}^{l-2}\e^{-\tilde{\lambda} \tilde{M}(q+1)}\sum\limits_{t=qm}^{(q+1)m-r-1}\frac{\big(\tilde{\lambda}((q+1)\tilde{M}-K)\big)^t}{t!}\\
	&P_{\text{out}}=P_MP_{\text{out}}\big|_{\Pul=M}+(1-P_M)\end{aligned}$
\\ \bottomrule\vspace{-1cm}
\end{tabular}		
\end{table*}
\vspace{-0.6cm}\begin{proposition}\normalfont
For the on-off policy, the outage probability is the probability that either the EH node does not transmit because the stored energy is less than $\tilde{M}$ or the EH node transmits with power $M$ but an outage occurs, i.e.,  $P_{\text{out}}\big|_{\rm on-off}=P_MP_{\text{out}}\big|_{\Pul=M}+(1-P_M)$, where $P_M$ is defined as the probability that the EH node transmits with  power $M$, i.e., $P_M\!=\!\mathbb{P}(\Pul(i)\!=\!M)=\mathbb{P}(B(i)\geq \tilde{M})=\int_{\tilde{M}}^{K}g(x)\dd x$ and $P_{\text{out}}\big|_{\Pul=M}$ is the outage probability given that the EH node transmits with power $M$, i.e., $P_{\text{out}}\big|_{\Pul=M}=\mathbb{P}(\frac{Mh_{\rm UL}}{\sigma^2}<\gamma_{\rm thr})=\frac{\gamma(\mul,\Gamma_{\rm thr}/M)}{\Gamma(\mul)}$, where $\Gamma_{\rm thr}=\frac{\mul\gamma_{\rm thr}\sigma^2}{\Omegaul}$.  For the best-effort policy, the outage probability is given by $P_{\text{out}}\big|_{\rm best-effort}\!=\! \!P_MP_{\text{out}}\big|_{\Pul=M}+\!\int_0^{\tilde{M}}\!\mathbb{P}\left(\!\frac{[x-\Pc]^+}{\rho}\frac{h_{\rm UL}}{\sigma^2}<\gamma_{\text{thr}}\!\right)g(x) \dd x$, since transmissions are also allowed when the energy stored is less than $\tilde{M}$. Using these definitions and the pdf $g(x)$ of the stored energy given in Corollaries  \ref{theo:stationary_dist_infinite_Nakagami}--\ref{coro:limiting_dist_Finite_Nakagami_exact_OO}, the outage probabilities of the considered transmission policies with finite- and infinite-size energy buffers can be computed and are summarized in Table \ref{tab:outage_probabilities}. Furthermore, for a transmission rate of $R$ bits/(channel use), the average throughput is $T=R(1-P_{\rm out})$.
\end{proposition}
\begin{proof}
Consider first an infinite-size buffer. If $\tilde{\delta}\leq1$, both transmission policies are identical since $\Pul(i)\!=\!M$, $\forall i$, cf. Theorem \ref{theo:no_stationary_dist}. For an infinite-size buffer with $\tilde{\delta}>1$, $P_M=\int_{\tilde{M}}^\infty g(x) \dd x$, which can be easily obtained from $g(x)=\sum_{n=0}^{m-1}\lambda_nc_n\e^{-\lambda_n x}$ for both policies. We note that the best-effort and the on-off policies have different coefficients $c_n$, cf. Corollaries  \ref{theo:stationary_dist_infinite_Nakagami} and \ref{theo:stationary_dist_infinite_Nakagami_OO}, and therefore different $P_M$. For the best-effort policy, the second term in the outage equation is given by $\int_0^{\tilde{M}}\!\mathbb{P}\left(\!\frac{[x-\Pc]^+}{\rho}\frac{h_{\rm UL}}{\sigma^2}<\gamma_{\text{thr}}\!\right)g(x) \dd x\!=\!\int_0^{\Pc} g(x)\dd x+\int_{\Pc}^{\tilde{M}}\frac{\gamma\left(\mul,\frac{\Gamma_{\rm thr}\rho}{x-\Pc}\right)}{\Gamma(\mul)}g(x) \dd x$. Using $\frac{\gamma\left(\mul,\frac{\Gamma_{\rm thr}\rho}{x-\Pc}\right)}{\Gamma(\mul)}\!=\!1\!-\!\e^{-\frac{\Gamma_{\rm thr}\rho}{x-\Pc}}\sum_{t=0}^{\mul-1}\frac{\left(\frac{\Gamma_{\rm thr}\rho}{x-\Pc}\right)^t}{t!}$ and $\int_0^{\Pc} g(x)\dd x+\int_{\Pc}^{\tilde{M}} g(x)\dd x=1-P_M$, we get\\\vspace{-0.5cm}
\begin{equation}
\int_0^{\tilde{M}}\!\mathbb{P}\left(\!\frac{[x-\Pc]^+}{\rho}\frac{h_{\rm UL}}{\sigma^2}<\gamma_{\text{thr}}\!\right)g(x) \dd x=(1-P_M)-\sum_{t=0}^{\mul-1}\frac{(\Gamma_{\rm thr}\rho)^t}{t!}\underbrace{\int_{\Pc}^{\tilde{M}}\!\frac{\e^{-\frac{\Gamma_{\rm thr}\rho}{x-\Pc}}}{(x-\Pc)^t}g(x) \dd x}_{I_t}.
\label{eq:second_term_outage_BE}
\end{equation}
Using $g(x)=\sum_{n=0}^{m-1}\lambda_nc_n\e^{-\lambda_n x}$, the outage probability can be obtained. Consider next a finite-size buffer. In this case, the mathematical expressions for $P_M\!=\!\int_{\tilde{M}}^K g(x)\dd x\!+\!\pi(K)$  are identical for both the best-effort and the on-off policies, since the expressions for $g(x)$, $x>\tilde{M}$, and $\pi(K)$  are identical  for both policies, cf. (\ref{eq:pdf_Nakagami_m}), (\ref{eq:g2_OO}), (\ref{eq:Pi_K_Nakagami}), and (\ref{eq:Pi_K_Nakagami_OO}). However, $P_M$ has a different value for each policy due to the different coefficients $\alpha_r$, $r=0,\ldots,m-1$, cf. Corollary \ref{coro:limiting_dist_Finite_Nakagami_exact_BE} and Corollary \ref{coro:limiting_dist_Finite_Nakagami_exact_OO}. $P_M$ can be obtained in a similar manner as (\ref{eq:area_under_g_step4}) in Appendix \ref{app:limiting_dist_Finite_Nakagami_exact_BE} except that the lower integral limit is $\tilde{M}$. By solving integrals of the form $\int_{\tilde{M}}^b\e^{-\tilde{\lambda} u}\frac{(u-b)^c}{c!}=-\tilde{\lambda}^{-c-1}\left(\e^{-\tilde{\lambda} b}-\e^{-\tilde{\lambda} \tilde{M}}\sum_{t=0}^{c}\frac{(\tilde{\lambda}(\tilde{M}-b))^t}{t!}\right)$, for $b=K-q\tilde{M}$, and $c=(q\!+\!1)m\!-\!r\!-\!1$ and $c=qm\!-\!1$, respectively, $P_M$ can be obtained. Note that the upper limit of $q$ in $P_M$ in the last two rows of Table \ref{tab:outage_probabilities} is $l-1$ for the best-effort policy and $l-2$ for the on-off policy. This is because for the on-off policy, we assume that $K=l\tilde{M}$, with $l\in\mathbb{Z}^+$, i.e., $\Delta=0$, cf. Fig \ref{fig:g_x_stripes}. Finally, to obtain the integral $I_t$ in (\ref{eq:second_term_outage_BE}), define $E(x)=\!\frac{\e^{-\frac{\Gamma_{\rm thr}\rho}{x-\Pc}}}{(x-\Pc)^t}$, then $I_t$ can be written as  $I_t=\int\limits_{\Pc}^{\tilde{M}}\!E(x)g(x) \dd x$. If $\Delta<\Pc$, then $I_t=\int\limits_{\Pc}^{\tilde{M}}\!\!E(x)g_{l-1}(x) \dd x$. Otherwise, if $\Delta>\Pc$, then $I_t=\int\limits_{\Pc}^\Delta\!\!E(x)g_l(x) \dd x+\int\limits_\Delta^{\tilde{M}}\!\!E(x)g_{l-1}(x) \dd x=\int\limits_{\Pc}^{\tilde{M}}\!\!E(x)g_{l-1}(x) \dd x+\int\limits_{\Pc}^\Delta\!\!E(x)L(x) \dd x$, where $L(x)=g_l(x)-g_{l-1}(x)$. This requires solving integrals of the form $\int_0^1\frac{\e^{-\frac{\Gamma_{\rm thr}\rho}{(D-\Pc)x}-\tilde{\lambda} (D-\Pc) x-\tilde{\lambda}\Pc}}{(D-\Pc)^t x^t}\frac{(x-\frac{K-(q\tilde{M}+\Pc)}{D-\Pc})^{d}}{d!} (D-\Pc)^{d+1}\,\dd x= (D-\Pc)^{d+1-t}\e^{-\tilde{\lambda}\Pc} N(t,\frac{\Gamma_{\rm thr}\rho}{D-\Pc},\tilde{\lambda} (D-\Pc),\frac{K-(q\tilde{M}+\Pc)}{D-\Pc},d)$ for $d=(q\!+\!1)m\!-\!r\!-\!1$ and $d=qm\!-\!1$, where \small  $D=\begin{cases} \tilde{M} & q<l\\	  \Delta & q=l\end{cases}$ \normalsize and $N(t,a,b,c,d)=\int\limits_0^1 \frac{\e^{-\frac{a}{x}-b x}}{x^t}\frac{(x-c)^d}{d!}\dd x$ is a bounded integral that has no closed form and is therefore solved numerically. This completes the proof.
\end{proof}
\begin{proposition}\normalfont
Define $G\definedas\frac{\Gamma(\mul,\mul b)}{\Gamma(\mul)}$, $G_{\tilde{\delta}_b}\definedas\frac{\Gamma(\mul,\mul b/\tilde{\delta}_b)}{\Gamma(\mul)}$, $b\definedas\frac{\gamma_{\rm thr}\sigma^2\rho}{\Omegaul(\widetilde{\bar{X}}-\Pc)}$, $\tilde{\delta}_b=\frac{\tilde{M}-\Pc}{\widetilde{\bar{X}}-\Pc}$ \footnote{Note that in $G_{\tilde{\delta}_b}$, $\frac{b}{\tilde{\delta}_b}=\frac{\gamma_{\rm thr}\sigma^2\rho}{\Omegaul(\tilde{M}-\Pc)}$ is always non-negative, since $\tilde{M}>\Pc$, by definition. Also, $b$ in $G$ is always non-negative, since $G=1-P_{\rm out}|_{K\to\infty,\,\tilde{\delta}=1}$, where at $\tilde{\delta}=1$, $\widetilde{\bar{X}}=\tilde{M}>\Pc$.}, and $\Sigma\definedas\Sigma_I$ for an infinite-size buffer and  $\Sigma\definedas\Sigma_F$ for a finite size buffer, where $\Sigma_I$ and $\Sigma_F$ are defined in the second and fourth rows of Table \ref{tab:outage_probabilities}, respectively.
Then, the on-off policy has a superior outage performance compared to the best-effort policy if $\exists\, \tilde{M}$ such that $P_{M,\rm on-off}\,G_{\tilde{\delta}_b}>P_{M,\rm best-effort}G_{\tilde{\delta}_b}+\Sigma$, where $P_M{}_{\rm,on-off}$ and $P_M{}_{\rm,best-effort}$ are the probabilities that the EH node transmits with power $M$ for the on-off and best-effort policies, respectively, as given in Table \ref{tab:outage_probabilities} for infinite- and finite-size energy buffers. Furthermore, for an infinite-size buffer, this implies that $P_{M,\rm on-off}\,G_{\tilde{\delta}_b}>G$ must also hold, which implies that $\tilde{\delta}=\tilde{M}/\widetilde{\bar{X}}>1$ yields the minimum outage probability for the on-off policy.
\label{prop:On_off_wins}
\end{proposition}\vspace{-0.3cm}
\begin{proof} 
From Table \ref{tab:outage_probabilities}, for both infinite- and finite-size energy buffers, the outage probabilities for the on-off and the best-effort policies can be written respectively as $P_{\rm out,on-off}=1-P_{M,\rm on-off}G_{\tilde{\delta}_b}$ and $P_{\rm out,best-effort}=1-\left(P_{M,\rm best-effort}G_{\tilde{\delta}_b}+\Sigma\right)$, where we used $\frac{\mul b}{\tilde{\delta}_b}=\frac{\Gamma_{\rm thr}\rho}{(\tilde{M}-\Pc)}$. Therefore, the on-off policy is superior to the best-effort policy if $\max\limits_{\tilde{\delta}}P_{M,\rm on-off}G_{\tilde{\delta}_b}>\!\max\limits_{\tilde{\delta}}P_{M,\rm best-effort}G_{\tilde{\delta}_b}+\Sigma$. In addition, for an infinite-size energy buffer, we know that $\tilde{\delta} \geq 1$ is optimal for both the best-effort and the on-off policies, cf. Corollary \ref{coro:delta1_opt_BE_INF}, and that $P_{\rm out,on-off}|_{\tilde{\delta}=1}=P_{\rm out,best-effort}|_{\tilde{\delta}=1}=1-G$, cf. Table \ref{tab:outage_probabilities}. Hence, for the on-off policy to be superior to the best-effort policy, $P_{\rm out,on-off}|_{\tilde{\delta}>1}<P_{\rm out,on-off}|_{\tilde{\delta}=1}$ also has to hold, i.e., $P_{M,\rm on-off}G_{\tilde{\delta}_b}>G$. This completes the proof. 
\end{proof}
\begin{remark}
We note that although the outage probability expressions obtained in Table \ref{tab:outage_probabilities} seem quite involved, they are easy to implement and simple to evaluate, cf. Remark \ref{rem:low_computation_energy_distributions}. Furthermore, the outage probability expressions depend on the statistical properties of the channel. Hence, for given statistical properties of the channel, the derived expressions can be used to optimize the system parameters, as will be shown in Section \ref{s:Simulations}. For example, the desired transmit power $M$ that minimizes the outage probability for a given buffer size can be optimized with just a single one-dimensional search for the optimal $\tilde{\delta}$. Equivalently, the system designer can determine the minimum size of the energy buffer that guarantees a certain outage performance without performing extensive simulations. 
\label{eq:usefulness_of_outage_expressions}
\end{remark}
\vspace{-0.3cm}
\section{Simulation and Numerical Results}
\label{s:Simulations}
In this section, we validate the analytical expressions obtained for the energy distribution, the outage probability, and the average throughput of the investigated energy management policies through simulations. The simulation parameters are summarized in Table \ref{tab:simulation_parameters}. We assume that the EH node stores the harvested energy in a micro-supercapacitor\footnote{It has been recently reported in \cite{supercapacitor_on_chip} that small-size micro-supercapacitors with a storage capacity of up to $K=\unit[0.2]{J}$ can be integrated on a chip  which makes them suitable storage devices for EH nodes.} of size $K$\footnote{Although the energy distribution obtained in Corollary \ref{coro:limiting_dist_Finite_Nakagami_exact_OO} for the on-off policy is only valid for $K=l\tilde{M}$ with $l=\{3,4,\ldots\}$, in this section, this condition will be violated when $\tilde{M}$ is varied for a given $K$. In this case, we approximate $l$ to $l={\rm round}(K/\tilde{M})$ only in the summation limits of Corollary \ref{coro:limiting_dist_Finite_Nakagami_exact_OO} and Table \ref{tab:outage_probabilities} but we use the exact value of $K$ elsewhere. This approximation is tight and provides results very close to the simulations.}.
 \begin{table}[!tp]
\caption{Simulation Parameters}\vspace{-0.2cm}
\begin{tabular}{@{}ll@{}}  \toprule   
Parameter & Value \\ \midrule 
DL and UL center frequencies  & $915\,$MHz and $2.45\,$GHz\\ 
Noise power at the AP &  $\sigma^2=-101\,$dBm \\ 
DL and UL path loss exponent & $2.5$ \\ 
AP and EH node antenna gains & $12\,$dBi and $3\,$dBi\\ 
Power amplifier inefficiency & $\rho=1.4$ \\ 
Storage efficiency & $\beta=0.9$ \\ 
RF-to-DC conversion efficiency & $\eta=0.5$\\ 
DL transmit power & $\Pdl=1\,$W \\ 
AP to EH node distance & $d=7\,$m in Fig. \ref{fig:Publication_fig_pdf_energy_distribution}, \ref{fig:Publication_fig_Outage_compare_with_optimal}, \ref{fig:Publication_fig_Outage_diff_Pc} and $d=12\,$m  in Fig.  \ref{fig:Publication_fig_Outage_compare_with_optimal_on_off_wins}--\ref{fig:Opt_Throughput_vs_R}.\\ 
Constant circuit power consumption & $\Pc=0$, except for Fig. \ref{fig:Publication_fig_Outage_diff_Pc}. \\ 
Average harvested energy & $\widetilde{\bar{X}}=\beta\bar{X}=\beta\eta\Pdl\Omegadl=1.2\times 10^{-5}\,$J  in Fig. \ref{fig:Publication_fig_pdf_energy_distribution}, \ref{fig:Publication_fig_Outage_compare_with_optimal}, \ref{fig:Publication_fig_Outage_diff_Pc} and $\widetilde{\bar{X}}=3.14\times 10^{-6}\,$J in Fig. \ref{fig:Publication_fig_Outage_compare_with_optimal_on_off_wins}--\ref{fig:Opt_Throughput_vs_R}.\\ 
DL and UL channel models & Nakagami-$m$ fading with $m\!=\!\mul\!=\!2$ in Fig. \ref{fig:Publication_fig_pdf_energy_distribution}, \ref{fig:Publication_fig_Outage_compare_with_optimal}, \ref{fig:Publication_fig_Outage_diff_Pc} and $m\!=\!\mul\!=\!3$ in Fig.  \ref{fig:Publication_fig_Outage_compare_with_optimal_on_off_wins}--\ref{fig:Opt_Throughput_vs_R}.\\ 
\bottomrule
\end{tabular}
\label{tab:simulation_parameters}
\vspace{-0.5cm}
\end{table}
\vspace{-0.3cm}
\subsection{Baseline Buffer-less Policy and the Optimal Policy}
We compare our proposed transmission policies with a baseline buffer-less transmission policy that uses all the energy harvested in a time slot for UL information transmission in the subsequent slot. Since the buffer-less policy performs only short-term energy storage (for one time slot), the harvested energy can be stored in a conventional capacitor which is characterized by almost perfect storage efficiency, i.e., $\beta\approx 1$. We note that as $M\to\infty$, the best-effort transmission policy described in (\ref{eq:Puwith_imperfections}) with $\beta=1$ tends to the buffer-less policy, i.e., $\Pul(i)=\frac{ [X(i-1)-\Pc]^+}{\rho},\,\forall\, i$. Hence, the outage probability of the buffer-less policy can be obtained from (\ref{eq:second_term_outage_BE}) for $\tilde{M}\to\infty$  and $g(x)\to f(x)$, where $f(x)\!=\!\frac{\lambda^m}{\Gamma(m)}x^{m-1}\e^{-\lambda x}$, which results in $P_{\rm out}\big|_{\rm buffer-less}=\int_0^{\Pc}f(x)\dd x+\int_{\Pc}^{\infty}\Pr\left(\hul<\frac{\gamma_{\rm thr}\sigma^2\rho}{x-\Pc}\right)f(x)\dd x=1-\int_{\Pc}^{\infty}\frac{\Gamma\left(\mul,\frac{\Gamma_{\rm thr}\rho}{x-\Pc}\right)}{\Gamma(\mul)}f(x)\dd x$. Using \cite[eq. 3.471.9]{table_of_integrals_Ryzhik}, we get 
\begin{equation}
P_{\rm out}\big|_{\rm buffer-less}\!=\!1-2\frac{\lambda^m}{\Gamma(m)}\e^{-\lambda\Pc}\!\!\sum\limits_{t=0}^{\mul-1}\!\!\frac{(\Gamma_{\rm thr}\rho)^t}{t!}\sum\limits_{l=0}^{m-1}\!\binom{\!m\!-\!1\!}{l}{\Pc}^{m-1-l}\left(\!\frac{\Gamma_{\rm thr}\rho}{\lambda}\!\right)^{\!\!\frac{l-t+1}{2}}\!\!\!\!K_{l-t+1}\left(\!2\sqrt{\Gamma_{\rm thr}\rho\lambda}\!\right),
\end{equation}
where $K_n(\cdot)$ is the modified Bessel function of the second kind and order $n$.

In addition, we compare our proposed policies to the optimal transmission policy in \cite{outage_minimization_Rui_Zhang_2014}, which minimizes the outage probability assuming no CSI knowledge at the EH node and infinite-size energy buffers. In particular, we use the optimal offline power allocation in \cite[Algorithm II]{outage_minimization_Rui_Zhang_2014} and implement it for sufficiently many time slots (e.g. $10^7$ slots) to ensure a fair comparison with the average outage probability of our proposed policies. We note that the  optimal power allocation algorithm in \cite{outage_minimization_Rui_Zhang_2014} assumes zero constant circuit power consumption. Hence, we use $\Pc=0$ to compare with the optimal policy. The effect of the storage efficiency $\beta$ and the power amplifier inefficiency $\rho$ can be easily incorporated into \cite[Algorithm II]{outage_minimization_Rui_Zhang_2014} by using $\widetilde{\bar{X}}=\beta \bar{X}$ as the average harvested energy and scaling down the optimal power allocation by $\rho$.
\vspace{-0.3cm}
\subsection{Numerical Results}
For all Figs. \ref{fig:Publication_fig_pdf_energy_distribution}--\ref{fig:Publication_fig_Outage_diff_Pc}, the analytical results in Corollaries  \ref{theo:stationary_dist_infinite_Nakagami}--\ref{coro:limiting_dist_Finite_Nakagami_exact_OO} for the energy distribution and in Table \ref{tab:outage_probabilities} for the outage probability are in perfect agreement with the simulated results. In Figs. \ref{fig:Publication_fig_pdf_energy_distribution} and \ref{fig:Publication_fig_Outage_compare_with_optimal}, we use a simulation set up with  $m\!=\!\mul\!=\!2$ and $d\!=\!7\,$m, cf. Table \ref{tab:simulation_parameters}. 
In Fig. \ref{fig:Publication_fig_pdf_energy_distribution}, we show the limiting distribution $g(x)$ of the energy stored in a finite-size energy buffer of size $K=\unit[0.05]{mJ}$ and an infinite-size energy buffer for both the best-effort and the on-off transmission policies. We use $\tilde{\delta}=1.04$, i.e., $\tilde{\delta}>1$, as otherwise, a stationary energy distribution does not exist for the infinite-size buffer case, cf. Theorem \ref{theo:no_stationary_dist}. For the finite-size buffer, we additionally plot the probability of a full buffer $\pi(K)$. For the on-off policy, on average, there is more energy in the buffer compared to the best-effort policy, since in the low-energy mode of operation, the on-off policy accumulates energy whereas the best-effort policy consumes all the energy in its buffer in a best-effort manner. This also  explains why the probability of buffer overflow, $\pi(K)$, is higher for the on-off policy compared to the best-effort policy as shown in Fig. \ref{fig:Publication_fig_pdf_energy_distribution}.

\begin{figure}[!t] 
\centering
\begin{minipage}[b]{.47\textwidth}
\centering
\includegraphics[width=1\textwidth, trim= {1cm 0 0cm 0.7cm},clip]{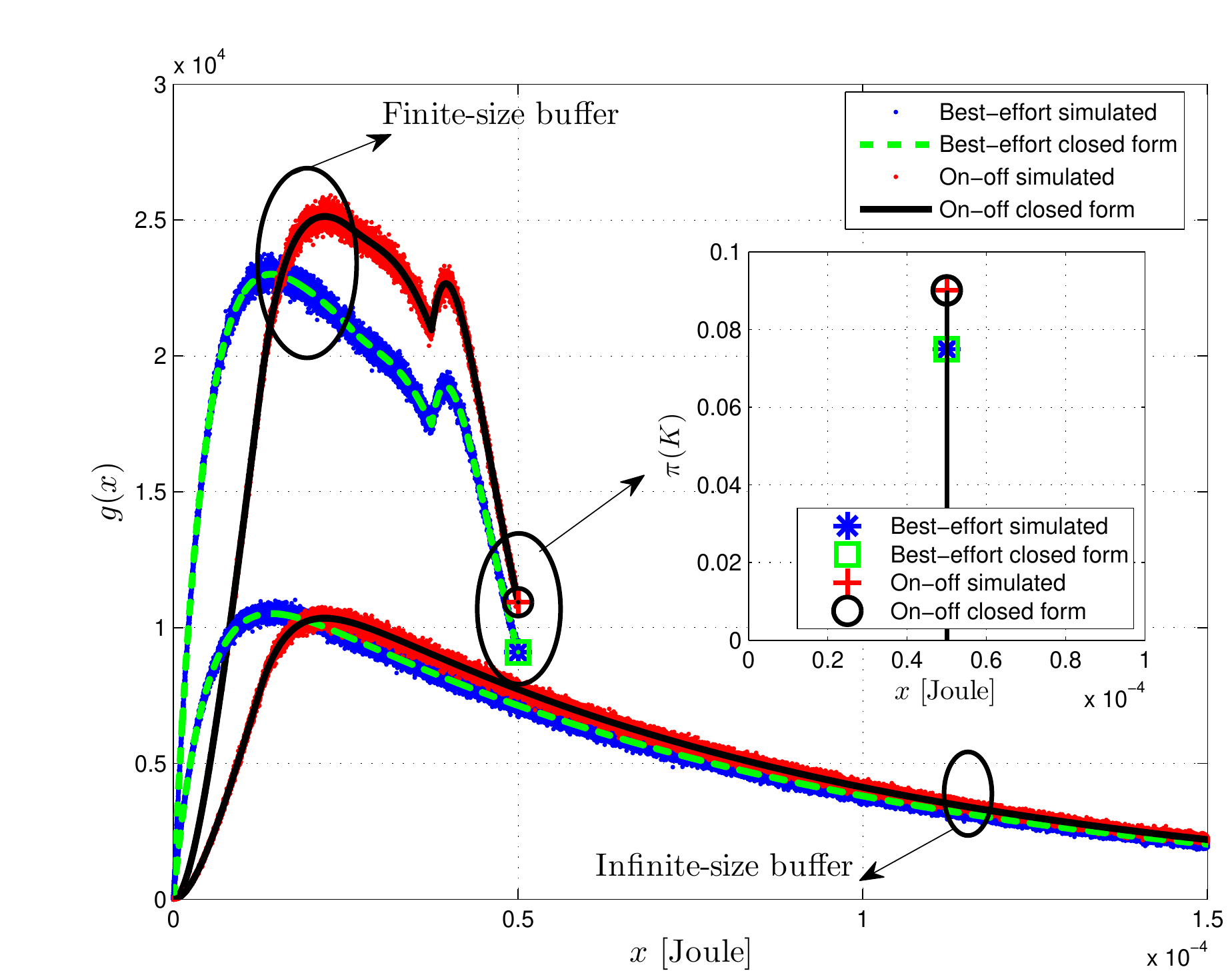}
\caption{The energy buffer distribution for $m\!=\!\mul\!=\!2$, $d\!=\!7\,$m, $\tilde{\delta}\!=\!1.04$, and $K\!=\!\unit[0.05]{mJ}$ .}
\label{fig:Publication_fig_pdf_energy_distribution}
\end{minipage}
\hspace{0.5cm}
\begin{minipage}[b]{0.47\textwidth}
\centering
\includegraphics[width=1\textwidth]{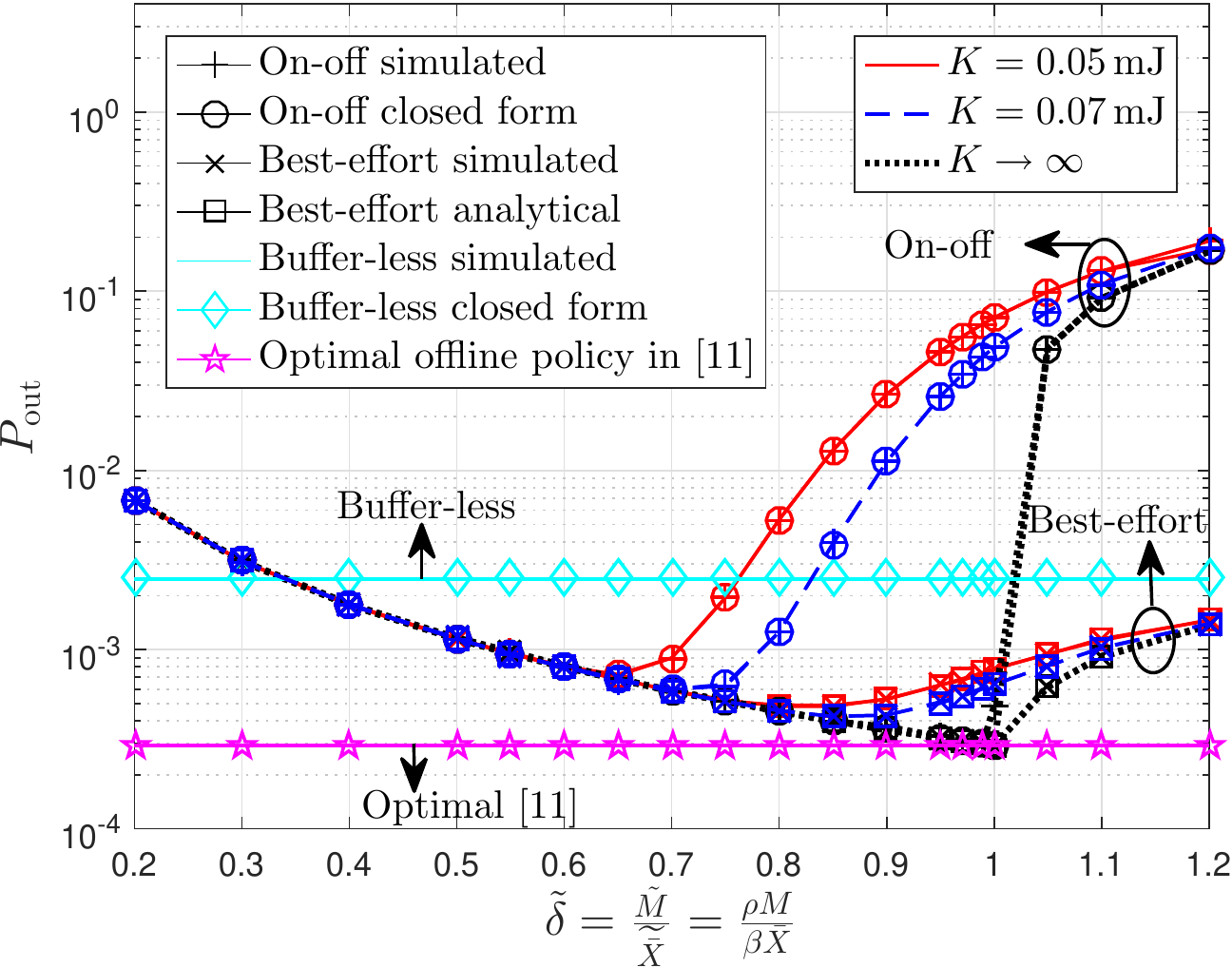}
\caption{Outage probability versus $\tilde{\delta}\!=\!\frac{\tilde{M}}{\widetilde{\bar{X}}}$ for  $m\!=\!\mul\!=\!2$, $d\!=\!7\,$m, and $R\!=\!2$ bits/(channel use).} 
\label{fig:Publication_fig_Outage_compare_with_optimal}
\end{minipage}\vspace{-0.7cm}
\end{figure}

In Fig. \ref{fig:Publication_fig_Outage_compare_with_optimal}, we plot the outage probability for a transmission rate of $R=\unit[2]{bits/(channel\,\, use)}$ and buffer sizes $K=\unit[0.05]{mJ}$, $K=\unit[0.07]{mJ}$, and $K\to\infty$. We change the maximum UL transmit power $M$, which results in different $\tilde{\delta}=\frac{\tilde{M}}{\widetilde{\bar{X}}}=\frac{\rho M+\Pc}{\beta \bar{X}}$ for the given  $\widetilde{\bar{X}}$. For both the best-effort and the on-off transmission policies, there exists an optimal value for the maximum UL transmit power $M=(\tilde{M}-\Pc)/\rho$ that minimizes the outage probability. The existence of such an optimal value can be explained based on (\ref{eq:Puwith_imperfections}) as follows. If $\tilde{M}$ is very small, then in most cases the desired amount of energy $\tilde{M}$ is available in the energy buffer and the UL transmit power is $M$ for both policies. In this case, increasing $\tilde{M}$, increases the UL transmit power and reduces the probability of outages. On the other hand, very large values of $\tilde{M}$ imply that for the best-effort policy, most of the time the harvested energy will be directly used in the subsequent time slot without buffering which results in a loss in the outage performance. For the on-off policy, a large value of $\tilde{M}$ means that in many time slots, the EH node remains silent although outage-free transmission is possible, i.e., if  $\tilde{M}\to\infty$, the outage probability approaches $1$. This explains why for the on-off policy, the outage probability increases significantly when $\tilde{M}$ is increased beyond its optimal value. At the considered transmission rate, it is observed that the optimal choice of $\tilde{M}$ is always less than or equal to the effective average harvested energy $\widetilde{\bar{X}}$ and increases with the capacity of the energy buffer,  i.e., optimally $\tilde{\delta}\leq 1$ and as $K\to\infty$, $\tilde{\delta}_{\rm opt}\to1$. For the considered setup, the optimal outage performance of the best-effort policy, for $K\to\infty$ and $\tilde{\delta}_{\rm opt}$, is superior to that of the on-off policy and it closely approaches the outage probability of the optimal offline power allocation policy in \cite{outage_minimization_Rui_Zhang_2014}.
\begin{figure}[!t] 
\begin{minipage}[b]{0.48\textwidth}
\centering
\includegraphics[width=1\textwidth]{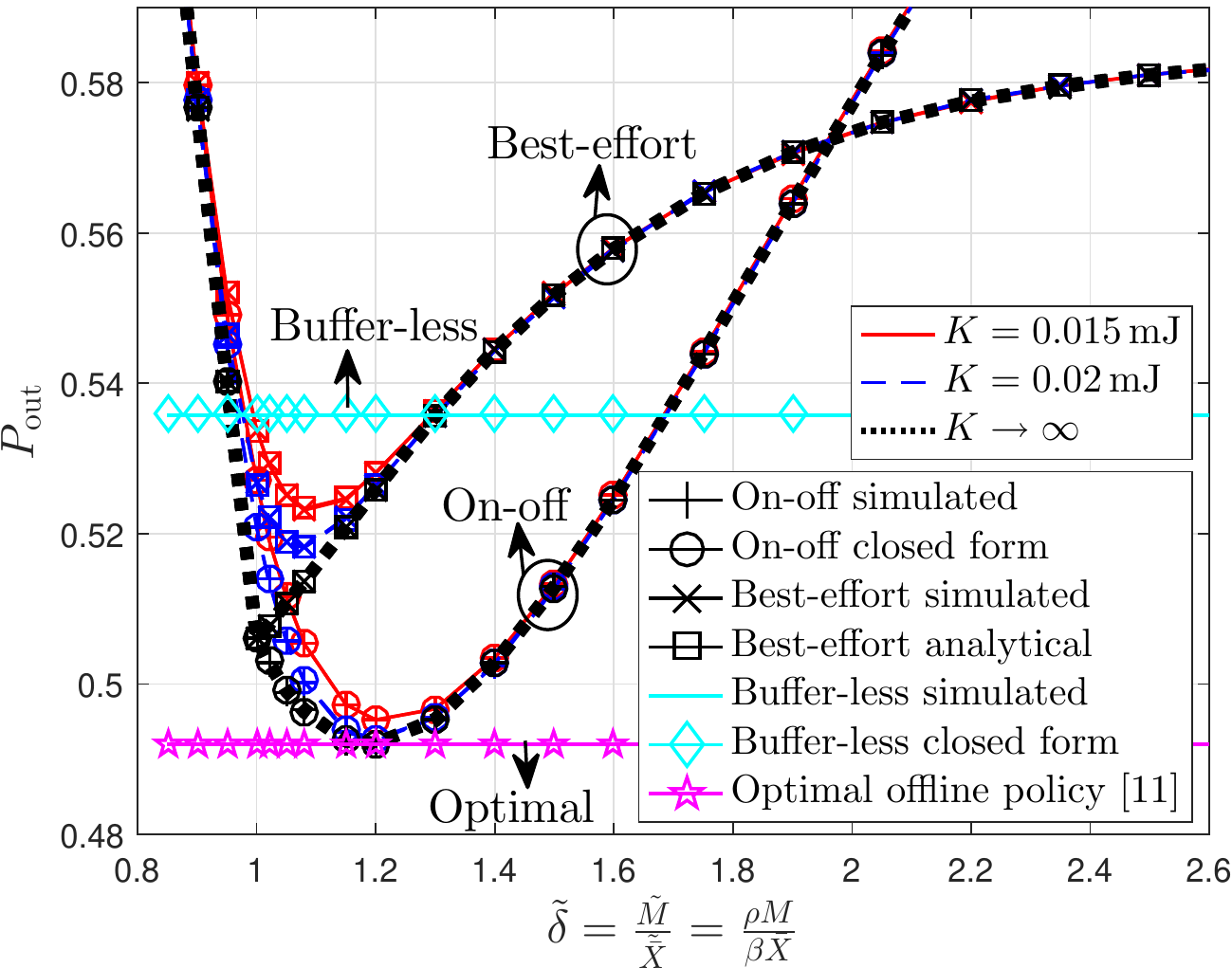}
\vspace{-0.2cm}
\caption{Outage Probability versus $\tilde{\delta}\!=\!\frac{\tilde{M}}{\widetilde{\bar{X}}}$ for  $m\!=\!\mul\!=\!3$, $d\!=\!12\,$m, and $R\!=\!4$ bits/(channel use).} 
\label{fig:Publication_fig_Outage_compare_with_optimal_on_off_wins}
\end{minipage}
\hspace{0.5cm}
\begin{minipage}[b]{0.46\textwidth}
\centering
\includegraphics[width=1\textwidth]{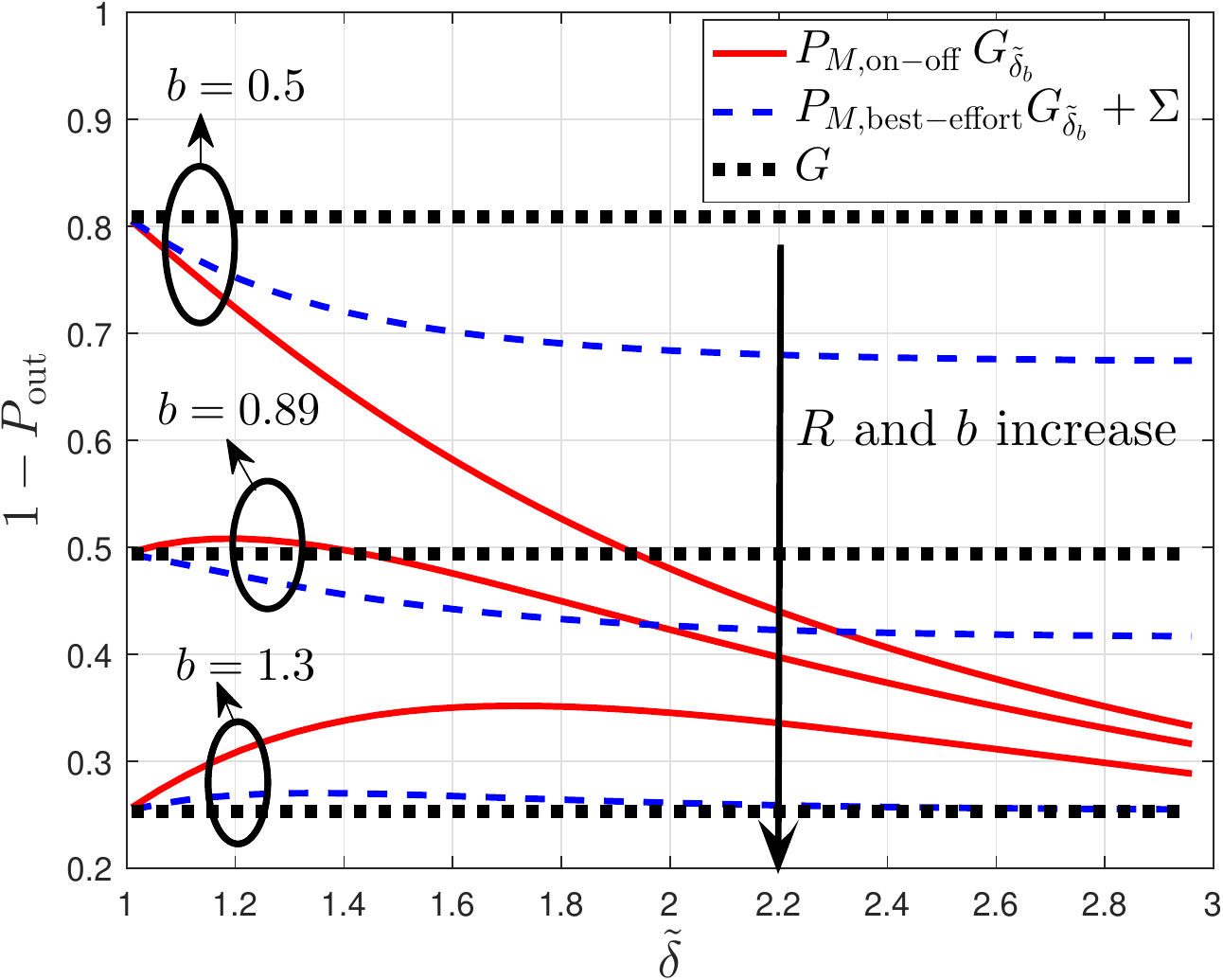}\vspace{-0.2cm}
\caption{$1\!-\!P_{\rm out}$ versus $\tilde{\delta}$, cf. Proposition \ref{prop:On_off_wins},  for  $m\!=\!\mul\!=\!3$, $d\!=\!12\,$m, and  different $b\!=\!\frac{\gamma_{\rm thr}\sigma^2\rho}{\Omegaul(\widetilde{\bar{X}}-\Pc)}$.}
\label{fig:When_on_off_wins}   
\end{minipage}
\vspace{-0.7cm}
\end{figure}

In Figs. \ref{fig:Publication_fig_Outage_compare_with_optimal_on_off_wins}--\ref{fig:Opt_Throughput_vs_R}, we use a simulation set up with  $m\!=\!\mul\!=\!3$ and $d\!=\!12\,$m, cf. Table \ref{tab:simulation_parameters}. In Fig. \ref{fig:Publication_fig_Outage_compare_with_optimal_on_off_wins}, we show the outage probability for a transmission rate of $R=\unit[4]{bits/(channel\,\, use)}$ and buffer sizes $K=\unit[0.015]{mJ}$, $K=\unit[0.02]{mJ}$, and $K\to\infty$. Unlike in Fig. \ref{fig:Publication_fig_Outage_compare_with_optimal}, in Fig. \ref{fig:Publication_fig_Outage_compare_with_optimal_on_off_wins}, the optimal choice of $\tilde{M}$ is always larger than or equal to the effective average harvested energy $\widetilde{\bar{X}}$ but decreases with the capacity of the energy buffer,  i.e., optimally $\tilde{\delta}\geq 1$ and as $K\to\infty$, $\tilde{\delta}_{\rm opt}$ decreases. Moreover, in the considered case, the optimal outage performance of the on-off policy, for $K\to\infty$ and $\tilde{\delta}_{\rm opt}$, is superior to that of the best-effort policy and it closely approaches the performance of the optimal power allocation policy in \cite{outage_minimization_Rui_Zhang_2014}. 
\begin{figure}[!thp]
\begin{minipage}[b]{0.46\textwidth}
\centering
\includegraphics[width=1\textwidth,height=0.27\textheight]{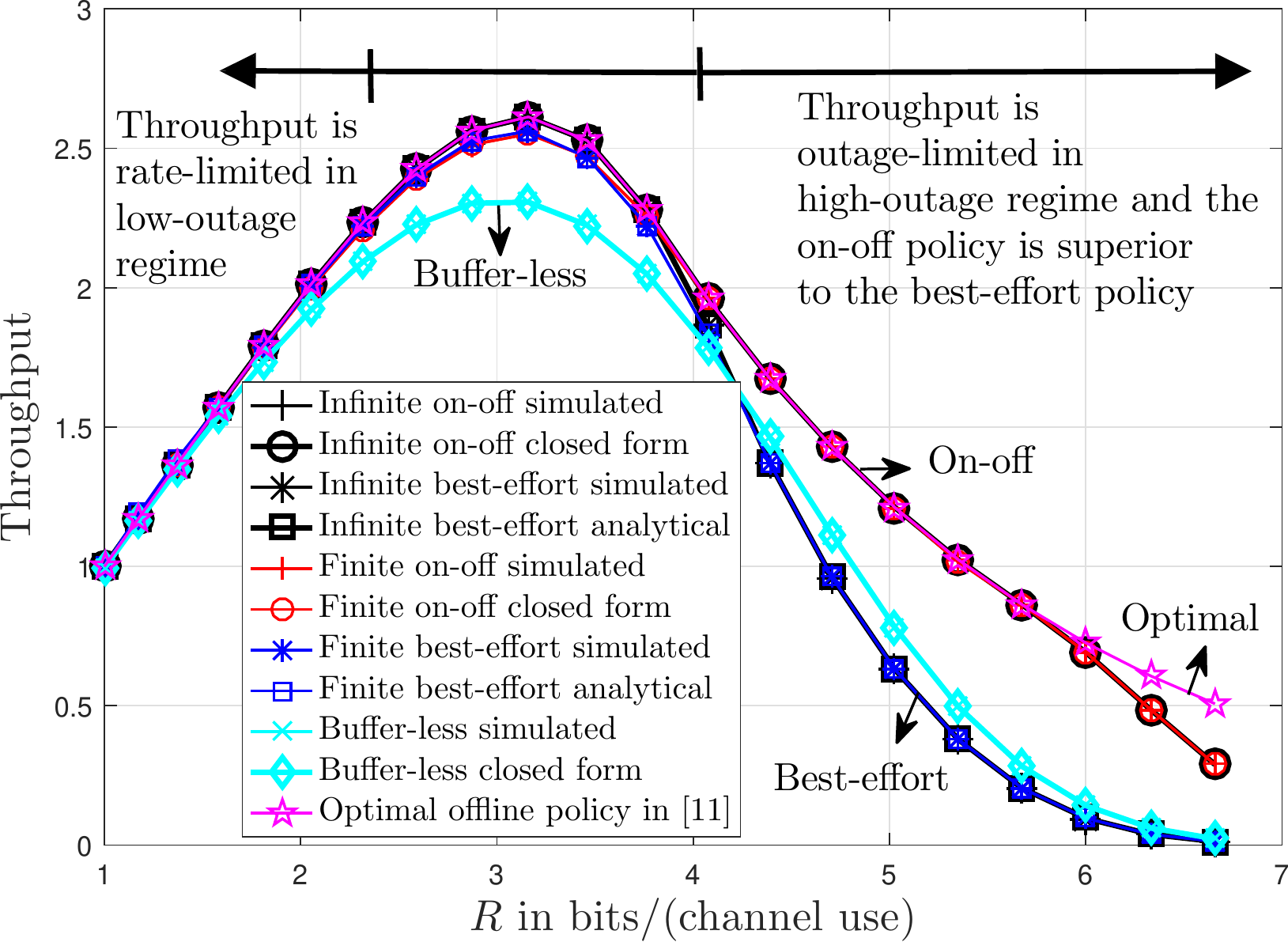}\vspace{-0.2cm}
\caption{Average throughput at optimal $\tilde{\delta}$ versus $R$ for $m\!=\!\mul\!=\!3$, $d\!=\!12\,$m, and $K\!=\!\unit[0.025]{mJ}$.}
\label{fig:Opt_Throughput_vs_R}  
\end{minipage}
\hspace{0.5cm}
\begin{minipage}[b]{0.46\textwidth}
\centering
\includegraphics[width=1\textwidth]{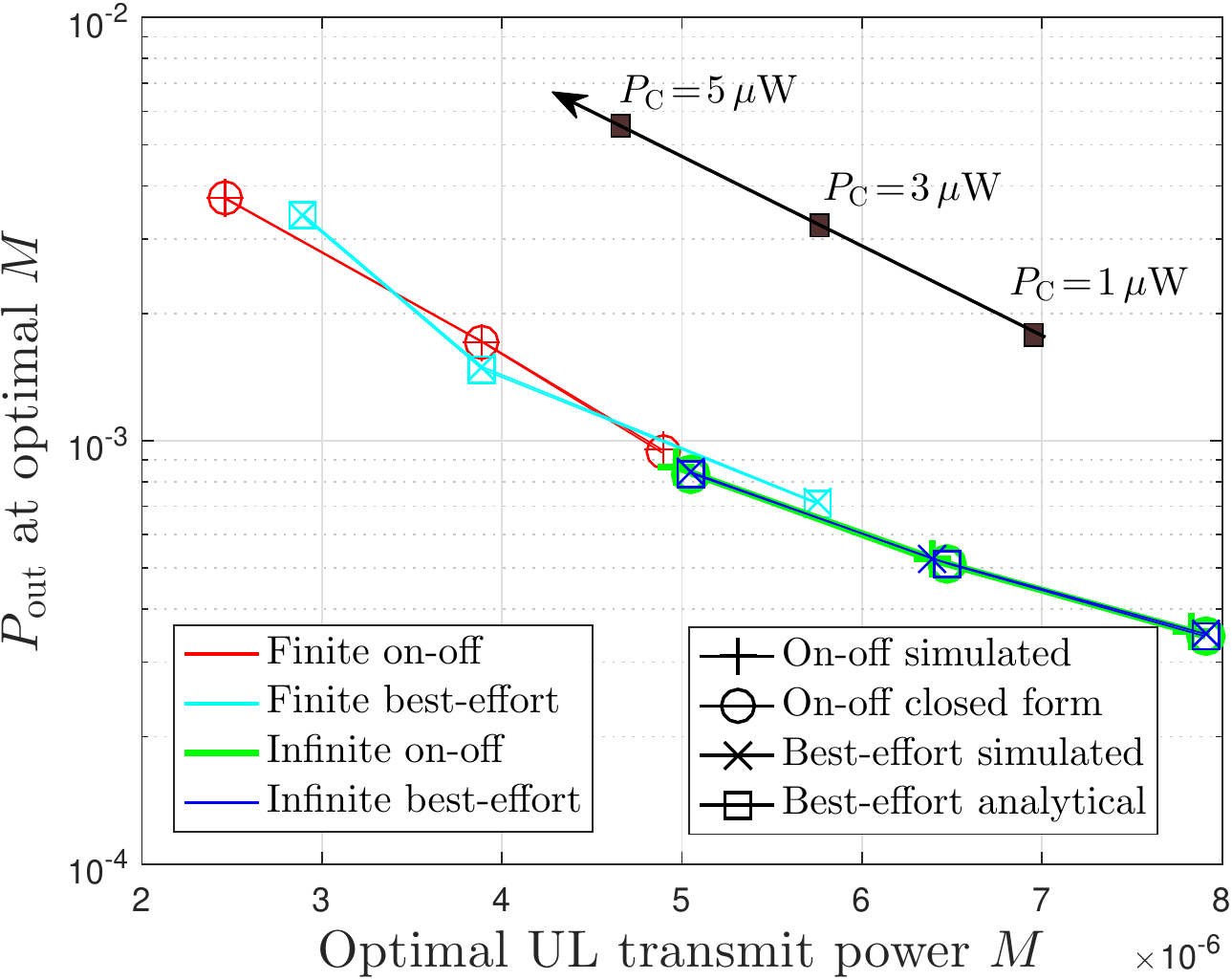}\vspace{-0.2cm}
\caption{Outage Probability at optimal $M$ for different $\Pc$ for $m\!=\!\mul\!=\!2$, $d\!=\!7\,$m, $R\!=\!2$ bits/(channel use), and $K\!=\!\unit[0.05]{mJ}$.}
\label{fig:Publication_fig_Outage_diff_Pc}  
\end{minipage}
\vspace{-0.7cm}
\end{figure}
From Fig. \ref{fig:Publication_fig_Outage_compare_with_optimal} and Fig. \ref{fig:Publication_fig_Outage_compare_with_optimal_on_off_wins}, we conclude that, for low outage probabilities, the best-effort policy is superior to the on-off policy and vice versa for high outage probabilities. This result is in agreement with \cite{outage_minimization_Rui_Zhang_2014}, where it was shown that on-off power allocation is superior to best-effort power allocation in the high-outage regime.
In Fig. \ref{fig:When_on_off_wins}, this behaviour is further explained by illustrating the result of Proposition \ref{prop:On_off_wins} for infinite-size buffers. In particular, we plot $1-P_{\rm out}$ given by  $P_{M,\rm on-off}\,G_{\tilde{\delta}_b}$ and $P_{M,\rm best-effort}G_{\tilde{\delta}_b}+\Sigma$ for the on-off and the best-effort policies, respectively, vs. $\tilde{\delta}\geq 1$. We also plot $G=1-P_{\rm out}\big|_{\tilde{\delta}=1}$. 
As can be observed, the larger the value of $b\!=\!\frac{\gamma_{\rm thr}\sigma^2\rho}{\Omegaul(\widetilde{\bar{X}}-\Pc)}$, the more likely it is that the on-off policy provides a better outage performance compared to the best-effort policy. However, a larger $b$ also implies a higher outage probability. Therefore, the on-off policy is more likely to outperform the best-effort policy for high outage probabilities, e.g. $P_{\rm out}>0.4$. We note that a large $b\!=\!\frac{\gamma_{\rm thr}\sigma^2\rho}{\Omegaul(\widetilde{\bar{X}}-\Pc)}$ may result from poor UL and/or DL channel conditions, low DL transmit power, large noise power, and/or high transmission rates. For example, in Fig. \ref{fig:Publication_fig_Outage_compare_with_optimal_on_off_wins}, the on-off policy is superior to the best-effort policy due to the larger transmission distance and the higher transmission rate compared to the case considered in Fig. \ref{fig:Publication_fig_Outage_compare_with_optimal}. In particular, we have $b=0.89$ in Fig. \ref{fig:Publication_fig_Outage_compare_with_optimal_on_off_wins}, in which case, the on-off policy has a better outage performance compared to the best-effort policy as shown in Fig. \ref{fig:When_on_off_wins}. 

In order to analyze the system performance at different transmission rates $R$, we plot in Fig. \ref{fig:Opt_Throughput_vs_R} the average throughput given by $T=R(1-P_{\rm out})$ obtained for the optimal choice of $\tilde{\delta}$ for each rate. First, we observe that there exists an optimal $R$ that maximizes the average throughput. This is because at low rates, the outage probability is low and the throughput is limited by $R$, i.e. $T\approx R$, for all policies. In contrast, at high rates, the throughput is limited by the high outage probability. Therefore, an optimal rate $R$ must exist for which the average throughput is maximized. In the medium-rate regime, the  the best-effort and the on-off policies result in a comparable throughput performance which   is superior to that of the buffer-less policy. Moreover, for the considered system parameters, there is only a slight loss in throughput performance if a finite-size buffer of size $\unit[0.025]{mJ}$ is used compared to an infinite-size buffer. In the high-rate regime, the on-off policy outperforms the best-effort policy and, except for very high rates, for $K\to\infty$, it closely approaches the performance of the optimal offline policy in \cite{outage_minimization_Rui_Zhang_2014}.

Finally, in Fig. \ref{fig:Publication_fig_Outage_diff_Pc}, we study the effect of a constant power consumption $\Pc$ on the optimal UL transmit power $M$ for the proposed transmission policies. As can be observed, the larger the power consumption $\Pc$, the lower the optimal  transmit power $M$ and the worse the outage performance.
\begin{remark}
From Figs. \ref{fig:Publication_fig_Outage_compare_with_optimal}, \ref{fig:Publication_fig_Outage_compare_with_optimal_on_off_wins}, and \ref{fig:Opt_Throughput_vs_R}, we conclude that the optimal performance of the simple online transmission policies proposed in (\ref{eq:Puwith_imperfections}) approaches the performance of the optimal offline power allocation policy in \cite{outage_minimization_Rui_Zhang_2014}. 
We note that in \cite{outage_minimization_Rui_Zhang_2014}, the optimal offline power allocation requires non-causal knowledge of the EH profile, whereas the optimal online algorithm in \cite{outage_minimization_Rui_Zhang_2014} requires causal EH profile knowledge and is based on dynamic programming which entails a high computational complexity for a continuous energy state space and a long transmission horizon. On the other hand, in each time-slot, our simple online transmission policies in (\ref{eq:Puwith_imperfections}) only require knowledge of whether the desired constant transmit power is available in the energy buffer. Furthermore, for given statistical properties of the channel, only  a \emph{single} one-dimensional search is required for the proposed policies to obtain the optimal constant transmit power $M$ using the analytical results in Table \ref{tab:outage_probabilities}.
\end{remark}
\vspace{-0.2cm}
\section{Conclusion}
\label{s:conclusion}
In this paper, we studied two simple online transmission policies for EH nodes with finite/infinite-size energy buffers that do not require instantaneous CSI nor EH profile knowledge. Using the theory of discrete-time continuous-state Markov chains, we analyzed the limiting distribution of the  energy stored in the buffer for a general i.i.d. EH process and obtained it in closed form for a Gamma distributed i.i.d. EH process. We have shown that buffering energy for future use may significantly improve the outage performance compared to directly consuming all the harvested energy without buffering. Furthermore, our results reveal that, for low outage probabilities, the best-effort policy has a superior outage performance compared to the on-off policy and vice versa for high outage probabilities. Moreover, using the derived analytical results for the outage probability, the UL transmit power of the EH node can be optimized and the minimum outage probability of the two proposed online policies is near-optimal and closely approaches the outage probability of the optimal offline power allocation policy. 
\appendices  
\ifARXIV
\section{Proof of Theorem \ref{theo:no_stationary_dist} (Non-Existence of a Stationary Energy Distribution for Infinite-Size Buffers with Best-Effort and On-Off Policies)}
\label{app:no_stationary_dist}
Setting $K\to\infty$ and taking the expectation of both sides of (\ref{eq:general_storage_equation}), we obtain \vspace{-0.2cm}
\begin{equation} 
\E[B(i+1)]-\E[B(i)]=\bar{X}-\E[\Pul(i)],\vspace{-0.2cm}
\label{eq:expection_storage_equation}
\end{equation}
where $\{B(i)\}$ represents the storage process of the best-effort policy in (\ref{eq:general_storage_equation_BE}) or that of the on-off policy in (\ref{eq:general_storage_equation_OO}). From (\ref{eq:Pul_policy_BE}) and (\ref{eq:Pul_policy_OO}), $\Pul(i) \leq M, \, \forall \, i \Rightarrow \E[\Pul(i)] \leq M$, hence from (\ref{eq:expection_storage_equation}), $\E[B(i+1)]-\E[B(i)]\geq \bar{X}\!-\!M$ follows.
If $M\!<\!\bar{X}$, then \vspace{-0.2cm}
\begin{equation}
\E[B(i+1)]>\E[B(i)]\vspace{-0.2cm}
\label{eq:accumulating_energy}
\end{equation}
must hold. That is, the mean of the process $\{B(i)\}$ changes (increases) with time, and therefore a stationary distribution for $\{B(i)\}$ does not exist. Furthermore, from (\ref{eq:accumulating_energy}), $\lim_{i\to \infty} \E[B(i)]=\infty$, i.e., the energy accumulates in the buffer. Hence, there must be some time slot $j$, after which for $i>j$, $B(i)> M$ a.s. Next, we prove by contradiction that $j$ must be finite. If $\Pul(j)=B(j)<M$ and $j\to \infty$, then $\lim\limits_{j\to\infty} \E[B(j)]<M$ which violates $\lim\limits_{i\to \infty} \E[B(i)]=\infty$. Hence, $j$ must be finite. This completes the proof.
\section{Proof of Theorem \ref{theo:stationary_dist_infinite} (Uniqueness of a Limiting Energy  Distribution for an Infinite-Size Buffer with Best-Effort and on-off policies)}
\label{app:stationary_dist_infinite_BE}
Consider first the best-effort policy. From Remark \ref{remark:equivalence_to_Morans_Model}, it can be observed that Moran's process $\{Z(i)\}$ is equivalent to the waiting time of a customer in a GI/G/1 queue \cite{asmussen2003applied}, where $X(i)$ is equivalent to the customer service time and $M$ is equivalent to the customers' inter-arrival time. Now, our storage process $\{B(i)\}$ in (\ref{eq:general_storage_equation_BE}) with $K\!\to\!\infty$ is equivalent to the process $U(i)=Z(i)+X(i)$. That is, $\{B(i)\}$ is equivalent to the sojourn time (waiting time plus service time) of a customer in a GI/G/1 queue. Since $\{Z(i)\}$ and $\{X(i)\}$ are independent and $\{X(i)\}$ is stationary, then the steady state behavior of $\{B(i)\}$ is solely governed by that of $\{Z(i)\}$. Hence, from \cite[Corollary 6.5 and Corollary 6.6]{asmussen2003applied}, $M>\bar{X}$ is a sufficient condition for the process $\{B(i)\}$ to possess a unique stationary distribution to which it converges in total variation from any initial distribution.\\
From the similarity between the energy storage process with the on-off policy in (\ref{eq:general_storage_equation_OO}) and the double service rate model by Gaver and Miller in  \cite[Section 3]{Gaver_Miller_1962}, cf. Remark \ref{remark:equivalence_onoff_to_Millers_model}, $M > \bar{X}$ guarantees the existence and uniqueness of the limiting distribution for the storage process $\{B(i)\}$ in (\ref{eq:general_storage_equation_OO}), see \cite[p.112]{Gaver_Miller_1962}.

Next, we prove that $\E[\Pul(i)]=\bar{X}$ holds when $M>\bar{X}$. From the law of conservation of energy flow in the buffer, $\E[\Pul(i)]\leq \bar{X}$ must always hold. Hence, from (\ref{eq:expection_storage_equation}), we obtain 
\begin{equation}
\E[B(i+1)]-\E[B(i)]\geq 0
\label{eq:law_conservation_flow}
\end{equation}
As mentioned in Appendix \ref{app:no_stationary_dist}, a stationary distribution may only exist if (\ref{eq:law_conservation_flow}) holds with equality. In this case $\E[\Pul(i)]=\bar{X}$ holds from (\ref{eq:expection_storage_equation}). This completes the proof.
\section{Proof of Theorem \ref{theo:integral_eq_infinite_BE} (Integral Equation of the Energy Distribution for an Infinite-Size Buffer with the Best-Effort Policy)}
\label{app:proof_BE_integral_eqn_infinite}
To understand the integral equation in (\ref{eq:Integral_eqn_infinite_BE}), we rewrite (\ref{eq:general_storage_equation_BE}) with $K\to\infty$ as 
\begin{equation}
B(i+1)=\begin{cases} X(i) & B(i)\leq M \\ B(i)-M+X(i) & B(i)>M \end{cases}.\vspace{-0.1cm}
\label{eq:Storage_eq_Infinite_u_x_BE}
\end{equation}
By setting $B(i)=u$ and $B(i+1)=x$, the relations $g(x|u\leq M)=f(x)$ and $g(x|u> M)=f(x-u+M)$ hold for the conditional limiting pdfs.  The upper limit ($M+x$) in the second integral in (\ref{eq:Integral_eqn_infinite_BE}) follows from the fact that $f(x-u+M)$ is non-zero only for a non-negative amount of harvested energy, i.e., $x-u+\!M\!\geq\!0$. These considerations lead to (\ref{eq:Integral_eqn_infinite_BE}). 

The integral equation for the limiting cdf in (\ref{eq:CDF_Integral_eqn_infinite_BE}) can be derived from (\ref{eq:Storage_eq_Infinite_u_x_BE}) as follows. Let $G_i(x)$ be the cdf of $B(i)$, then $G_{i+1}(x)$ is given by
$G_{i+1}(x)=\pr(B(i+1)\leq x)\!=\!\pr\left(X(i)\!\leq\! x \middle| B(i)\!\leq \!M\right)+\pr\left(B(i)-M+X(i)\leq x \middle| B(i)> M\right)=\int_{\nu=0}^M F(x) \dd G_i(\nu)+\int_{\nu=M}^{M+x}F(x+M-\nu)\dd G_i(\nu).$
In the steady state of the storage process, i.e., as $i\to \infty$, $G_{i}(x)=G_{i+1}(x)=G(x)$. Furthermore, $G(x)$ can be simplified by performing integration by parts for its second term to get
$G(x)=F(x)G(M)+F(x+M-\nu)G(\nu)\Big|_{\nu=M}^{M+x}-\int_{\nu=M}^{M+x}G(\nu)\dd F(x+M-\nu)=-\int_{\nu=M}^{M+x}G(\nu)\dd F(x+M-\nu)$
where we used $F(0)=0$. By substituting $u=x+M-\nu$, $G(x)$ reduces to (\ref{eq:CDF_Integral_eqn_infinite_BE}).
\fi
\ifARXIV
\section{Proof of Corollary \ref{theo:stationary_dist_infinite_Nakagami} (Energy Distribution for an Infinite-Size Buffer with the Best-Effort Policy and Gamma-Distributed EH Process)}
\else
\section{Proof of Corollary \ref{theo:stationary_dist_infinite_Nakagami} (Best-Effort Policy with Infinite-Size Buffer)}
\fi

\label{app:stationary_dist_infinite_Nakagami_BE}
When $M>\bar{X}$, i.e., $\delta>1$, we know from Theorem \ref{theo:stationary_dist_infinite} that the integral equation in (\ref{eq:CDF_Integral_eqn_infinite_BE}) has a unique solution for $G(x)$. Similar to  \cite[eq. (11)]{Infinite_dam_Gani_Prabhu_1957}, we postulate a solution of the type $G(x)=1-\sum\limits_{n=0}^{m-1}c_n\e^{-\lambda_n x}$,
where the unit term in $G(x)$ ensures that $G(\infty)=1$. To obtain $\lambda_n$ and $c_n$, we substitute with the postulated $G(x)$ in (\ref{eq:CDF_Integral_eqn_infinite_BE}) and use $\dd F(u)=\frac{\lambda^m}{\Gamma(m)}u^{m-1}\e^{-\lambda u}\dd u$, then (\ref{eq:CDF_Integral_eqn_infinite_BE})  reduces to \small
\begin{equation}
1-\sum\limits_{n=0}^{m-1}c_n\e^{-\lambda_n x}=\int\limits_0^x \left(1-\sum\limits_{n=0}^{m-1}c_n\e^{-\lambda_n (x+M-u)}\right)\frac{\lambda^m}{\Gamma(m)}u^{m-1}\e^{-\lambda u}\dd u.
\label{eq:CDF_as_sum_exps_proof}
\end{equation} \normalsize
Using $\int\limits_0^x u^{m-1}\e^{(\lambda_n-\lambda)u}\dd u=(\lambda-\lambda_n)^{-m}\gamma\left(m,x(\lambda-\lambda_n)\right)=\frac{(m-1)!}{(\lambda-\lambda_n)^{m}}\left(1-\e^{(\lambda_n-\lambda)x}\sum\limits_{s=0}^{m-1}\frac{(\lambda-\lambda_n)^s}{s!} x^s\right)$, see \cite[Eq. 3.381.1]{table_of_integrals_Ryzhik}, we get
\small
 \begin{equation}
\sum\limits_{n=0}^{m-1}\!c_n\e^{-\lambda_n x}=\!\sum\limits_{n=0}^{m-1}\!c_n\e^{-\lambda_n x} \left(\frac{\lambda}{\lambda-\lambda_n}\right)^m \e^{-\lambda_n M}
+ \e^{-\lambda x}\sum\limits_{s=0}^{m-1}\!\! x^s \left(\frac{\lambda^s}{s!}\!-\!\sum\limits_{n=0}^{m-1}c_n \frac{(\lambda-\lambda_n)^s}{s!}  \left(\!\frac{\lambda}{\lambda-\lambda_n}\!\right)^m \!\e^{-\lambda_n M}\right),
\label{eq:CDF_as_sum_exps_proof2}
\end{equation}
\normalsize
which holds if the coefficient of $\e^{-\lambda_n x}$ is $c_n$, i.e., if $\left(\frac{\lambda}{\lambda-\lambda_n}\right)^m \e^{-\lambda_n M}=1$ and if the coefficient of $x^s \e^{-\lambda x}$ is zero $\forall \,s$, i.e., $\lambda^s-\sum_{n=0}^{m-1}c_n (\lambda-\lambda_n)^s=0, \,\forall s=0,\ldots,m-1$. The first condition $\left(\frac{\lambda}{\lambda-\lambda_n}\right)^m \e^{-\lambda_n M}=1$ has $m$ roots for $\lambda_n$. These can be obtained by writing the condition as $\left(\frac{\lambda}{\lambda-\lambda_n}\right)^m \e^{-\lambda_n M}=\e^{\J 2\pi n}$. After taking the $m^{\text{th}}$ root of both sides and multiplying by $-\frac{M}{m} \e^{-\frac{\lambda  M}{m}}=-\frac{\delta}{\lambda}\e^{-\delta}$, the condition reduces to $-\delta \e^{-\delta}\e^{-\J\frac{2\pi n}{m}}=(\lambda_n-\lambda)\frac{M}{m}\e^{(\lambda_n-\lambda)\frac{M}{m}}$ whose solution is given by the Lambert W function, namely $(\lambda_n-\lambda)\frac{M}{m}=W_0\left(-\delta \e^{-\delta}\e^{-\J\frac{2\pi n}{m}}\right)$, $\forall n=0,\ldots,m-1$. Having obtained $\lambda_n$, we can obtain $c_n,\,n=0,\ldots,m-1$ by solving the second condition, namely $\sum_{n=0}^{m-1}c_n \left(\frac{\lambda-\lambda_n}{\lambda}\right)^s=1$, $\forall s=0,\ldots,m-1$, which corresponds to a non-homogeneous system of linear equations in $c_n$. Define $\V{c}=[c_0,\ldots,c_{m-1}]^T$ and matrix $\V{A}$ with $A_{sn}=\left(\frac{\lambda-\lambda_n}{\lambda}\right)^s$ as its entry in the $s^{\text{th}}$ row and the $n^{\text{th}}$ column, then $\V{A}\V{c}=\V{1}$, i.e., $\V{c}=\V{A}^{-1}\V{1}$. Finally, the limiting pdf $g(x)$ is obtained by differentiating the cdf $G(x)$. This completes the proof. 
\ifARXIV
\section{Proof of Theorem \ref{theo:OO_integral_eqn_infinite} (Integral Equation of the Energy Distribution for an Infinite-Size Buffer with the On-Off Policy)}
\label{app:proof_OO_integral_eqn_infinite}
To understand the integral equation in (\ref{eq:g_integral_eqn_infinite_OO}), one may set $B(i)=u$ and $B(i+1)=x$, then (\ref{eq:general_storage_equation_OO}) reads\vspace{-0.1cm}
\begin{equation}
x=\begin{cases} u+X(i), & u\leq M \\ u-M+X(i), & u>M \end{cases}.\vspace{-0.1cm}
\label{eq:Storage_eq_Infinite_u_x}
\end{equation}
Thus, for $u\leq M$, $g(x|u\leq M)=f(x-u)$ which is non-zero only for a non-negative amount of harvested energy, i.e., for $u\leq x$. Hence, in the range $u\leq M$, the upper limit on $u$ in the first integral of (\ref{eq:Integral_eqn_infinite_parta_OO}) and (\ref{eq:Integral_eqn_infinite_partb_OO}) is given by $u=\min(x,M)$. The second integral in (\ref{eq:Integral_eqn_infinite_parta_OO}) and (\ref{eq:Integral_eqn_infinite_partb_OO}) corresponds to the range $u>M$, where from (\ref{eq:Storage_eq_Infinite_u_x}), $g(x|u> M)=f(x-u+M)$ which is non-zero only for a non-negative amount of harvested energy, i.e., for $u\leq M+x$. This completes the proof. 
\fi
\ifARXIV
\section{Proof of Corollary \ref{theo:stationary_dist_infinite_Nakagami_OO} (Energy Distribution for an Infinite-Size Buffer with the On-Off Policy and Gamma-Distributed EH Process)}
\label{app:stationary_dist_infinite_Nakagami_OO}
Define $g_1(x)\definedas g(x), \, 0\!\leq \!x<\!M$, and $g_2(x)\definedas g(x), \, x\!>\!M$. Similar to the approach in Appendix \ref{app:stationary_dist_infinite_Nakagami_BE}, we postulate a solution of the form $g_2(x)=\sum\limits_{n=0}^{m-1} \lambda_n c_n\e^{-\lambda_n x}$, see Corollary \ref{theo:stationary_dist_infinite_Nakagami}.  Substituting $g_2(x)$ in (\ref{eq:Integral_eqn_infinite_partb_OO}) and using the pdf of the Gamma distribution $f(x)=\frac{\lambda^m}{(m-1)!}\e^{-\lambda x}x^{m-1}$, we get
\begin{equation}
\begin{aligned}
\sum\limits_{n=0}^{m-1} \lambda_n c_n\e^{-\lambda_n x}=&\int\limits_{u=0}^M\frac{\lambda^m}{(m-1)!}\e^{-\lambda(x-u)}(x-u)^{m-1}g_1(u)\dd u\\
&+\sum\limits_{n=0}^{m-1} \lambda_n c_n \Big(\frac{\lambda}{\lambda-\lambda_n}\Big)^m\e^{-\lambda_n M}\Bigg[\e^{-\lambda_n x}-\e^{-\lambda x}\sum\limits_{s=0}^{m-1}\frac{(\lambda-\lambda_n)^s}{s!}x^s\Bigg],
\end{aligned}
\label{eq:integral_eqn_g2_postulation_OO}
\end{equation}
where we used $\int_{u=M}^{M+x}\e^{(\lambda-\lambda_n)u}(x-u+M)^{m-1}\dd u=\e^{(\lambda-\lambda_n)(x+M)}(\lambda-\lambda_n)^{-m}\gamma\Big(m,(\lambda-\lambda_n)x\Big)$, see \cite[Eq. 3.382.1]{table_of_integrals_Ryzhik},   and $\gamma\Big(m,(\lambda-\lambda_n)x\Big)=(m-1)!\Big[1-\e^{-(\lambda-\lambda_n)x}\sum\limits_{s=0}^{m-1}\frac{(\lambda-\lambda_n)^sx^s}{s!}\Big]$.
Eq. (\ref{eq:integral_eqn_g2_postulation_OO}) holds if the coefficient of $\e^{-\lambda_n x}$ is $\lambda
_nc_n$, i.e., if $\left(\frac{\lambda}{\lambda-\lambda_n}\right)^m \e^{-\lambda_n M}=1$, which leads to $(\lambda_n-\lambda)\frac{M}{m}=W_0\left(-\delta \e^{-\delta}\e^{-\J\frac{2\pi n}{m}}\right)$, $\forall n=0,\ldots,m-1$, see Appendix \ref{app:stationary_dist_infinite_Nakagami_BE}, and if the coefficient of $x^s \e^{-\lambda x}$ is zero $\forall \,s$, i.e.,
\begin{equation}
\frac{\lambda^m (-1)^{m-1-s}}{(m-1-s)!}\underbrace{\int\limits_{u=0}^M\e^{\lambda u}u^{m-1-s}g_1(u)\dd u}_{I_s}-\sum\limits_{n=0}^{m-1}\lambda_nc_n(\lambda-\lambda_n)^s=0, \quad \quad \forall s=0,\ldots,m-1,
\label{eq:cn_homogeneous_eqns_OO}
\end{equation}
where we used the binomial expansion $(x-u)^{m-1}=\sum\limits_{s=0}^{m-1}\binom{m-1}{s}(-u)^{m-1-s}x^s$ in the first term in (\ref{eq:integral_eqn_g2_postulation_OO}). In order to obtain a unique solution for $c_n$, we need to use the unit area condition $\int_0^M g_1(u)\dd u+\int_M^{\infty} g_2(u)\dd u=1$ and add it to (\ref{eq:cn_homogeneous_eqns_OO}) to get a non-homogeneous system of linear equations in $c_n$.  Hence, in the following, we first obtain $g_1(x)$, and use it to calculate the integral $I_s$ in (\ref{eq:cn_homogeneous_eqns_OO}) as well as the unit area condition in terms of $c_n$ and finally we solve for $c_n$. 
\subsection{Expressing $g_1(x)$ in Terms of $c_n$}
Substituting the pdf of the Gamma distribution and $g_2(x)$ in (\ref{eq:Integral_eqn_infinite_parta_OO}), we get the following integral equation for $g_1(x)$
\begin{equation}
g_1(x)=\int\limits_{u=0}^x\frac{\lambda^m}{(m-1)!}\e^{-\lambda(x-u)}(x-u)^{m-1}g_1(u)\dd u\\
+\underbrace{\sum\limits_{n=0}^{m-1} \lambda_n c_n \Bigg[\e^{-\lambda_n x}-\e^{-\lambda x}\sum\limits_{t=0}^{m-1}\frac{(\lambda-\lambda_n)^t}{t!}x^t\Bigg]}_{T(x)},
\label{eq:volterra_integral_eqn_g1_OO}
\end{equation}
where $T(x)$ is identical to the second term in (\ref{eq:integral_eqn_g2_postulation_OO}) since the second terms in (\ref{eq:Integral_eqn_infinite_parta_OO}) and (\ref{eq:Integral_eqn_infinite_partb_OO}) are identical. Eq. (\ref{eq:volterra_integral_eqn_g1_OO}) is a Volterra integral equation of the second kind, whose solution is given by \cite[eq. 2.2-2.31]{polyanin2008handbook}
\begin{equation}
g_{1}(x)=T(x)+\underbrace{\int\limits_{u=0}^{x}R(x-u)T(u)\dd u}_{I_R(x)},
\label{eq:g_l_1_OO_INF}
\end{equation}
where 
\begin{equation}
R(x)=\frac{\e^{-\lambda x}}{m}\sum\limits_{k=0}^{m-1}\lambda\e^{\lambda x\cos(\eta_k)}\cos\Big(\eta_k+\lambda x\sin(\eta_k)\Big)
\label{eq:R_x_OO_INF}
\end{equation}
and $\eta_k\definedas\frac{2\pi k}{m}$. Hence, $I_R(x)$ can be written as
\begin{equation}
I_R(x)=\sum\limits_{n=0}^{m-1} \lambda_n c_n \Bigg[\underbrace{\int\limits_{u=0}^x R(x-u)\e^{-\lambda_n u}\dd u}_{I_n(x)}-\sum\limits_{t=0}^{m-1}\frac{(\lambda-\lambda_n)^t}{t!}\underbrace{\int\limits_{u=0}^x R(x-u)\e^{-\lambda u} u^t\dd u}_{I_t(x)} \Bigg],
\label{eq:IRx_INF_OO}
\end{equation}
and $g_1(x)$ in (\ref{eq:g_l_1_OO_INF}) reduces to
\begin{equation}
g_1(x)=\sum\limits_{n=0}^{m-1} \lambda_n c_n \Bigg[\e^{-\lambda_n x}+I_n(x)-\sum\limits_{t=0}^{m-1}\frac{(\lambda-\lambda_n)^t}{t!}\Bigg(\e^{-\lambda x}x^t+I_t(x)\Bigg)\Bigg].
\label{eq:g_l_1_OO_INF2}
\end{equation}
Using $R(x)$ in (\ref{eq:R_x_OO_INF}), Both $I_t(x)$ can be written as
\begin{equation}
I_t(x)=\sum\limits_{k=0}^{m-1}\frac{\lambda\e^{-\lambda x(1-\cos(\eta_k))}}{m}\underbrace{\Big[\cos(\mu_{k,x})I_{\cos}(x)+\sin(\mu_{k,x})I_{\sin}(x)\Big]}_{I_{\rm total}},
\label{eq:I_g}
\end{equation}
where $\mu_{k,x}=\eta_k+\lambda x\sin(\eta_k)$, $I_{\cos}(x)=\int_{u=0}^{x}\e^{(-\lambda\cos(\eta_k))u}\cos\Big(\lambda\sin(\eta_k)u\Big)u^{t}\dd u=\Re\{I_{\text{C}}\}$,
and 
$I_{\sin}(x)=\int_{u=0}^{x}\e^{(-\lambda\cos(\eta_k))u}\sin\Big(\lambda\sin(\eta_k)u\Big)u^{t}\dd u=-\Im\{I_{\text{C}}\}$,
where we used \cite[eq. 3.944.3]{table_of_integrals_Ryzhik} and \cite[eq. 3.944.1]{table_of_integrals_Ryzhik} with $I_{\text{C}}=(\lambda \e^{\J\eta_k})^{-(t+1)}\gamma\Big(t+1,(\lambda \e^{\J\eta_k})x\Big)$. Therefore, $I_{\rm total}$ in (\ref{eq:I_g})  can be written as $I_{\rm total}=\Re\{\e^{\J\mu_{k,x}}I_{\text{C}}\}$, and $I_t$ reduces to
$I_t=\frac{1}{m}\sum_{k=0}^{m-1}\Re\Big\{\e^{-\theta_k x}(\lambda\e^{\J\eta_k})^{-t}\gamma(t+1,\lambda x\e^{\J\eta_k})\Big\}$, where $\theta_k\definedas \lambda(1-\e^{\J\eta_k})$.

Moreover, $I_n(x)$ has the same form as (\ref{eq:I_g}) but with $I_{\cos}(x)\!=\!\!\int\limits_{0}^{x}\e^{-(\lambda_n\!-\!\lambda\!+\!\lambda\!\cos\!\!(\!\eta_k))u}\cos\big(\lambda\sin(\eta_k)u\big)\dd u$ $=I_A\!+I_B$, and $I_{\sin}(x)=\int_{u=0}^{x}\e^{-(\lambda_n-\lambda+\lambda\cos(\eta_k))u}\sin\big(\lambda\sin(\eta_k)u\big)\dd u=\J(I_A- I_B)$, with $I_A=[1-\e^{-(\lambda_n-\theta_k)x}]/[2(\lambda_n-\theta_k)]$ and $I_B=[1-\e^{-(\lambda_n-\theta_k^*)x}]/[2(\lambda_n-\theta_k^*)]$. Hence, $I_{\rm total}$ in (\ref{eq:I_g}) reduces to $I_{\rm total}=I_A\e^{\J\mu_{k,x}}\!+\!I_B\e^{-\J\mu_{k,x}}$ and $I_n(x)$ reduces to $I_n(x)\!=\!\frac{\lambda}{m}\sum_{k=0}^{m-1}\big(a_{nk}\e^{-\theta_k x}\!+\!b_{nk}\e^{-\theta_k^* x}\!-\!(a_{nk}+b_{nk})\e^{-\lambda_n x}\big)$ with $a_{nk}\!=\![2\big(\lambda_n\e^{-\J\eta_k}\!+\!\theta_k^*\big)]^{-1}$ and  $b_{nk}\!=\![2\big(\lambda_n\e^{\J\eta_k}\!+\!\theta_k\big)]^{-1}$. We note that $I_A\!\neq\! I_B^*$ and $a_{nk}\!\neq\! b_{nk}^*$  since $\lambda_n$ is complex-valued. Substituting with $I_n(x)$ and $I_t(x)$ in (\ref{eq:g_l_1_OO_INF2}), $g_1(x)$ reduces to (\ref{eq:g1_INF_OO}).

\vspace{-0.3cm}\subsection{Defining a Non-Homogeneous System of Linear Equations in $c_n$}
Using $g_1(x)$ in (\ref{eq:g1_INF_OO}), $I_s$ in (\ref{eq:cn_homogeneous_eqns_OO})  reduces to 
\begin{equation}
\begin{aligned}
I_s&=\int\limits_{u=0}^M\e^{\lambda u}u^{m-1-s}g_1(u)\dd u=\sum\limits_{n=0}^{m-1}\!\!\lambda_n c_n \Biggg[\!I_1\!+\frac{\lambda}{m}\sum\limits_{k=0}^{m-1}\Big(a_{nk}I_2+b_{nk}I_3-(a_{nk}+b_{nk})I_4\Big)\\[-1ex]
&-\!\!\sum\limits_{t=0}^{m-1}\!\frac{(\lambda-\lambda_n)^t}{t!}\Bigg\{\!I_5+\sum\limits_{k=0}^{m-1}\Re\Big\{\frac{(\lambda \e^{\J\eta_k})^{-t}t!}{m}\Big(\!I_6-\sum\limits_{q=0}^{t}\!\!\frac{(\lambda\e^{\J\eta_k})^q}{q!}I_7\Big)\Big\}\Bigg\}\Biggg].
\end{aligned}
\label{eq:Integral_in_non_homegenous_eqns}
\end{equation}
Using \cite[Eq. 3.381.1]{table_of_integrals_Ryzhik} and defining $I(\beta)=\int_{u=0}^M\e^{-\beta u}u^{m-1-s}\dd u=\beta^{-m+s}\gamma(m-s,\beta M)$, then $I_1=I_4=I(\lambda_n-\lambda)$, $I_2=I_3^*=I(\theta_k-\lambda)$,  $I_5=\int_{u=0}^{M}u^{m+t-s-1}\dd u=\frac{M^{m+t-s}}{(m+t-s)}$, $I_6=I(-\lambda\e^{\J\eta_k})$, and $I_7=I_5$ for $t=q$. Furthermore, we scale  (\ref{eq:cn_homogeneous_eqns_OO}) by $1/\lambda^{s+1}$  in order to improve the accuracy of the inversion of matrix $\V{A}$ in Corollary \ref{theo:stationary_dist_infinite_Nakagami_OO}. If we define (\ref{eq:cn_homogeneous_eqns_OO}) after scaling as $\sum_{n=0}^{m-1}B_{sn}c_n=0$ then $B_{sn}$ reduces to (\ref{eq:B_sn_OO}).

Finally, to solve for $c_n$, we need to add the unit area condition to (\ref{eq:cn_homogeneous_eqns_OO}). The unit area condition is given by $\int_0^M g_1(x)\dd x+\int_M^{\infty} g_2(x)\dd x=1$. Using $\Delta g_{12}(x)\definedas g_1(x)-g_2(x)$, we have $\int_0^\infty g(x) \dd x=\int_0^\infty g_2(x)\dd x+\int_0^M \Delta g_{12}(x)\dd x=1$ with $\int_0^\infty g_2(x)\dd x=\sum_{n=0}^{m-1}c_n$ and $\int_0^M \Delta g_{12}(x)\dd x$ is identical to (\ref{eq:Integral_in_non_homegenous_eqns}) but with $I_1=0$,  $I_2=I_6=I_3^*\definedas \zeta_k=\int_{u=0}^{M}\e^{-\theta_k u}\dd u=(1-\e^{-\theta_k M})/\theta_k$ for $k\neq 0$ and $\zeta_k=M$ for $k=0$, $I_4=\int_{u=0}^M \e^{-\lambda_n u}\dd u=(1-\e^{-\lambda_n M})/\lambda_n$, $I_5=\int_{u=0}^{M}\e^{-\lambda u}u^{t}\dd u=\lambda^{-t-1}\gamma(t+1,\lambda M)$, and $I_7=I_5$ for $t=q$. If we define the unit area condition as $\sum_{n=0}^{m-1}D_{n}c_n=1$ then $D_{n}$ reduces to (\ref{eq:D_n_OO}). Adding equations  $\sum_{n=0}^{m-1}B_{sn}c_n=0$ and $\sum_{n=0}^{m-1}D_{n}c_n=1$, we get a non-homogeneous system of linear equations which has a unique solution for $c_n$. This completes the proof.
%
\else
\ifARXIV
\section{Proof of Corollary \ref{theo:stationary_dist_infinite_Nakagami_OO} (Energy Distribution for an Infinite-Size Buffer with the On-Off Policy and Gamma-Distributed EH Process)}
\else
\section{Proof of Corollary \ref{theo:stationary_dist_infinite_Nakagami_OO} (On-Off Policy with Infinite-Size Buffer)}
\fi
\label{app:stationary_dist_infinite_Nakagami_OO_shortened}
Define $g_1(x)\definedas g(x), \, 0\!\leq \!x<\!M$, and $g_2(x)\definedas g(x), \, x\!>\!M$. Similar to the approach in Appendix \ref{app:stationary_dist_infinite_Nakagami_BE}, we postulate a solution of the form $g_2(x)=\sum\limits_{n=0}^{m-1} \lambda_n c_n\e^{-\lambda_n x}$, see Corollary \ref{theo:stationary_dist_infinite_Nakagami}.  Substituting $g_2(x)$ in (\ref{eq:Integral_eqn_infinite_partb_OO}) and using the pdf of the Gamma distribution $f(x)=\frac{\lambda^m}{(m-1)!}\e^{-\lambda x}x^{m-1}$, we get $\sum\limits_{n=0}^{m-1} \lambda_n c_n\e^{-\lambda_n x}=$\small
\begin{equation}
\int\limits_{u=0}^M\frac{\lambda^m}{(m-1)!}\e^{-\lambda(x-u)}(x-u)^{m-1}g_1(u)\dd u\\[-1ex]
+\sum\limits_{n=0}^{m-1} \lambda_n c_n \Big(\frac{\lambda}{\lambda-\lambda_n}\Big)^m\e^{-\lambda_n M}\Bigg[\e^{-\lambda_n x}-\e^{-\lambda x}\sum\limits_{s=0}^{m-1}\frac{(\lambda-\lambda_n)^s}{s!}x^s\Bigg],
\label{eq:integral_eqn_g2_postulation_OO}
\end{equation}\normalsize
where we used $\int_{u=M}^{M+x}\e^{(\lambda-\lambda_n)u}(x-u+M)^{m-1}\dd u\!=\!\e^{(\lambda-\lambda_n)(x+M)}(\lambda-\lambda_n)^{-m}\gamma\left(m,(\lambda-\lambda_n)x\right)$, see \cite[Eq. 3.382.1]{table_of_integrals_Ryzhik},   and $\gamma\left(m,(\lambda-\lambda_n)x\right)=(m-1)!\Big[1-\e^{-(\lambda-\lambda_n)x}\sum\limits_{s=0}^{m-1}\frac{(\lambda-\lambda_n)^sx^s}{s!}\Big]$.
Eq. (\ref{eq:integral_eqn_g2_postulation_OO}) holds if the coefficient of $\e^{-\lambda_n x}$ is $\lambda
_nc_n$, i.e., if $\left(\frac{\lambda}{\lambda-\lambda_n}\right)^m \!\e^{-\lambda_n M}\!=\!1$, which leads to $(\lambda_n-\lambda)\frac{M}{m}=W_0\left(-\delta \e^{-\delta}\e^{-\J\frac{2\pi n}{m}}\right)$, $\forall\, n\!=\!0,\ldots,m-1$, see Appendix \ref{app:stationary_dist_infinite_Nakagami_BE}, and if the coefficient of $x^s \e^{-\lambda x}$ is zero $\forall \,s$, i.e.,
\begin{equation}
\frac{\lambda^m (-1)^{m-1-s}}{(m-1-s)!}\int\limits_{u=0}^M\e^{\lambda u}u^{m-1-s}g_1(u)\dd u-\sum\limits_{n=0}^{m-1}\lambda_nc_n(\lambda-\lambda_n)^s=0, \quad \quad \forall s=0,\ldots,m-1,
\label{eq:cn_homogeneous_eqns_OO}
\end{equation}
where we used the binomial expansion $(x-u)^{m-1}=\sum\limits_{s=0}^{m-1}\binom{m-1}{s}(-u)^{m-1-s}x^s$ in the first term in (\ref{eq:integral_eqn_g2_postulation_OO}).   Due to space limitations, we provide in the following only the main steps needed to finalize the proof. First, $g_1(x)$ is obtained by substituting $f(x)$ and $g_2(x)$ in (\ref{eq:Integral_eqn_infinite_parta_OO}), i.e., 
\begin{equation}
g_1(x)\!=\!\!\int\limits_{u=0}^x\!\!\frac{\lambda^m}{(m-1)!}\e^{-\lambda(x-u)}(x-u)^{m-1}g_1(u)\dd u
+\underbrace{\sum\limits_{n=0}^{m-1}\!\!\lambda_n c_n \Bigg[\e^{-\lambda_n x}-\e^{-\lambda x}\sum\limits_{t=0}^{m-1}\!\frac{(\lambda-\lambda_n)^t}{t!}x^t\Bigg]}_{T(x)},
\label{eq:volterra_integral_eqn_g1_OO}
\end{equation}
where $T(x)$ is identical to the second term in (\ref{eq:integral_eqn_g2_postulation_OO}) since the second terms in (\ref{eq:Integral_eqn_infinite_parta_OO}) and (\ref{eq:Integral_eqn_infinite_partb_OO}) are identical. Eq. (\ref{eq:volterra_integral_eqn_g1_OO}) is a Volterra integral equation of the second kind, whose solution is given by $g_{1}(x)=T(x)+I_R(x)$, where $I_R(x)=\int_{u=0}^{x}R(x-u)T(u)\dd u$ and $R(x)=\frac{\e^{-\lambda x}}{m}\sum_{k=0}^{m-1}\lambda\e^{\lambda x\cos(\eta_k)}\cos\left(\mu_{k,x}\right)$
with $\mu_{k,x}\definedas\eta_k+\lambda x\sin(\eta_k)$ and $\eta_k\definedas\frac{2\pi k}{m}$, see \cite[eqs. 2.2-2.31]{polyanin2008handbook}. Solving $I_R(x)$ and using $T(x)$ in (\ref{eq:volterra_integral_eqn_g1_OO}), $g_1(x)$ reduces to (\ref{eq:g1_INF_OO}). 
In order to improve the accuracy of the inversion of matrix $\V{A}$ in Corollary \ref{theo:stationary_dist_infinite_Nakagami_OO}, we scale (\ref{eq:cn_homogeneous_eqns_OO}) by $1/\lambda^{s+1}$ and define it after scaling as  $\sum_{n=0}^{m-1}B_{sn}c_n=0$, then $B_{sn}$ reduces to (\ref{eq:B_sn_OO}).
Finally, to solve for $c_n$, we need to add the unit area condition  $\int_0^\infty g(x)\dd x\!=\!1$ in (\ref{eq:cn_homogeneous_eqns_OO}). If we define the unit area condition $\int_0^M g_1(x)\dd x+\int_M^{\infty} g_2(x)\dd x\!=\!1$ as $\sum_{n=0}^{m-1}D_{n}c_n\!=\!1$, then $D_{n}$ reduces to (\ref{eq:D_n_OO}). Adding equations  $\sum_{n=0}^{m-1}B_{sn}c_n\!=\!0$ and $\sum_{n=0}^{m-1}D_{n}c_n\!=\!1$, we get a non-homogeneous system of linear equations which has a unique solution in $c_n$. This completes the proof.
\fi
\ifARXIV
\section{Proof of Remark \ref{coro:real_valued_pdfs} (Real-Valued  Energy Distribution for an Infinite-Size Buffer with the Best-Effort and the On-Off Policies)}
\label{app:real_valued_pdfs}
The limiting pdf $g(x)$ of the energy buffer content for the best-effort and the on-off policies, obtained in Corollaries \ref{theo:stationary_dist_infinite_Nakagami} and \ref{theo:stationary_dist_infinite_Nakagami_OO} are expressed in terms of the complex-valued coefficients $\lambda_n$ and $c_n$. Coefficient $\lambda_n$ is complex-valued given by $\lambda_n\!=\lambda+\frac{m}{M} W_0\left(-\delta \e^{-\delta}\e^{-\J\frac{2 \pi n}{m}}\right)$. Furthermore, coefficients $c_n$, $n=0,\ldots,m-1$ are also complex-valued since they are obtained by solving the equation $\V{A} \V{c}=\V{1}$, where the coefficients of matrix $\V{A}$ are based on $\lambda_n$ and therefore are complex-valued. Next, we show that although $\lambda_n$ and $c_n$ are complex-valued, $g(x)$ is real-valued. 

The Lambert W function is the inverse function of $z=w\e^{w}$, i.e., $W_0(z)=w$. If $z$ is complex-valued, it means that $w=W_0(z)$ is also complex-valued. Furthermore, since $w^*\e^{w^*}=z^*$, it follows that $W_0(z^*)=w^*=[W_0(z)]^*$. From  $\lambda_n\!=\lambda+\frac{m}{M} W_0\left(-\delta \e^{-\delta}\e^{-\J\frac{2 \pi n}{m}}\right)$, $n=0,\ldots,m-1$, $\lambda_0$ is real and  using $\e^{\frac{-\J 2\pi n}{m}}=\left(\e^{\frac{-\J 2\pi (m-n)}{m}}\right)^{*}$ for $n\neq 0$, it follows that $W_0\left(\e^{\frac{-\J 2\pi n}{m}}\right)=W_0\left(\left(\e^{\frac{-\J 2\pi (m-n)}{m}}\right)^{*}\right)=\left(W_0\left(\e^{\frac{-\J 2\pi (m-n)}{m}}\right)\right)^{*}$. Since $\lambda$, $m/M$, and $-\delta\e^{-\delta}$ are all real-valued, it follows that $\lambda_{n}=\lambda_{m-n}^*$. In particular, if $m$ is odd, then $\lambda_0$ is real-valued and $\lambda_{n}=\lambda_{m-n}^*$ $\forall$ $n=1,\ldots,\frac{m-1}{2}$, and if $m$ is even, then $\lambda_0$ and $\lambda_{m/2}$ are real-valued and $\lambda_{n}=\lambda_{m-n}^*$ $\forall$ $n=1,\ldots,\frac{m}{2}-1$.

Next, we use $\sum_{n=0}^{m-1} c_n A_{s,n}= 1$, $\forall\, s=0,\ldots,m-1$. For both the best-effort and the on-off transmission policies, $A_{s,n}$ is complex-valued only due to $\lambda_n$. Thus, $A_{s,n}=A_{s,m-n}^*$. Hence, for a given $s$, in order for the summation $\sum_{n=0}^{m-1} c_n A_{s,n}$ to be real-valued, $\left( c_n A_{s,n}\right)=\left(c_{m-n}A_{s,m-n}\right)^*$ must hold, i.e., $c_n=c_{m-n}^*$. In fact, the complex conjugate properties for $\lambda_n$ and $c_n$ can be used to solve for only half of the unknown coefficients $\lambda_n$ and $c_n$. Finally, using the identities $a^*b^*=(ab)^*$ and $\Im\{a+a*\}=0$, the pdfs $g(x)$ in Corollaries \ref{theo:stationary_dist_infinite_Nakagami} and \ref{theo:stationary_dist_infinite_Nakagami_OO} are real-valued. This completes the proof. \vspace{-0.5cm}

\section{Proof of Theorem \ref{theo:limiting_dist_Finite} (Uniqueness of a Limiting Energy Distribution for a Finite-Size Buffer with the Best-Effort and On-Off Policies)}
\label{app:limiting_dist_Finite_existence_uniqueness_BE_OO} 
Similar to the random walk process on a half line in \cite[Section 4.3.1]{Meyn_Tweedie}, if the distribution of the EH process $\{X(i)\}$ has an infinite positive tail, then the state space $S$ of the storage processes $\{B(i)\}$ of the best-effort policy in (\ref{eq:general_storage_equation_BE}) and that of the on-off policy in  (\ref{eq:general_storage_equation_OO}) contain an atom at $K$. That is, the energy level $B(i)\!=\!K$ is reachable with non-zero probability. Define the measure $\phi$ as $\phi(0,K)\!=\!0$ and $\phi(\{K\})\!=\!1$, then the process $\{B(i)\}$ is $\phi$-irreducible, see \cite[Section 4.2]{Meyn_Tweedie}. Furthermore, $\{B(i)\}$ is also $\psi$-irreducible with $\psi(A)=\sum_n \mathbb{P}^n(K,A)2^{-n}$, where $\mathbb{P}^n(x,A)$ is the probability that the Markov chain moves from energy state $x$ to energy set $A$ in $n$ time steps. The dynamics of $\{B(i)\}$ in (\ref{eq:general_storage_equation_BE}) and (\ref{eq:general_storage_equation_OO}) ensures that all energy sets are reachable a.s. from any initial state of the buffer in a finite mean time. Hence, the chain is positive Harris recurrent \cite[Proposition 9.1.1]{Meyn_Tweedie}, where \emph{positive} recurrence follows from \cite[Theorem 10.2.2]{Meyn_Tweedie}. Thus, $\{B(i)\}$ possesses a unique stationary distribution $\pi$. Finally, with the additional property of $\{B(i)\}$ being aperiodic (i.e., no energy level sets are only revisited after a fixed number of time slots $>1$ (period $>1$)), it follows from \cite[Theorem 13.3.3]{Meyn_Tweedie} that $\{B(i)\}$ converges to the distribution $\pi$ in total variation from any initial distribution $\chi$, i.e., $\lim\limits_{n\to\infty}\sup\limits_{A}|\int \chi(\dd x)\mathbb{P}^n(x,A)-\pi(A)|\to0$. This completes the proof.

\section{Proof of Theorem \ref{theo:BE_Finite_intgeral_eqn} (Integral Equation of the Energy Distribution for a Finite-Size Buffer with the Best-Effort Policy)}
\label{app:BE_Finite_intgeral_eqn}
The integral equations in (\ref{eq:g_integral_eqn_finite_BE}) and (\ref{eq:partc_BE}) can be understood by adopting the same approach used to prove (\ref{eq:Integral_eqn_infinite_BE}). In particular, if we set $B(i)=u$ and $B(i+1)=x$, then (\ref{eq:general_storage_equation_BE}) reads
\begin{equation}
x=
\begin{cases}
X(i) & u\leq M \quad\&\quad X(i)<K\\
u-M+X(i) & u>M \quad\&\quad u-M+X(i)<K\\
K & {\rm otherwise.}
\end{cases}\vspace{-0.1cm}
\label{eq:Our_storage_equation_cases}
\end{equation}
Consider first the continuous part of the distribution, i.e., $g(x)$ defined on $0\leq x <K$ given in (\ref{eq:g_integral_eqn_finite_BE}). Eq. (\ref{eq:parta_BE}) is identical to (\ref{eq:Integral_eqn_infinite_BE}), however, we need to further ensure that the upper limit on $u$ given by $M+x$ (for a non-negative harvested energy) is in the domain of $g(u)$, i.e., $M+x<K$ must hold. Hence, (\ref{eq:parta_BE}) is valid only for $x<K-M$ (with strict inequality). For the rest of the range of $x$ in (\ref{eq:partb_BE}), i.e., $K-M\leq x<K$, the upper limit  $M+x$ on $u$ is larger than or equal to $K$. Thus, the whole range of $0< u\leq K$ contributes to $g(x)$. The range $0< u<K$ is covered by the first two integrals in (\ref{eq:partb_BE}), and $u=K$ is considered in the last term. Finally, at $x=K$, the probability of a full buffer $\pi(K)$ in (\ref{eq:partc_BE}) is obtained similar to (\ref{eq:partb_BE}). However, rather than considering the pdf at the amount of harvested energy $x-[u-M]^+$ as in (\ref{eq:partb_BE}), we consider the ccdf $\bar{F}(x-[u-M]^+)$ instead (at $x\!=\!K$). This is because the full buffer level $K$ is attained when the amount of harvested energy is larger than or equal to $K\!-\![u\!-\!M]^+$, $\forall \, u$ in $0\!<\!u\!\leq\!K$. This completes the proof.\vspace{-0.3cm}

\fi
\ifARXIV
\section{Proof of Corollary \ref{coro:limiting_dist_Finite_Nakagami_exact_BE} (Energy Distribution for a Finite-Size Buffer with the Best-Effort Policy and Gamma-Distributed EH Process)}
\label{app:limiting_dist_Finite_Nakagami_exact_BE}
We obtain the limiting pdf of the storage process in (\ref{eq:general_storage_equation_BE}) using an approach similar to that used by Prabhu in \cite{Prabhu_1985} to obtain the corresponding limiting cdf of process $\{U(i)\}$ in Moran's model with finite dam size and Gamma distributed inputs, cf. Remark \ref{remark:equivalence_to_Morans_Model}. Substituting the pdf of the Gamma-distributed EH process $f(x)\!=\!\frac{\lambda^m}{\Gamma(m)}x^{m-1}\e^{-\lambda x}$ and the corresponding ccdf $\bar{F}(x)=\e^{-\lambda x}\sum_{r=0}^{m-1}\frac{(\lambda x)^r}{r!}$ in (\ref{eq:g_integral_eqn_finite_BE}) and (\ref{eq:partc_BE}), respectively, we obtain
\begin{numcases}{g(x)\!=\!\! \label{eq:g_integral_eqn_finite_Nakagami}}
\hspace{-0.1cm}\frac{\lambda^m\e^{-\lambda x}}{(m\!-\!1)!}\left[x^{m-1}\!\!\! \int\limits_{u=0}^{M}\!\!\! g(u) \dd u +\!\!\!\!\!\int\limits_{u=M}^{M+x}\!\!\! (x\!-\!u\!+\! M)^{m-1} \e^{-\lambda(M-u)}g(u)\dd u \right]\!,  &  \hspace{-0.5cm} $0\!\leq\!x\!<\!\! K\!\!-\!\!M$  \label{eq:parta_Nakagami} \\ 
\hspace{-0.1cm}\begin{aligned}\frac{\lambda^m\e^{-\lambda x}}{(m\!-\!1)!}\Biggg[x^{m-1}\!\!\! \int\limits_{u=0}^{M}\!\!\! g(u) \dd u &+ \!\!\!\!\int\limits_{u=M}^{K}\!\!\!(x\!-\!u\!+\! M)^{m-1} \e^{-\lambda(M-u)}g(u)\dd u  \\[-2ex]
& +\pi(K)(x\!-\!K\!+\!M)^{m-1}\e^{-\lambda(M-K)}\Biggg], \end{aligned}&  \hspace{-0.6cm} $K\!\!-\!\!M \!\leq \!x\!<\!\!K$ \label{eq:partb_Nakagami}
\end{numcases}
\begin{equation}\vspace{-0.2cm}
\pi(K)\!=\!\frac{\e^{-\lambda K}}{1\!-\!\e^{-\lambda M}\!\!\sum\limits_{t=0}^{m-1}\!\frac{(\lambda M)^t}{t!}}\!\left[\!\sum\limits_{r=0}^{m-1}\!\frac{(\lambda K)^r}{r!}\!\!\int\limits_{u=0}^{M}\!\! g(u) \dd u +  \!\!\!\int\limits_{u=M}^{K}\!\!\!\!\e^{-\lambda(M-u)}\sum\limits_{r=0}^{m-1}\!\!\frac{(\lambda(K\!-\!u\!+\!M))^r}{r!} g(u)\dd u \right].
\label{eq:partc_Nakagami}
\end{equation}
Define $I_1(r)=K^r I_1$, $I_1=\int_0^M g(u) \dd u$, $I_2(r)=\int_{u=M}^{K}(K-u+M)^{r} \e^{-\lambda(M-u)}g(u)\dd u$, and 
$\alpha_r=\lambda^r[I_1(r)+I_2(r)+M^r\pi(K)\e^{-\lambda(M-K)}]$.
Hence, $g(x)$ and $\pi(K)$ can be written as

\begin{numcases}{g(x)\!=\label{eq:g_integral_eqn_finite_Nakagami2}}
\begin{aligned}\hspace{-0.1cm}g_0(x)-\frac{\lambda^m\e^{-\lambda x}}{(m\!-\!1)!}\Biggg[&\underbrace{\int\limits_{u=M+x}^{K}\!\!\! (x\!-\!u\!+\! M)^{m-1} \e^{-\lambda(M-u)}g(u)\dd u}_{I_3}\\[-2ex]
&+\pi(K)(x\!-\!K\!+\!M)^{m-1}\e^{-\lambda(M-K)}\Biggg]\!, \hspace{-0.4cm}\end{aligned} &  $0\!\leq\!x\!<\!\! K\!\!-\!\!M$  \label{eq:parta_Nakagami2} \\ 
\hspace{-0.1cm}\definedas g_0(x)=\frac{\lambda^{m-r}\e^{-\lambda x}}{(m\!-\!1)!}\sum\limits_{r=0}^{m-1}\binom{m-1}{r}(x-K)^{m-1-r}\alpha_r,&   $K\!\!-\!\!M \!\leq \!x\!<\!\!K$ \label{eq:partb_Nakagami2}
\end{numcases}
\begin{equation}
\pi(K)\!=\!\frac{\e^{-\lambda K}}{1-\e^{-\lambda M}\sum\limits_{t=0}^{m-1}\frac{(\lambda M)^t}{t!}}\sum\limits_{r=0}^{m-1}\frac{\lambda ^r}{r!}\left(I_1(r)+I_2(r)\right).
\label{eq:partc_Nakagami2}
\end{equation}
First, we start by writing $\pi(K)$ in (\ref{eq:partc_Nakagami2}) in terms of the unknown coefficients $\alpha_r,\, r=0,\ldots, m-1$. In (\ref{eq:partc_Nakagami2}), $\lambda^r[I_1(r)+I_2(r)]$ can be replaced by $\alpha_r-(\lambda M)^r\pi(K)\e^{-\lambda(M-K)}$ and $\pi(K)$ reduces to (\ref{eq:Pi_K_Nakagami}).
\subsection{Expressing $g(x)$ in Terms of $\alpha_r$}
Next, we derive the limiting pdf $g(x)$ in stripes of width $M$. In particular, we derive $g_n(x)\definedas g(x),\quad [K-(n+1)M]^+\leq x<K-nM$, $\forall n=0,\ldots,l'$, where $l'=l-1$ if $K$ is an integer multiple of $M$ (i.e., $K=lM$, $l\in\mathbb{Z}^+$), and $l'=l$, otherwise (i.e., $K=lM+\Delta$, with $\Delta\neq 0$). The $M$-width pdf section $g_n(x)$ will be derived recursively by induction from $n=0$ backwards till $n=l'$  in terms of the unknown coefficients $\alpha_r$, $r=0,\ldots,m-1$. With $g_0(x)$ given in (\ref{eq:partb_Nakagami2}), we start from $n=1$ and use 
(\ref{eq:parta_Nakagami2}) to obtain $g_1(x)$, $K-2M\leq x<K-M$. In this case, the lower integral limit $M+x$ in $I_3$ satisfies $K-M\leq M+x<K$, hence, $g(u)$ in $I_3$ is $g_0(u)$. Substituting $g_0(u)$ in $I_3$ and using 
\begin{equation}
\int\limits_{u=b}^a (b-u)^{c}(u-a)^{d}\dd u=-(b-a)^{c+d+1}\frac{c!d!}{(c+d+1)!}, \quad \text{if } b<a \text{ and } c,d\in\mathbb{N},
\label{eq:general_integral1}
\end{equation}
 with $b=M+x$, $a=K$, $c=m-1$, and $d=m-1-r$, $g_1(x)$ reduces to
\begin{align}
g_1(x)=\frac{\lambda^m \e^{-\lambda x}}{(m\!-\!1)!}\Biggg[&\sum\limits_{r=0}^{m-1}\alpha_r\lambda^{-r}\binom{m\!-\!1}{r}\Bigg(\!(x\!-\!K)^{m-1-r}\!+\!\lambda^m\e^{-\lambda M}\frac{(m\!-\!1\!-\!r)!}{(2m\!-\!1\!-\!r)!}(M\!+\!x\!-\!K)^{2m-1-r}\!\Bigg)\notag\\[-2ex]
&-\e^{-\lambda(M-K)}\pi(K)(M+x-K)^{m-1}\Biggg].
\label{eq:g1_Nakagami}
\end{align}\vspace{-0.2cm}
To obtain $g_2(x)$, $K-3M\leq x<K-2M$, we use (\ref{eq:parta_Nakagami2}) and define the integrand of $I_3$ as $A(x,u) g(u)$, then $I_3$ can be written as
\begin{equation*}
I_3=\!\!\int\limits_{M+x}^{K-M}\!\!\!\!A(x,u)g_1(u)\dd u+\int\limits_{K-M}^{K}\!\!\!\!A(x,u)g_0(u)\dd u = \int\limits_{M+x}^{K}\!\!\!\!A(x,u)g_0(u)\dd u+\int\limits_{M+x}^{K-M}\!\!\!\!A(x,u)(g_1(u)-g_0(u))\dd u,
\label{eq:I_3_redefine_for_g2}
\end{equation*}
which can be obtained by solving integrals of the form (\ref{eq:general_integral1}). Substituting $I_3$ in (\ref{eq:parta_Nakagami2}), we get
\begin{align}
g_2(x)=&\frac{\lambda^m \e^{-\lambda x}}{(m-1)!}\Biggg[\sum\limits_{r=0}^{m-1}\!\alpha_r\lambda^{-r}\binom{m\!-\!1}{r}\Bigg(\!(x\!-\!K)^{m-1-r}+\lambda^m\e^{-\lambda M}\frac{(m\!-\!1\!-\!r)!}{(2m\!-\!1\!-\!r)!}(M+x-K)^{2m\!-\!1\!-\!r}\notag\\[-2ex]
&+\lambda^{2m}\e^{-2\lambda M}\frac{(m-1-r)!}{(3m-1-r)!}(2M+x-K)^{3m-1-r}\Bigg)
-\e^{-\lambda M}\e^{\lambda K}\pi(K)(M+x-K)^{m-1}\notag\\[-2ex]
&-\lambda^m\e^{-2\lambda M}\e^{\lambda K}\pi(K)(2M+x-K)^{2m-1}\frac{(m-1)!}{(2m-1)!}\Biggg].
\label{eq:g2_Nakagami}
\end{align} 
Similarly, $g_n(x)$, $n\geq 3$, can be obtained by induction, and $g_n(x)$, $n=0,\ldots,l'$, can be written as
\begin{align}
g_n(x)=\frac{\lambda^m \e^{-\lambda x}}{(m-1)!}\Biggg[&\sum\limits_{q=0}^{n}\lambda^{qm}\e^{-q\lambda M}\sum\limits_{r=0}^{m-1}\!\alpha_r\lambda^{-r}\binom{m\!-\!1}{r}\frac{(m\!-\!1\!-\!r)!}{((q\!+\!1)m\!-\!1\!-\!r)!}(qM\!+\!x\!-\!K)^{(q+1)m-1-r}\notag\\[-2ex]
&-\sum\limits_{q=1}^{n}\lambda^{(q-1)m}\e^{-q\lambda M}\e^{\lambda K}\pi(K)(qM+x-K)^{qm-1}\frac{(m-1)!}{(qm-1)!}\Biggg].
\label{eq:gn_Nakagami_step}
\end{align}
Substituting $\pi(K)$ from (\ref{eq:Pi_K_Nakagami}) in (\ref{eq:gn_Nakagami_step}), $g_n(x)$ reduces to (\ref{eq:pdf_Nakagami_m}).

\subsection{Defining a Non-Homogeneous System of Linear Equations in $\alpha_s$}
Now, only the unknown coefficients $\alpha_s$, $s=0,\ldots,m-1$, remain to be determined. Using (\ref{eq:Pi_K_Nakagami}), $\alpha_s$ satisfies
\begin{equation}
\alpha_s\definedas\sum\limits_{r=0}^{m-1}\alpha_r d_{sr}=\lambda^s\left(I_1(s)+I_2(s)+M^s\e^{-\lambda M}\sum\limits_{r=0}^{m-1}\frac{\alpha_r}{r!}\right),\quad s=0,\ldots,m-1.
\label{eq:alpha_r_definition}
\end{equation}
Next, we obtain $I_1(s)=K^sI_1$ and $I_2(s)$ in terms of $\alpha_r$, $r=0,\ldots,m-1$, so that (\ref{eq:alpha_r_definition}) forms a system of linear equations in $\alpha_r$. Using $g(x)$ in (\ref{eq:pdf_Nakagami_m}) and assuming a general buffer size $K=lM+\Delta$ with $l\in\mathbb{Z}^+$, $\Delta<M$, $I_1$ can be written as
\begin{equation}
I_1=\int\limits_0^M g(u) \dd u=\!\!\!\!\int\limits_0^{K-lM}g_l(u)\dd u +\!\!\!\!\int\limits_{K-lM}^M g_{l-1}(u) \dd u=\underbrace{\int\limits_0^{M}g_{l-1}(u)\dd u}_{I_a} +\underbrace{\!\!\!\!\int\limits_0^{K-lM} \left(g_{l}(u)-g_{l-1}(u)\right) \dd u}_{I_b}.
\label{eq:I1_expansion}
\end{equation}
To get $I_a$, we solve integrals of the form
\begin{equation}
\int\limits_{0}^d\e^{-\lambda u}(u-a)^b=b!\lambda^{-(b+1)}\sum\limits_{t=0}^b\frac{\lambda^t}{t!}\left[(-a)^t-\e^{-\lambda d} (d-a)^t\right], \text{if } a\geq d \text{ and } b\in\mathbb{N},
\label{eq:Ia_expansion}
\end{equation}
with $a=K-qM$, $d=M$, and $b=(q+1)m-r-1$ in one integral and $b=qm-1$ in another integral. This reduces $I_a$ to
\begin{equation}
I_a=\sum\limits_{r=0}^{m-1}\frac{\alpha_r}{r!}\sum\limits_{q=0}^{l-1}\e^{-\lambda M q}\sum\limits_{t=qm}^{(q+1)m-r-1}\frac{\lambda^t}{t!}\Big\{(qM-K)^t-\e^{-\lambda M}((q+1)M-K)^t\Big\}.
\label{eq:Ia_expansion3}
\end{equation}
Similarly, $I_b=\sum_{r=0}^{m-1}\frac{\alpha_r}{r!}\e^{-\lambda M l}\sum_{t=lm}^{(l+1)m-r-1}\frac{\lambda^t}{t!}(lM-K)^t$. 

Define the integrand of $I_2(s)$ as $B(u,s)g(u)$, then $I_2(s)=\int_M^K B(u,s)g(u) \dd u$ can be expanded as
\begin{equation}
I_2(s)=\int\limits_{K-M}^K B(u,s) g_0(u)\dd u + \int\limits_{K-2M}^{K-M} B(u,s) g_1(u)\dd u + \ldots+ \int\limits_{M}^{K-(l-1)M} B(u,s) g_{l-1}(u)\dd u.
\label{eq:I2_expansion1}
\end{equation}
Since the sum in variable $q$ in (\ref{eq:pdf_Nakagami_m}) has $q=n$ as upper limit, if we combine the integrals of the terms with variable $q$ in  (\ref{eq:I2_expansion1}), then the $q^{\text{th}}$ term is integrated over the limits $u=M\to$ $u=K-qM$ and $I_2(s)$ reduces to
\begin{align}
I_2(s)=\sum\limits_{r=0}^{m-1}\frac{\alpha_r}{r!}\Biggg[&
\sum\limits_{q=0}^{l-1}\lambda^{(q+1)m-r}\e^{-\lambda M (q+1)} \int\limits_{M}^{K-qM}\frac{(u-(K-qM))^{(q+1)m-r-1}}{((q+1)m-r-1)!}(M+K-u)^s \dd u\notag\\[-1ex]
&- \sum\limits_{q=1}^{l-1}\lambda^{qm}\e^{-\lambda M (q+1)} \int\limits_{M}^{K-qM}\frac{ (u-(K-qM))^{qm-1}}{(qm-1)!}(M+K-u)^s \dd u\Biggg].
\label{eq:I2_expansion2}
\end{align}
Using $(M+K-u)^s=\sum_{t=0}^s\binom{s}{t}(M-u)^t K^{s-t}$, and solving integrals of the form (\ref{eq:general_integral1}) with $b=M$, $a=K-qM$, $c=t$, and $d=(q+1)m-r-1$ in one integral and $d=qm-1$ in another integral, then $I_2(s)$ reduces to 
\begin{equation}
\begin{aligned}
I_2(s)=-K^s\sum\limits_{r=0}^{m-1}\frac{\alpha_r}{r!}\Biggg[
&\sum\limits_{q=0}^{l-1}\lambda^{(q+1)m-r}\e^{-\lambda M (q+1)}\sum\limits_{t=0}^s\binom{s}{t}K^{-t}t!\frac{((q+1)M-K)^{(q+1)m-r+t}}{((q+1)m-r+t)!}\\[-1ex]
&-\sum\limits_{q=1}^{l-1}\lambda^{qm}\e^{-\lambda M (q+1)} \sum\limits_{t=0}^s\binom{s}{t}K^{-t}t!\frac{((q+1)M-K)^{qm+t}}{(qm+t)!}\Biggg].
\end{aligned}
\label{eq:I2_expansion3}
\end{equation}
Finally, the last term inside the brackets in (\ref{eq:alpha_r_definition}) is the same as the term in the second line in (\ref{eq:I2_expansion3}) for $q=0$, namely $\e^{-\lambda M}\sum_{r=0}^{m-1}\frac{\alpha_r}{r!}M^s=\sum_{r=0}^{m-1}\frac{\alpha_r}{r!}\e^{-\lambda M} \sum_{t=0}^s\binom{s}{t}K^{s-t}(M-K)^t$. 

Combining $I_a$ in (\ref{eq:Ia_expansion3}), $I_b$, and $I_2(s)$ in (\ref{eq:I2_expansion3}), then $d_{sr}$ in (\ref{eq:alpha_r_definition}) reduces to
\begin{equation}
\begin{aligned}
&d_{sr}=\frac{K^s\lambda^s}{r!}\sum\limits_{q=0}^{l}\e^{-\lambda M q}\sum\limits_{t=qm}^{(q+1)m-r-1}\frac{(-\lambda)^t}{t!}\Big\{(K-qM)^t-\e^{-\lambda M}\left([K-(q+1)M]^+\right)^t\Big\}\\[-1.2ex]
&-\!\frac{\lambda^sK^s}{r!}\!\sum\limits_{q=0}^{l-1}\!\lambda^{qm}\e^{\!\!-\lambda M (q+1)}\!\sum\limits_{t=0}^s\!\binom{s}{t}K^{-t}t!\Big\{\lambda^{m-r}\frac{((q\!+\!1)M\!-\!K)^{(q\!+\!1)m\!-\!r\!+\!t}}{((q\!+\!1)m\!-\!r\!+\!t)!}- \frac{((q\!+\!1)M\!-\!K)^{qm\!+\!t}}{(qm\!+\!t)!}\Big\}.
\label{eq:d_sr}
\end{aligned}
\end{equation}
In order to obtain a unique solution for $\alpha_r$, $r=0,\ldots,m-1$, we additionally use the unit area condition $\int_0^K g(u)\dd u+\pi(K)=1$, where similar to the analysis used to expand $I_2(s)$, the integral $\int_0^K g(u)\dd u$ can be expanded as follows
\begin{equation}
\int\limits_0^K g(u)\dd u=\int\limits_{K-M}^K g_0(u)\dd u + \int\limits_{K-2M}^{K-M}  g_1(u)\dd u + \ldots+ \int\limits_{0}^{K-lM} g_{l}(u)\dd u.
\label{eq:area_under_g}
\end{equation}
From (\ref{eq:pdf_Nakagami_m}), the term with variable $q$ in (\ref{eq:area_under_g}) is integrated over the limits $u=0\to$ $u=K-qM$. Hence, (\ref{eq:area_under_g}) reduces to
\begin{equation}
\begin{aligned}
\int\limits_0^K g(u)\dd u=\sum\limits_{r=0}^{m-1}\frac{\alpha_r}{r!}\Biggg[&
\sum\limits_{q=0}^{l}\lambda^{(q+1)m-r}\e^{-\lambda M q}\int\limits_{0}^{K-qM}\e^{-\lambda u}\frac{ (u-(K-qM))^{(q+1)m-r-1}}{((q+1)m-r-1)!}\dd u\\
 &- \sum\limits_{q=1}^{l}\lambda^{qm}\e^{-\lambda M q}\int\limits_{0}^{K-qM}\e^{-\lambda u}\frac{ (u-(K-qM))^{qm-1}}{(qm-1)!}\dd u \Biggg],
\end{aligned}
\label{eq:area_under_g_step2}
\end{equation}
where the integrals in (\ref{eq:area_under_g_step2}) have the form of (\ref{eq:Ia_expansion}) with $d=a$, hence $\sum\limits_{t=0}^b \frac{\lambda^t}{t!} (d-a)^t=1$ since at $t=0$, $0^0=1$ and (\ref{eq:area_under_g_step2}) reduces to
\begin{equation*}
\begin{aligned}
\int\limits_0^K g(u)\dd u=\sum\limits_{r=0}^{m-1}\frac{\alpha_r}{r!}\Biggg[
-\e^{-\lambda K}+\sum\limits_{t=0}^{m-r-1}\frac{(-\lambda K)^t}{t!}+\sum\limits_{q=1}^l\e^{-\lambda M q }\sum\limits_{t=qm}^{(q+1)m-r-1}\frac{(\lambda(qM-K))^t}{t!}\Biggg],
\end{aligned}
\label{eq:area_under_g_step3}
\end{equation*}
which when added to $\pi(K)$ in (\ref{eq:Pi_K_Nakagami}) reduces to 
\begin{equation}
\int\limits_0^K g(u)\dd u+\pi(K)=\sum\limits_{r=0}^{m-1}\alpha_r\underbrace{\frac{1}{r!}\sum\limits_{q=0}^l\e^{-\lambda M q }\sum\limits_{t=qm}^{(q+1)m-r-1}\frac{(\lambda(qM-K))^t}{t!}}_{a_r}=1.
\label{eq:area_under_g_step4}
\end{equation}
Adding (\ref{eq:area_under_g_step4}) to (\ref{eq:alpha_r_definition}), we get $\alpha_s+\sum_{r=0}^{m-1}\alpha_r (a_r-d_{sr})=1, \,\,s=0,\ldots,m-1,$ which forms a non-homogeneous system of linear equations that can be written in matrix-vector notation as $\V{\alpha}+\V{A}\V{\alpha}=\V{1}$, where $\V{\alpha}=[\alpha_0,\ldots,\alpha_{m-1}]^T$, and $\V{A}$ is an $m\times m$  matrix whose entry in the $s^{\text{th}}$  row and the $r^{\text{th}}$ column is $A_{sr}=a_r-d_{sr}$, where $a_r$ is given in (\ref{eq:area_under_g_step4}) and $d_{sr}$ is given in (\ref{eq:d_sr}). This completes the proof. 


\else
\section{Proof of Corollary \ref{coro:limiting_dist_Finite_Nakagami_exact_BE} (Best-Effort Policy with Finite-Size Buffer)}
\label{app:limiting_dist_Finite_Nakagami_exact_BE}
We obtain the limiting pdf of the storage process in (\ref{eq:general_storage_equation_BE}) using an approach similar to that used by Prabhu in \cite{Prabhu_1985} to obtain the corresponding limiting cdf of process $\{U(i)\}$ in Moran's model with finite dam size and Gamma distributed inputs, cf. Remark \ref{remark:equivalence_to_Morans_Model}. Substituting the pdf of the Gamma-distributed EH process $f(x)\!=\!\frac{\lambda^m}{\Gamma(m)}x^{m-1}\e^{-\lambda x}$ and the corresponding ccdf $\bar{F}(x)=\e^{-\lambda x}\sum_{r=0}^{m-1}\frac{(\lambda x)^r}{r!}$ in (\ref{eq:g_integral_eqn_finite_BE}) and (\ref{eq:partc_BE}), respectively, we obtain
\begin{numcases}{g(x)\!=\!\! \label{eq:g_integral_eqn_finite_Nakagami}}
\hspace{-0.1cm}\frac{\lambda^m\e^{-\lambda x}}{(m\!-\!1)!}\left[x^{m-1}\!\!\! \int\limits_{u=0}^{M}\!\!\! g(u) \dd u +\!\!\!\!\!\int\limits_{u=M}^{M+x}\!\!\! (x\!-\!u\!+\! M)^{m-1} \e^{-\lambda(M-u)}g(u)\dd u \right]\!, \hspace{-0.4cm} &  $0\!\leq\!x\!<\!\! K\!\!-\!\!M$  \label{eq:parta_Nakagami} \\ 
\hspace{-0.1cm}\begin{aligned}\frac{\lambda^m\e^{-\lambda x}}{(m\!-\!1)!}\Biggg[x^{m-1}\!\!\! \int\limits_{u=0}^{M}\!\!\! g(u) \dd u &+ \!\!\!\!\int\limits_{u=M}^{K}\!\!\!(x\!-\!u\!+\! M)^{m-1} \e^{-\lambda(M-u)}g(u)\dd u  \\[-2ex]
& +\pi(K)(x\!-\!K\!+\!M)^{m-1}\e^{-\lambda(M-K)}\Biggg], \hspace{-0.4cm}\end{aligned}&   $K\!\!-\!\!M \!\leq \!x\!<\!\!K$ \label{eq:partb_Nakagami}
\end{numcases}
\begin{equation}\vspace{-0.2cm}
\pi(K)\!=\!\frac{\e^{-\lambda K}}{1\!-\!\e^{-\lambda M}\!\!\sum\limits_{t=0}^{m-1}\!\frac{(\lambda M)^t}{t!}}\!\left[\!\sum\limits_{r=0}^{m-1}\!\frac{(\lambda K)^r}{r!}\!\!\int\limits_{u=0}^{M}\!\! g(u) \dd u +  \!\!\!\int\limits_{u=M}^{K}\!\!\!\!\e^{-\lambda(M-u)}\sum\limits_{r=0}^{m-1}\!\!\frac{(\lambda(K\!-\!u\!+\!M))^r}{r!} g(u)\dd u \right].
\label{eq:partc_Nakagami}
\end{equation}
Define $I_1(r)=K^r I_1$, $I_1=\int_0^M g(u) \dd u$, $I_2(r)=\int_{u=M}^{K}(K-u+M)^{r} \e^{-\lambda(M-u)}g(u)\dd u$, and 
$\alpha_r=\lambda^r[I_1(r)+I_2(r)+M^r\pi(K)\e^{-\lambda(M-K)}]$.
Hence, $g(x)$ and $\pi(K)$ can be written as

\begin{numcases}{g(x)\!=\label{eq:g_integral_eqn_finite_Nakagami2}}
\begin{aligned}\hspace{-0.1cm}g_0(x)-\frac{\lambda^m\e^{-\lambda x}}{(m\!-\!1)!}\Biggg[&\underbrace{\int\limits_{u=M+x}^{K}\!\!\! (x\!-\!u\!+\! M)^{m-1} \e^{-\lambda(M-u)}g(u)\dd u}_{I_3}\\[-2ex]
&+\pi(K)(x\!-\!K\!+\!M)^{m-1}\e^{-\lambda(M-K)}\Biggg]\!, \hspace{-0.4cm}\end{aligned} &  $0\!\leq\!x\!<\!\! K\!\!-\!\!M$  \label{eq:parta_Nakagami2} \\ 
\hspace{-0.1cm}\definedas g_0(x)=\frac{\lambda^{m-r}\e^{-\lambda x}}{(m\!-\!1)!}\sum\limits_{r=0}^{m-1}\binom{m-1}{r}(x-K)^{m-1-r}\alpha_r, &   $K\!\!-\!\!M \!\leq \!x\!<\!\!K$ \label{eq:partb_Nakagami2}
\end{numcases}
\begin{equation}
\pi(K)\!=\!\frac{\e^{-\lambda K}}{1-\e^{-\lambda M}\sum\limits_{t=0}^{m-1}\frac{(\lambda M)^t}{t!}}\sum\limits_{r=0}^{m-1}\frac{\lambda ^r}{r!}\left(I_1(r)+I_2(r)\right).
\label{eq:partc_Nakagami2}
\end{equation}
We start by writing $\pi(K)$ in (\ref{eq:partc_Nakagami2}) in terms of the unknown coefficients $\alpha_r,\, r=0,\ldots, m-1$. In (\ref{eq:partc_Nakagami2}), $\lambda^r[I_1(r)+I_2(r)]$ can be replaced by $\alpha_r\!-\!(\lambda M)^r\pi(K)\e^{-\lambda(M-K)}$ and $\pi(K)$ reduces to (\ref{eq:Pi_K_Nakagami}).

Next, we derive the limiting pdf $g(x)$ in stripes of width $M$. In particular, we  derive $g_n(x)\definedas g(x),\, [K-(n+1)M]^+\leq x<K-nM$, $\forall n=0,\ldots,l'$, where $l'=l-1$ if $K$ is an integer multiple of $M$ (i.e., $K=lM$, $l\in\mathbb{Z}^+$), and $l'=l$, otherwise (i.e., $K=lM+\Delta$, with $\Delta\neq 0$). The $M$-width pdf sections $g_n(x)$ can be derived recursively by induction from $n=0$ backwards till $n=l'$  in terms of the unknown coefficients $\alpha_r$, $r=0,\ldots,m-1$. With $g_0(x)$ given in (\ref{eq:partb_Nakagami2}), we start from $n=1$ and use 
(\ref{eq:parta_Nakagami2}) to obtain $g_1(x)$, $K-2M\leq x<K-M$. Then $g_1(x)$ is used in (\ref{eq:parta_Nakagami2}) to obtain $g_2(x)$, and so on. By induction, $g_n(x)$ can be obtained as
\begin{align}
g_n(x)=\frac{\lambda^m \e^{-\lambda x}}{(m-1)!}\Biggg[\!\!&\sum\limits_{q=0}^{n}\lambda^{qm}\e^{-q\lambda M}\sum\limits_{r=0}^{m-1}\alpha_r\lambda^{-r}\binom{m\!-\!1}{r}\frac{(m\!-\!1\!-\!r)!}{((q\!+\!1)m\!-\!1\!-\!r)!}(qM\!+\!x\!-\!K)^{(q+1)m-1-r}\notag\\[-2ex]
&-\sum\limits_{q=1}^{n}\lambda^{(q-1)m}\e^{-q\lambda M}\e^{\lambda K}\pi(K)(qM+x-K)^{qm-1}\frac{(m-1)!}{(qm-1)!}\Biggg].
\label{eq:gn_Nakagami_step}
\end{align}
Substituting $\pi(K)$ from (\ref{eq:Pi_K_Nakagami}) in (\ref{eq:gn_Nakagami_step}), $g_n(x)$ reduces to (\ref{eq:pdf_Nakagami_m}).

Now, only the unknown coefficients $\alpha_r$, $r=0,\ldots,m-1$, remain to be determined. Using (\ref{eq:Pi_K_Nakagami}), $\alpha_s$ satisfies
\begin{equation}
\alpha_s\definedas \sum\limits_{r=0}^{m-1}\alpha_r d_{sr}=\lambda^s\Big(I_1(s)+I_2(s)+M^s\e^{-\lambda M}\sum\limits_{r=0}^{m-1}\frac{\alpha_r}{r!}\Big),\quad s=0,\ldots,m-1.
\label{eq:alpha_r_definition}
\end{equation}
Next, we obtain $I_1(s)=K^sI_1$ and $I_2(s)$ in terms of $\alpha_r$, $r=0,\ldots,m-1$, so that (\ref{eq:alpha_r_definition}) forms a system of linear equations in $\alpha_r$. Using $g(x)$ in (\ref{eq:pdf_Nakagami_m}) and assuming a general buffer size $K=lM+\Delta$ with $l\in\mathbb{Z}^+$, $\Delta<M$, we get $I_1=\int_0^M g(u) \dd u\!=\!\int_0^{K-lM}g_l(u)\dd u\! +\!\int_{K-lM}^M g_{l-1}(u) \dd u=I_a+I_b$, where 
\begin{equation*}
I_a=\int\limits_0^{M}g_{l-1}(u)\dd u=\sum\limits_{r=0}^{m-1}\frac{\alpha_r}{r!}\sum\limits_{q=0}^{l-1}\e^{-\lambda M q}\sum\limits_{t=qm}^{(q+1)m-r-1}\frac{\lambda^t}{t!}\Big\{(qM-K)^t-\e^{-\lambda M}((q+1)M-K)^t\Big\}
\end{equation*}
and $I_b=\int_0^{K-lM} \left(g_{l}(u)-g_{l-1}(u)\right) \dd u=\sum_{r=0}^{m-1}\frac{\alpha_r}{r!}\e^{-\lambda M l}\sum_{t=lm}^{(l+1)m-r-1}\frac{\lambda^t}{t!}(lM-K)^t.$

Define the integrand of $I_2(s)$ as $B(u,s)g(u)$, then $I_2(s)$ can be expanded as 
\begin{equation}
I_2(s)=\int\limits_{K-M}^K B(u,s) g_0(u)\dd u + \int\limits_{K-2M}^{K-M} B(u,s) g_1(u)\dd u + \ldots+ \int\limits_{M}^{K-(l-1)M} B(u,s) g_{l-1}(u)\dd u.
\label{eq:I2_expansion1}
\end{equation}
Since the sum in variable $q$ in (\ref{eq:pdf_Nakagami_m}) has $q=n$ as upper limit, if we combine the integrals of the terms with variable $q$ in  (\ref{eq:I2_expansion1}), then the $q^{\text{th}}$ term is integrated over the limits $u=M\to$ $u=K-qM$ and $I_2(s)$ reduces to
\begin{equation}
\begin{aligned}
I_2(s)=-K^s\sum\limits_{r=0}^{m-1}\frac{\alpha_r}{r!}\Biggg[
&\sum\limits_{q=0}^{l-1}\lambda^{(q+1)m-r}\e^{-\lambda M (q+1)}\sum\limits_{t=0}^s\binom{s}{t}K^{-t}t!\frac{((q+1)M-K)^{(q+1)m-r+t}}{((q+1)m-r+t)!}\\[-1ex]
&- \sum\limits_{q=1}^{l-1}\lambda^{qm}\e^{-\lambda M (q+1)} \sum\limits_{t=0}^s\binom{s}{t}K^{-t}t!\frac{((q+1)M-K)^{qm+t}}{(qm+t)!}\Biggg].
\end{aligned}
\label{eq:I2_expansion3}
\end{equation}
The last term inside the brackets in (\ref{eq:alpha_r_definition}) is the same as the term in the second line in (\ref{eq:I2_expansion3}) for $q\!=\!0$, namely $\e^{-\lambda M}\sum_{r=0}^{m-1}\frac{\alpha_r}{r!}M^s\!=\!\sum_{r=0}^{m-1}\frac{\alpha_r}{r!}\e^{-\lambda M} \sum_{t=0}^s\binom{s}{t}K^{s-t}(M-K)^t$. Hence, $d_{sr}$ in (\ref{eq:alpha_r_definition}) is
\begin{equation}
\begin{aligned}
&d_{sr}=\frac{K^s\lambda^s}{r!}\sum\limits_{q=0}^{l}\e^{-\lambda M q}\sum\limits_{t=qm}^{(q+1)m-r-1}\frac{(-\lambda)^t}{t!}\Big\{(K-qM)^t-\e^{-\lambda M}\left([K-(q+1)M]^+\right)^t\Big\}\\[-1.2ex]
&-\!\frac{\lambda^sK^s}{r!}\!\sum\limits_{q=0}^{l-1}\!\lambda^{qm}\e^{\!\!-\lambda M (q+1)}\!\sum\limits_{t=0}^s\!\binom{s}{t}K^{-t}t!\Big\{\lambda^{m-r}\frac{((q\!+\!1)M\!-\!K)^{(q\!+\!1)m\!-\!r\!+\!t}}{((q\!+\!1)m\!-\!r\!+\!t)!}-\frac{((q\!+\!1)M\!-\!K)^{qm\!+\!t}}{(qm\!+\!t)!}\Big\}.
\label{eq:d_sr}
\end{aligned}\vspace{-0.2cm}
\end{equation}
In order to obtain a unique solution for $\alpha_r$, $r=0,\ldots,m-1$, we additionally use the unit area condition $\int_0^K g(u)\dd u+\pi(K)=1$, where similar to the analysis used to expand $I_2(s)$, the integral $\int_0^K g(u)\dd u$ can be expanded as follows\vspace{-0.2cm}
\begin{equation}
\int\limits_0^K g(u)\dd u=\int\limits_{K-M}^K g_0(u)\dd u + \int\limits_{K-2M}^{K-M}  g_1(u)\dd u + \ldots+ \int\limits_{0}^{K-lM} g_{l}(u)\dd u.
\label{eq:area_under_g}
\end{equation}
From (\ref{eq:pdf_Nakagami_m}), the term with variable $q$ in (\ref{eq:area_under_g}) is integrated over the limits $u=0\to$ $u=K-qM$. Hence, the unit area condition reduces to
\begin{equation}
\int\limits_0^K g(u)\dd u+\pi(K)=\sum\limits_{r=0}^{m-1}\alpha_r\underbrace{\frac{1}{r!}\sum\limits_{q=0}^l\e^{-\lambda M q }\sum\limits_{t=qm}^{(q+1)m-r-1}\frac{(\lambda(qM-K))^t}{t!}}_{a_r}=1.
\label{eq:area_under_g_step4}\vspace{-0.2cm}
\end{equation}
Adding (\ref{eq:area_under_g_step4}) to (\ref{eq:alpha_r_definition}), we get $\alpha_s+\sum_{r=0}^{m-1}\alpha_r (a_r-d_{sr})=1, \,\,\, s=0,\ldots,m-1,$ which forms a non-homogeneous system of linear equations that can be written in matrix-vector notation as $\V{\alpha}+\V{A}\V{\alpha}=\V{1}$, where $\V{\alpha}=[\alpha_0,\ldots,\alpha_{m-1}]^T$, and $\V{A}$ is a matrix whose entry in the $s^{\text{th}}$  row and the $r^{\text{th}}$ column is $A_{sr}=a_r-d_{sr}$, where $a_r$ is given in (\ref{eq:area_under_g_step4}) and $d_{sr}$ is given in (\ref{eq:d_sr}). This completes the proof. 
\fi
\ifARXIV
\section{Proof of Theorem \ref{theo:OO_Finite_intgeral_eqn} (Integral Equation of the Energy Distribution for a Finite-Size Buffer with the On-Off Policy)}
\label{app:OO_Finite_intgeral_eqn}
The integral equations in (\ref{eq:g_integral_eqn_finite}) and (\ref{eq:partd}) can be derived by adopting the same approach used for the proof of (\ref{eq:g_integral_eqn_infinite_OO}). In particular, if we set $B(i)=u$ and $B(i+1)=x$, then (\ref{eq:general_storage_equation_OO}) reads
\begin{equation}
x=
\begin{cases}
u+X(i) & u\leq M \quad\&\quad u+X(i)<K\\
u-M+X(i) & u>M \quad\&\quad u-M+X(i)<K\\
K & {\rm otherwise.}
\end{cases}
\label{eq:Our_storage_equation_cases}
\end{equation}
Consider first the continuous part of the distribution, i.e., $g(x)$ defined on $0\leq x <K$ given in (\ref{eq:g_integral_eqn_finite}). Eqs. (\ref{eq:parta}) and (\ref{eq:partb}) are identical to (\ref{eq:Integral_eqn_infinite_parta_OO}) and (\ref{eq:Integral_eqn_infinite_partb_OO}), respectively. However, we need to further ensure that the upper limit on $u$ given by $M+x$ (for a non-negative harvested energy) is in the domain of $g(u)$. That is, in (\ref{eq:parta}), $\max_x(M+x)<K$ must hold, i.e., $K>2M$ and in (\ref{eq:partb}), $M+x<K$ must hold. Hence, (\ref{eq:partb}) is valid only for $x<K-M$ (with strict inequality). For the rest of the range of $x$ in (\ref{eq:partc}), i.e., $K-M\leq x<K$, the upper limit  $M+x$ on $u$ is larger than or equal to $K$. Thus, the whole range of $0< u\leq K$ contributes to $g(x)$. The range $0< u<K$ is covered by the first two integrals in (\ref{eq:partc}), and $u=K$ is considered in the last term. Finally, at $x=K$, the probability that the buffer is full, $\pi(K)$, in (\ref{eq:partd}) is obtained similar to (\ref{eq:partc}). However, rather than considering the pdf at the amount of harvested energy $x-u+M\mathds{1}_{u>M}$ as in (\ref{eq:partc}), we consider the ccdf $\bar{F}(x-u+M\mathds{1}_{u>M})$ instead (at $x\!=\!K$). This is because the full buffer level $K$ is attained when the amount of harvested energy is larger than or equal to $K-u+M\mathds{1}_{u>M}$, $\forall\, u$ in $0\!<\!u\!\leq\!K$. This completes the proof.

\fi
\ifARXIV
\section{Proof of Corollary \ref{coro:limiting_dist_Finite_Nakagami_exact_OO}  (Energy Distribution for a Finite-Size Buffer with the On-Off Policy and Gamma-Distributed EH Process)}
\label{app:limiting_dist_Finite_Nakagami_exact_OO_arxiv}
We first derive $g(x)$ for $M<x<K$ in (\ref{eq:g2_OO}) and $\pi(K)$
in (\ref{eq:Pi_K_Nakagami_OO}). Notice that the integral equations of the on-off policy in (\ref{eq:partb}), (\ref{eq:partc}), and (\ref{eq:partd}) differ from those of the best-effort policy in (\ref{eq:parta_BE}), (\ref{eq:partb_BE}), and (\ref{eq:partc_BE}), respectively, only in the first term. Hence, when substituting the pdf and the ccdf of the Gamma distributed EH process in (\ref{eq:partb})-(\ref{eq:partd}), we get (\ref{eq:parta_Nakagami})-(\ref{eq:partc_Nakagami}) with the difference that the first term in the brackets in  (\ref{eq:parta_Nakagami}) and (\ref{eq:partb_Nakagami}) has to be replaced by $\int_{u=0}^M(x-u)^{m-1}\e^{\lambda u}g(u)\dd u$, and the first term in the brackets in (\ref{eq:partc_Nakagami}) has to be replaced by $\sum_{r=0}^{m-1}\frac{\lambda^r}{r!}\int_{u=0}^{M}\e^{\lambda u} (K-u)^r g(u) \dd u$. Due to this similarity, it follows that the pdf $g(x)$ for the on-off policy in the ranges $M<x<K-M$, $K-M<x<K$, and $\pi(K)$ are identical to (\ref{eq:parta_Nakagami2}), (\ref{eq:partb_Nakagami2}), and (\ref{eq:partc_Nakagami2}), respectively. The sole difference is in the definition of $I_1(r)$, namely $I_1(r)=\int_{u=0}^M(K-u)^r\e^{\lambda u}g(u)\dd u$, which would result in a different solution for $\alpha_r$, cf. (\ref{eq:alpha_r_definition}).
Now, since the derivations from (\ref{eq:g_integral_eqn_finite_Nakagami2}) to (\ref{eq:gn_Nakagami_step}) are independent of the value of $\alpha_r$, it follows that $g(x)$ for $M<x<K$ and the atom $\pi(K)$ are identical to those of the best-effort policy. Hence, (\ref{eq:g2_OO}) and (\ref{eq:Pi_K_Nakagami_OO}) are identical to (\ref{eq:pdf_Nakagami_m}) and (\ref{eq:Pi_K_Nakagami}), respectively.
Now, we only have to determine $g_{l-1}(x)\definedas g(x)$ in the range $0<x<M$ and a non-homogeneous system of linear equations for the calculation of $\alpha_r$, $r=0,\ldots,m-1$.  
\subsection{Expressing $g_{l-1}(x)$ in Terms of $\alpha_r$}
Substituting the pdf of the Gamma distributed EH process $f(x)\!=\!\frac{\lambda^m}{\Gamma(m)}x^{m-1}\e^{-\lambda x}$ in (\ref{eq:parta}), we get
\begin{equation}
g_{l-1}(x)=\frac{\lambda^m\e^{-\lambda x}}{(m-1)!}\int\limits_{u=0}^{x}(x-u)^{m-1}\e^{\lambda u}g_{l-1}(u)\dd u+\underbrace{\frac{\lambda^m\e^{-\lambda x}}{(m-1)!}\int\limits_{u=M}^{M+x}(x-u+M)^{m-1}\e^{-\lambda(M-u)}g_{l-2}(u)\dd u}_{T(x)},
\label{eq:g_l_1_OO_step1}
\end{equation}
where $g_{l-2}(x)$ is $g(x)$ in the range $M<x<2M$, which is given by (\ref{eq:g2_OO}) for $n=l-2$. Eq. (\ref{eq:g_l_1_OO_step1}) is a Volterra integral equation of the second kind, whose solution is given by \cite[eq. 2.2-2.31]{polyanin2008handbook}
\begin{equation}
g_{l-1}(x)=T(x)+\underbrace{\int\limits_{u=0}^{x}R(x-u)T(u)\dd u}_{I_R(x)},\vspace{-0.5cm}
\label{eq:g_l_1_OO_step2}
\end{equation}
where 
\begin{equation}
R(x)=\frac{\e^{-\lambda x}}{m}\sum\limits_{k=0}^{m-1}\lambda\e^{\lambda x\cos(\eta_k)}\cos\Big(\eta_k+\lambda x\sin(\eta_k)\Big),
\label{eq:R_x}
\end{equation}
with $\eta_k=\frac{2\pi k}{m}$. Next, we calculate $T(x)$ defined in (\ref{eq:g_l_1_OO_step1}) by substituting $g_{l-2}(u)$ from (\ref{eq:g2_OO}) for $n=l-2$ in $T(x)$ and solving integrals of the form $\int_{u=M}^{M+x}(x+M-u)^{m-1}(qM-K+u)^d \dd u$ once for $d=(q+1)m-r-1$ and once for $d=qm-1$.  This integral can be solved after replacing $(qM-K+u)^d$ by $\sum_{t=0}^d \binom{d}{t}(u-M)^t((q+1)M-K)^{d-t}$ and using $\int_a^b(b-u)^c(u-a)^t\dd u=\frac{c!t!(b-a)^{c+t+1}}{(c+t+1)!}$ with $a<b$ at $a=M$, $b=M+x$, and $c=m-1$. Then, $T(x)$ reduces to
\begin{equation}
\begin{aligned}
T(x)&=\e^{-\lambda x}\sum\limits_{r=0}^{m-1}\frac{\alpha_r}{r!}\Biggg[\sum\limits_{q=0}^{l-2}\lambda^{(q+2)m-r}\e^{-\lambda M(q+1)}\sum\limits_{t=0}^{(q+1)m-r-1}\frac{((q+1)M-K)^{(q+1)m-r-1-t}}{((q+1)m-r-1-t)!(m+t)!}x^{m+t}\\
&-\sum\limits_{q=1}^{l-2}\lambda^{(q+1)m}\e^{-\lambda M(q+1)}\sum\limits_{t=0}^{qm-1}\frac{((q+1)M-K)^{qm-1-t}}{(qm-1-t)!(m+t)!}x^{m+t}\Biggg].\\[-3ex]
\end{aligned}
\label{eq:T_x}
\end{equation}
Next, we determine $I_R(x)$ defined in (\ref{eq:g_l_1_OO_step2}). Using $R(x)$ in (\ref{eq:R_x}) and and $T(x)$ in (\ref{eq:T_x}), $I_R(x)$ is
\begin{equation}
\begin{aligned}
I_R(x)&=\e^{-\lambda x}\sum\limits_{r=0}^{m-1}\frac{\alpha_r}{r!}\Biggg[\sum\limits_{q=0}^{l-2}\lambda^{(q+2)m-r}\e^{-\lambda M(q+1)}\sum\limits_{t=0}^{(q+1)m-r-1}\frac{((q+1)M-K)^{(q+1)m-r-1-t}}{((q+1)m-r-1-t)!(m+t)!} I_t(x)\\
&-\sum\limits_{q=1}^{l-2}\lambda^{(q+1)m}\e^{-\lambda M(q+1)}\sum\limits_{t=0}^{qm-1}\frac{((q+1)M-K)^{qm-1-t}}{(qm-1-t)!(m+t)!}I_t(x)\Biggg],
\end{aligned}
\label{eq:IR_x}
\end{equation}
where $I_t(x)$ is given by
\begin{equation}
I_t(x)=\frac{1}{m}\sum\limits_{k=0}^{m-1}\lambda\e^{\lambda x\cos(\eta_k)}\Biggg[\cos(\mu_{k,x})I_{\cos}(x)
+\sin(\mu_{k,x})I_{\sin}(x)\Biggg],
\label{eq:It_x_step1}
\end{equation}
with $\mu_{k,x}=\eta_k+\lambda x\sin(\eta_k)$, $I_{\sin}(x)=\int_{u=0}^{x}u^{m+t}\e^{-\lambda\cos(\eta_k)u}\sin\Big(\lambda\sin(\eta_k)u\Big)\dd u$ and \\$I_{\cos}(x)=\int_{u=0}^{x}u^{m+t}\e^{-\lambda\cos(\eta_k)u}\cos\Big(\lambda\sin(\eta_k)u\Big)\dd u$. 
$I_{\cos}(x)$ and $I_{\sin}(x)$ are solved using \cite[eq. 3.944.3]{table_of_integrals_Ryzhik} and \cite[eq. 3.944.1]{table_of_integrals_Ryzhik}, respectiely as $I_{\cos}(x)=\Re\{I_{\text{complex}}\}$ and $I_{\sin}(x)=-\Im\{I_{\text{complex}}\}$, where $I_{\text{complex}}=\Big(\lambda \e^{\J\eta_k}\Big)^{-(m+t+1)}\gamma\Big(m+t+1,\lambda x\e^{\J\eta_k}\Big)$. Hence, $I_t(x)$ reduces to
\begin{equation}
I_t(x)=\frac{1}{m}\sum\limits_{k=0}^{m-1}\lambda\e^{\lambda x\cos(\eta_k)}\Re\{\e^{\J\mu_{k,x}}I_{\text{complex}}\}
=\frac{1}{m}\sum\limits_{k=0}^{m-1}\Re\Big\{(\lambda\e^{\J\eta_k})^{-(m+t)}\e^{\lambda x \e^{\J\eta_k}}\gamma\Big(m+t+1,\lambda x \e^{\J\eta_k}\Big)\Big\}.
\label{eq:It_x}
\end{equation}
Substituting $I_{t}(x)$ from (\ref{eq:It_x}) in (\ref{eq:IR_x}) and adding $T(x)$ in (\ref{eq:T_x}) to $I_R(x)$ in (\ref{eq:IR_x}), then $g_{l-1}(x)$ in (\ref{eq:g_l_1_OO_step2}) reduces to (\ref{eq:g1_OO}).\\

\subsection{Defining a Non-Homogeneous System of Linear Equations in $\alpha_s$}
Now, only the unknown coefficients $\alpha_s$, $s=0,\ldots,m-1$, remain to be determined. Similar to the best-effort policy, $\alpha_s$ satisfies (\ref{eq:alpha_r_definition}) where only $I_1(s)$ has a different definition. The remaining terms in $\alpha_s$ can be obtained using the analysis in (\ref{eq:I2_expansion1})-(\ref{eq:I2_expansion3}) with the only difference that in the on-off policy we assume $K=lM$ with integer $l$. Hence in (\ref{eq:I2_expansion1}), the last integral is on $g_{l-2}(x)$ over the limits $M$ to $K-(l-2)M$. Therefore, $I_2(s)$ satisfies (\ref{eq:I2_expansion3}) with the inner summation limits up to $q=l-2$ instead of $q=l-1$. With these considerations, the last two terms inside the brackets of (\ref{eq:alpha_r_definition}) are given by $I_2(s)+M^s\e^{-\lambda M}\sum\limits_{r=0}^{m-1}\frac{\alpha_r}{r!}=$
\begin{equation}
-\!\!\sum\limits_{r=0}^{m-1}\!\!\frac{K^s\alpha_r }{r!}\!\sum\limits_{q=0}^{l-2}\!\!\lambda^{qm}\e^{-\lambda M (q+1)}\sum\limits_{t=0}^s\!\!\binom{s}{t}\!K^{-t}t!\Big\{\lambda^{m-r}\frac{((q\!+\!1)M\!-\!K)^{(q+1)m-r+t}}{((q+1)m-r+t)!}- \frac{((q\!+\!1)M\!-\!K)^{qm+t}}{(qm+t)!}\Big\}.
\label{eq:alpha_last2terms}
\end{equation}
Next, we obtain $I_1(s)$ defined at the beginning of this appendix as $I_1(s)=\int_{u=0}^{M}(K-u)^s\e^{\lambda u}g(u)\dd u$, where in the range $0<u<M$, $g(u)=g_{l-1}(u)$ given in (\ref{eq:g1_OO}). Substituting $g_{l-1}(u)$ in $I_1(s)$, we get
\begin{equation}
\begin{aligned}
I_1(s)=&\sum\limits_{r=0}^{m-1}\frac{\alpha_r}{r!}\Biggg[\sum\limits_{q=0}^{l-2}\lambda^{(q+2)m-r}\e^{-\lambda M(q+1)}\sum\limits_{t=0}^{(q+1)m-r-1}\frac{((q+1)M-K)^{(q+1)m-r-1-t}}{((q+1)m-r-1-t)!(m+t)!}F'(t,s)\\
&-\sum\limits_{q=1}^{l-2}\lambda^{(q+1)m}\e^{-\lambda M(q+1)}\sum\limits_{t=0}^{qm-1}\frac{((q+1)M-K)^{qm-1-t}}{(qm-1-t)!(m+t)!}F'(t,s)\Biggg],\end{aligned}
\label{eq:I_1_s_step1}\vspace{-0.4cm}
\end{equation}
where 
\begin{equation}
F'(t,s)=\int\limits_{u=0}^{M}(K-u)^s C(u,t)\dd u=\sum\limits_{b=0}^{s}\binom{s}{b}(K-M)^{s-b}\underbrace{\int\limits_{u=0}^M(M-u)^b C(u,t)\dd u}_{I_C(t,b)}.
\label{eq:F_dash_t_s}
\end{equation}
To solve $I_C(t,b)$, we first expand the Gamma function in (\ref{eq:C_x_t}) and replace $\e^{\lambda\e^{\J\eta_k} u}\gamma(m+t+1,\lambda\e^{\J\eta_k}u)$ by $(m+t)!\Big[\e^{\lambda\e^{\J\eta_k}u}-\sum_{w=0}^{m+t}\frac{(\lambda\e^{\J\eta_k} u)^w}{w!}\Big]$. Then, we solve integrals of the form $\int_{u=0}^M(M-u)^b u^n \dd u=\frac{M^{n+b+1}b!n!}{(n+b+1)!}$ for $n=m+t$ and $n=w$, and use $\int_{u=0}^{M}(M-u)^b\e^{\lambda\e^{\J\eta_k} u}\dd u=\e^{\lambda M\e^{\J\eta_k}}(\lambda\e^{\J\eta_k})^{-b-1}\gamma(b+1,\lambda M\e^{\J\eta_k})$. Hence, $F(t,s)\definedas F'(t,s)/(m+t)!$ reduces to (\ref{eq:F_t_s}).

Using (\ref{eq:alpha_last2terms}) and $I_1(s)$ in (\ref{eq:I_1_s_step1}), then $d_{sr}$  in (\ref{eq:alpha_r_definition}) reduces to
\begin{equation}
\begin{aligned}
d_{sr}\!&=\!-\frac{\lambda^sK^s }{r!}\!\!\sum\limits_{q=0}^{l-2}\!\!\lambda^{qm}\e^{-\lambda M (q+1)}\!\!\sum\limits_{t=0}^s\!\!\binom{s}{t}\!K^{-t}t!\Big\{\!\lambda^{m-r}\frac{((q\!+\!1)M\!-\!K)^{(q+1)m-r+t}}{((q+1)m-r+t)!}\!-\! \frac{((q\!+\!1)M\!-\!K)^{qm+t}}{(qm+t)!}\!\Big\}\\
&+\frac{\lambda^s}{r!}\Biggg[\sum\limits_{q=0}^{l-2}\lambda^{(q+2)m-r}\e^{-\lambda M(q+1)}\sum\limits_{t=0}^{(q+1)m-r-1}\frac{((q+1)M-K)^{(q+1)m-r-1-t}}{((q+1)m-r-1-t)!}F(t,s)\\
&-\sum\limits_{q=1}^{l-2}\lambda^{(q+1)m}\e^{-\lambda M(q+1)}\sum\limits_{t=0}^{qm-1}\frac{((q+1)M-K)^{qm-1-t}}{(qm-1-t)!}F(t,s)\Biggg].
\end{aligned}
\label{eq:dsr_appendix}
\end{equation}
In order to obtain the unique solution for $\alpha_r$, $r=0,\ldots,m-1$, we additionally use the unit area condition $\int_{u=0}^K g(u)\dd u + \pi(K)=1$, which can be written as
\begin{equation}
\underbrace{\int\limits_{u=0}^M g_{l-1}(u)\dd u}_{I_{\text{U1}}} +\underbrace{\int\limits_{u=M}^K g(u)\dd u + \pi(K)}_{I_{\text{U2}}}\definedas \sum\limits_{r=0}^{m-1}\alpha_r a_{r}=1.
\label{eq:unit_area_OO}
\end{equation}
$I_{\text{U2}}$ can be obtained in a similar manner as  (\ref{eq:area_under_g_step4}). Specifically, since the pdf $g(x)$ for the on-off policy in the range $M<x<K$ in (\ref{eq:g2_OO}) is identical to the pdf of the best-effort policy in (\ref{eq:pdf_Nakagami_m}). The integral $\int_{u=M}^K g(u)\dd u$ is identical to that in (\ref{eq:area_under_g}) with the last term replaced by $\int_{u=M}^{K-(l-2)M}g_{l-2}(u)\dd u$. Consequently,  $\int_{u=M}^K g(u)\dd u$ is the  same as (\ref{eq:area_under_g_step2}) after replacing the upper limit of the inner summations by $q=l-2$ and the inner lower integral limit by $M$. Thus, the inner integrals have the form $\int_{u=M}^{b}\e^{-\lambda u}(u-b)^c/c!\dd u=-\lambda^{-c-1}\Big[\e^{-\lambda b}-\e^{-\lambda M}\sum_{t=0}^{c}(\lambda(M-b))^t/t!\Big]$, at $b=K-qM$ and at $c=(q+1)m-r-1$ in the first term and $c=qm-1$ in the second term. These considerations lead to 
\begin{equation}
\begin{aligned}
\int\limits_{u=M}^K g(u)\dd u&=\sum\limits_{r=0}^{m-1}\frac{\alpha_r}{r!}\Biggg[\sum\limits_{q=0}^{l-2}\Bigg( \e^{-\lambda M(q+1)}\sum\limits_{t=0}^{(q+1)m-r-1}\frac{\Big(\lambda((q+1)M-K)\Big)^t}{t!}-\e^{-\lambda K}\Bigg)\\
 &+\sum\limits_{q=1}^{l-2}\Bigg( \e^{-\lambda K}-\e^{-\lambda M(q+1)}\sum\limits_{t=0}^{qm-1}\frac{\Big(\lambda((q+1)M-K)\Big)^t}{t!}\Bigg)\Biggg].
\end{aligned}
\label{eq:area_g_K_M}
\end{equation}
Separating $q=0$ in the first term in (\ref{eq:area_g_K_M}), we get \small
\begin{equation}
\int\limits_{u=M}^K g(u)\dd u=\sum\limits_{r=0}^{m-1}\frac{\alpha_r}{r!}\Biggg[ \e^{-\lambda M}\sum\limits_{t=0}^{m-r-1}\frac{\Big(\lambda(M-K)\Big)^t}{t!}-\e^{-\lambda
 K}
 +\sum\limits_{q=1}^{l-2} \e^{-\lambda
 M(q+1)}\sum\limits_{t=qm}^{(q+1)m-r-1}\frac{\Big(\lambda((q+1)M-K)\Big)^t}{t!}\Biggg].
\label{eq:area_g_K_M2}
\end{equation}
\normalsize
From (\ref{eq:Pi_K_Nakagami_OO}), the first term in (\ref{eq:area_g_K_M2}) is $-\pi(K)$. The second term in (\ref{eq:area_g_K_M2}) is equivalent to the summand in the third term at $q=0$. Hence, we can write
\begin{equation}
I_{\text{U2}}=\int\limits_{u=M}^K g(u)\dd u + \pi(K)=\sum\limits_{r=0}^{m-1}\frac{\alpha_r}{r!}\sum\limits_{q=0}^{l-2}\e^{-\lambda M(q+1)}\sum\limits_{t=qm}^{(q+1)m-r-1}\frac{\lambda^t((q+1)M-K)^{t}}{t!}.
\label{eq:area_g_K_M3}   
\end{equation}
Next, we solve $I_{\text{U1}}=\int_{u=0}^M g_{l-1}(u)\dd u$. Using $g_{l-1}(u)$ in (\ref{eq:g1_OO}), $I_{\text{U1}}$ is identical to $I_1(s)$ in (\ref{eq:I_1_s_step1}) after replacing $F'(t,s)$ by $H'(t)\definedas\int_{u=0}^{M}\e^{-\lambda u} C(u,t)\dd u$. Similar to the analysis used to obtain $I_C(t,b)$ in (\ref{eq:F_dash_t_s}), we first expand the Gamma function in (\ref{eq:C_x_t}) and replace $\e^{\lambda\e^{\J\eta_k} u}\gamma(m+t+1,\lambda\e^{\J\eta_k}u)$ by $(m+t)!\Big[\e^{\lambda\e^{\J\eta_k}u}-\sum_{w=0}^{m+t}\frac{(\lambda\e^{\J\eta_k} u)^w}{w!}\Big]$. We then solve integrals of the form $\int_{u=0}^{M}\e^{-\lambda u} u^n=\lambda^{-n-1}\gamma(n+1,\lambda M)$ for $n=m+t$ and $n=w$. Moreover, we use $\int_{u=0}^M \e^{-\lambda u (1-\e^{\J\eta_k})}\dd u=[1-\e^{-\lambda M (1-\e^{\J\eta_k})}]/[\lambda(1-\e^{\J\eta_k})]$ for $k=0$ and $\int_{u=0}^M \e^{-\lambda u (1-\e^{\J\eta_k})}\dd u=M$ for $k\neq 0$. With these considerations, $H(t)\definedas H'(t)/(m+t)!$ reduces to (\ref{eq:H_T}). Having obtained $I_{\text{U1}}$ and $I_{\text{U2}}$, $a_r$ in (\ref{eq:unit_area_OO}) can be concluded. Adding (\ref{eq:unit_area_OO}) to (\ref{eq:alpha_r_definition}), we get $\alpha_s\!+\!\sum\limits_{r=0}^{m-1}\alpha_r(a_r\!-\!d_{sr})\!=\!1$, which forms a non-homogeneous system of linear equations that can be written in matrix-vector notation as $\V{\alpha}+\V{A}\V{\alpha}=\V{1}$, where $\V{\alpha}=[\alpha_0,\ldots,\alpha_{m-1}]^T$, and $\V{A}$ is an $m\times m$  matrix whose entry in the $s^{\text{th}}$  row and the $r^{\text{th}}$ column is $A_{sr}\!=\!a_r\!-\!d_{sr}$. This completes the proof. 

\else
\vspace{-0.4cm}\section{Proof of Corollary \ref{coro:limiting_dist_Finite_Nakagami_exact_OO} (On-Off Policy with Finite-Size Buffer)}
\label{app:limiting_dist_Finite_Nakagami_exact_OO_short}
We first derive $g(x)$ for $M<x<K$ in (\ref{eq:g2_OO}) and $\pi(K)$
in (\ref{eq:Pi_K_Nakagami_OO}). Notice that the integral equations of the on-off policy in (\ref{eq:partb}), (\ref{eq:partc}), and (\ref{eq:partd}) differ from those of the best-effort policy in (\ref{eq:parta_BE}), (\ref{eq:partb_BE}), and (\ref{eq:partc_BE}), respectively, only in the first term. Hence, when substituting the pdf and the ccdf of the Gamma distributed EH process in (\ref{eq:partb})-(\ref{eq:partd}), we get (\ref{eq:parta_Nakagami})-(\ref{eq:partc_Nakagami}) with the difference that the first term in the brackets in  (\ref{eq:parta_Nakagami}) and (\ref{eq:partb_Nakagami}) has to be replaced by $\int_{u=0}^M(x-u)^{m-1}\e^{\lambda u}g(u)\dd u$, and the first term in the brackets in (\ref{eq:partc_Nakagami}) has to be replaced by $\sum_{r=0}^{m-1}\frac{\lambda^r}{r!}\int_{u=0}^{M}\e^{\lambda u} (K-u)^r g(u) \dd u$. Due to this similarity, it follows that the pdf $g(x)$ for the on-off policy in the ranges $M<x<K-M$, $K-M<x<K$, and $\pi(K)$ are identical to (\ref{eq:parta_Nakagami2}), (\ref{eq:partb_Nakagami2}), and (\ref{eq:partc_Nakagami2}), respectively. The sole difference is in the definition of $I_1(r)$, namely $I_1(r)=\int_{u=0}^M(k-u)^r\e^{\lambda u}g(u)\dd u$, which results in a different solution for $\alpha_r$,  cf. (\ref{eq:alpha_r_definition}).
Now, since the derivations from (\ref{eq:g_integral_eqn_finite_Nakagami2}) to (\ref{eq:gn_Nakagami_step}) are independent of the value of $\alpha_r$, it follows that $g(x)$ for $M<x<K$ and the atom $\pi(K)$ are identical to those of the best-effort policy. Hence, (\ref{eq:g2_OO}) and (\ref{eq:Pi_K_Nakagami_OO}) are identical to (\ref{eq:pdf_Nakagami_m}) and (\ref{eq:Pi_K_Nakagami}), respectively.

Now, we only have to determine $g_{l-1}(x)\definedas g(x)$ in the range $0<x<M$ and $\alpha_r$, $r=0,\ldots,m-1$. Function $g_{l-1}(x)$ can be obtained from (\ref{eq:parta}) by solving a Volterra integral equation of the second kind, similar to the analysis used to obtain $g_1(x)$ for an infinite-size buffer with the on-off policy, cf. Appendix \ref{app:stationary_dist_infinite_Nakagami_OO_shortened}. The set of linear equations for $\alpha_r$ can be obtained in a similar manner as for the best-effort policy for a finite-size energy buffer, cf. Appendix \ref{app:limiting_dist_Finite_Nakagami_exact_BE}. In particular, the unit area condition is added to (\ref{eq:alpha_r_definition}) to form a non-homogeneous system of equations in $\alpha_r$. With the different definition of $I_1(s)$, a different  $g(x)$ for $0<x<M$, and the assumption $K=lM$ for integer $l$, the system of equations for $\alpha_r$ will differ accordingly. For a detailed proof, see \cite[Appendix L]{Morsi_storage_arxiv}.

\fi

\vspace{-0.4cm}\bibliography{abbreviations,references}
\bibliographystyle{IEEEtran}

\end{document}